\documentclass[11pt,english]{article}
\usepackage{lmodern}
\usepackage[T1]{fontenc}
\usepackage[latin9]{inputenc}
\usepackage[a4paper]{geometry}
\usepackage{amsmath}
\usepackage{amsthm}
\usepackage{amssymb}
\usepackage{graphicx}
\usepackage{setspace}
\usepackage[authoryear]{natbib}
\makeatletter
\numberwithin{equation}{section}
\theoremstyle{definition}
 \newtheorem{example}{\protect\examplename}
\theoremstyle{plain}
\newtheorem{assumption}{\protect\assumptionname}
\theoremstyle{plain}
\newtheorem{thm}{\protect\theoremname}

\usepackage{placeins}
\usepackage[titletoc]{appendix}
\usepackage{graphicx}
\usepackage{placeins}
\usepackage{chngcntr}

\usepackage{threeparttable}
\usepackage{booktabs,caption} 
\usepackage{subfig}
\usepackage{pgfplots}
\usepackage[hang,splitrule]{footmisc}
\usepackage[pagebackref=false]{hyperref}
\hypersetup{colorlinks=true,citecolor=blue}
\hypersetup{colorlinks=true,linkcolor=blue}
\hypersetup{colorlinks=true,urlcolor=blue}

\makeatother

\usepackage{babel}
\providecommand{\assumptionname}{Assumption}
\providecommand{\examplename}{Example}
\providecommand{\theoremname}{Theorem}

\begin{document}

\title{\textbf{Solving Dynamic Discrete Choice Models Using Smoothing and
Sieve Methods}\thanks{We would like to thank Mike Keane, John Rust, Victor Aguirregabiria,
Lars Nesheim, Aureo de Paula and many other people for helpful comments
and suggestions. Kristensen gratefully acknowledges financial support
from the ERC (through starting grant No 312474 and advanced grant
No GEM 740369). Schjerning gratefully acknowledges the financial support
from the Independent Research Fund Denmark (grant no. DFF \textendash{}
4182-00052) and the URBAN research project financed by the Innovation
Fund Denmark (IFD).}}

\author{Dennis Kristensen\thanks{Department of Economics, University College London, Gower Street,
London, United Kingdom. E-mail: \url{d.kristensen@ucl.ac.uk}. Website:
\url{https://sites.google.com/site/econkristensen}.} \and Patrick K. Mogensen\thanks{Department of Economics, University of Copenhagen, Øster Farimagsgade
5, Building 35, DK-1353 Copenhagen K, Denmark. E-mail: \url{Patrick.Kofod.Mogensen@econ.ku.dk}.
Webpage: \url{http://www.economics.ku.dk/staff/phd_kopi/?pure=en/persons/374766}.}\and Jong Myun Moon\thanks{PIMCO}\and Bertel Schjerning\thanks{Department of Economics, University of Copenhagen, Øster Farimagsgade
5, Building 35, DK-1353 Copenhagen K, Denmark. E-mail: \url{Bertel.Schjerning@econ.ku.dk}.
Webpage: \url{http://bschjerning.com/}.}\\
 }
\maketitle
\begin{abstract}
\noindent We propose to combine smoothing, simulations and sieve approximations
to solve for either the integrated or expected value function in a
general class of dynamic discrete choice (DDC) models. We use importance
sampling to approximate the Bellman operators defining the two functions.
The random Bellman operators, and therefore also the corresponding
solutions, are generally non-smooth which is undesirable. To circumvent
this issue, we introduce smoothed versions of the random Bellman operators
and solve for the corresponding smoothed value functions using sieve
methods. We also show that one can avoid using sieves by generalizing
and adapting the ``self-approximating'' method of \citet{Rust1997curse}
to our setting. We provide an asymptotic theory for both approximate
solution methods and show that they converge with $\sqrt{N}$-rate,
where $N$ is number of Monte Carlo draws, towards Gaussian processes.
We examine their performance in practice through a set of numerical
experiments and find that both methods perform well with the sieve
method being particularly attractive in terms of computational speed
and accuracy.
\end{abstract}

\thispagestyle{empty}

\newpage{}

\pagenumbering{arabic}
\setcounter{page}{1}

\section{Introduction\label{sec:Introduction}}

Discrete Decision Processes (DDPs) are widely used in economics to
model forward-looking discrete decisions. For their implementation,
researchers are required to solve the model which generally cannot
be done in closed form. Instead, a number of methods have been proposed
for solving the model numerically; see, e.g., \citet{rust2008dynamic}
for an overview. We propose two novel methods for approximating the
solutions to a general class of Markovian DDP models in terms of either
the so-called integrated or expected value function. These two functions
are relevant for estimation of DDP's and for welfare analysis of policy
experiments. Our framework allows for both continuous and discrete
state variables, non-separable utility functions and unrestricted
dynamics. As such, we cover most relevant models used in empirical
work. The proposed implementation of model and estimators are found
to be computationally very efficient, and at the same time providing
precise results with small approximation errors due to the use of
simulations and sieve methods. 

Our first proposal proceeds in three steps: First, we develop smoothed
simulated versions of the Bellman operators that returns the integrated
and expected value functions as fixed points. Next, we approximate
the unknown value function by a sieve, that is, a parametric function
class, thereby turning the problem into a finite-dimensional one.
Finally, we solve for the parameters entering the chosen sieve using
projection-based methods. When the chosen sieve is linear in the parameters,
the approximate solution can be computed using an iterative procedure
where each step is on closed form.

As an alternative to the above sieve-based method, we also adapt and
generalize the so-called ``self-approximating'' method proposed
in \citet{Rust1997curse} to our setting: We design the importance
sampler used in the simulated Bellman operators so that the corresponding
expected and integrated value functions can be solved for directly
without the use of sieves. In comparison with the sieve approach,
the self-approximating solution method has the advantage that it will
not suffer from any biases due to function approximations. But at
the same time, the importance sampler used in its implementation will
generally have a larger variance compared to the class of samplers
that can be used for the sieve method. This larger variance also translates
into a larger simulation bias of the self-approximating solution due
to the non-linear nature of the problem. Thus, neither method strictly
dominates the other.

Our two procedures, the sieve-based and self-approximating one, differ
from existing proposals in three important aspects: First, we solve
for either the integrated or expected value function instead of the
value function itself. This reduces the dimensionality of the problem
since we integrate out any i.i.d. shocks appearing in the model before
solving it. Moreover, while the value function is non-differentiable,
the integrated and expected value functions are generally smooth which
means that our sieve method performs better compared to existing ones
that aim at approximating the value function. Second, we allow for
a general class of importance samplers in the simulation of the Bellman
operator; these can be designed to reduce variances and biases due
to simulations. Third, we smooth the simulated Bellman operator by
replacing the $\max$-function appearing in its expression by a smoothed
version where the degree of smoothing is controlled by a parameter
akin to the bandwidth in kernel smoothing methods. This is similar
to the logit-smoothed accept-reject simulator of probit models as
proposed by \citet{McFadden1989}; see also \citet{Fermanian&Salanie2004},
\citet{Kristensen&Shin20112} and \citet{Iskhakovetal2017}. The smoothing
turns the problem of solving for the integrated and expected value
functions into differentiable ones. In particular, the exact solutions
to the smoothed simulated Bellman equations become smooth as functions
of state variables and any underlying structural parameters. This
in turn means that standard sieves, such as polynomials, will approximate
the exact solutions well and that we can control the error rate due
to function approximation. Moreover, if used in estimation, standard
numerical solvers can be employed in computing estimators of the structural
parameters. The smoothing entails an additional bias but this can
be controlled for by suitable choice of aforementioned smoothing parameter.

The smoothing device also facilitates the theoretical analysis of
the approximate value functions since it allows us to use a functional
Taylor expansion of it. This expansion is then used to analyze the
leading numerical error terms of the approximate value functions due
to simulations, smoothing and function approximations. In particular,
under regularity conditions, we show that the approximate value function
will converge weakly towards a Gaussian process which is the first
result of its kind to our knowledge. These results allow researchers
to, for example, build confidence intervals around the approximate
value function and should be useful when analyzing the impact of value
function approximation when used in welfare analysis and estimation
of structural parameters. They may also be potentially helpful in
designing selection rules for number of basis functions and the smoothing
parameter.

A numerical study investigates the performance of the solution methods
in practice. We implement the proposed methods for the engine replacement
model of \citet{Rust1987} and investigate how smoothing, number of
basis functions and number of simulations affect the approximation
errors. We also investigate how the procedures are affected by the
dimensionality of the problem and how derivative-based solvers affect
computation times. We find that the sieve method generally performs
best of the two methods: It is computationally faster and in most
situations provides a better approximation in terms of bias and variance.
Moreover, the sieve method is found to also work well in higher dimensions
with its bias and variance being fairly stable as we increase the
the number of state variables of the model. In contrast, variances
of the self-approximating method increase dramatically as the number
of state variables increases and so appears to be less robust. Finally,
the errors due to simulations and function approximation behave according
to theory and are found to vanish at the expected rates.

Our proposed methods share similarities with the ones developed in,
amongst others, \citet{Arcidiaconoetal2013}, \citet{Keane&Wolpin1994},
\citet{Munos&Szepesvari2008}, \citet{Norets2012}, \citet{Pal&Stacurski2013}
and \citet{Rust1997curse} who also use simulations and/or sieve methods
to solve DDP's. However, except for \citet{Keane&Wolpin1994}, the
methods proposed in these papers approximate the value function while
ours target the integrated or expected value function which are more
well-behaved (smooth) objects and therefore easier to approximate.
Moreover, in contrast to the cited papers, we employ importance sampling
and smoothing in our implementation which comes with the aforementioned
computational advantages. From a theory perspective, we provide a
more complete asymptotic analysis of the approximate integrated and
expected value functions. On the other hand, \citet{Munos&Szepesvari2008}
and \citet{Rust1997curse} provide an analysis of the computational
complexity of solving for the value function and so the theories of
this paper and these studies complement each other.

The remains of the paper are organized as follows: Section \ref{sec:Model}
introduces a general class of DDP's and their corresponding value
functions. In Section \ref{sec:Simulated-Bellman}, we develop our
smoothed simulated versions of the Bellman operators that the integrated
and expected value functions are fixed points to. We then show how
to (approximately) solve these simulated Bellman equations in Section
\ref{sec:Approximate-value}. An asymptotic theory of the approximate
value function is presented in Section \ref{sec:Convergence}, while
the results of the numerical experiments are found in Section \ref{sec:Numerical}.
Appendix \ref{app:Auxiliary-Results} contains some general results
for approximate solutions to fixed point problems, while proofs of
the main results can be found in Appendix \ref{App:Proofs}.

\section{Model\label{sec:Model}}

We consider the following DDP where a single agent at time $t\geq1$
solves
\begin{equation}
d_{t}=\arg\max_{d\in\mathcal{D}}\left\{ u(S_{t},d)+\beta E\left[\nu(S_{t+1})|S_{t},d_{t}=d\right]\right\} ,\label{eq: model}
\end{equation}
where $\mathcal{D}=\left\{ 1,...,D\right\} $ is the set of alternatives,
$u(S_{t},d)$ is the per-period utility, $0<\beta<1$ is the discount
factor, $S_{t}$ is a set of state variables that follows a controlled
Markov process with transition kernel $F_{S}\left(S_{t}|S_{t-1},d_{t-1}\right)$
and the so-called value function $\nu$ solves the following fixed-point
problem, 
\begin{equation}
\nu(S_{t})=\max_{d\in\mathcal{D}}\left\{ u(S_{t},d)+\beta E\left[\nu(S_{t+1})|S_{t},d_{t}=d\right]\right\} .\label{eq:org bellman}
\end{equation}
Following \citet{Rust1987} and many subsequent empirical specifications,
we assume that $S_{t}=\left(Z_{t},\varepsilon_{t}\right)\in\mathcal{Z}\times\mathcal{E}\subseteq\mathbb{R}^{d_{Z}}\times\mathbb{R}^{d_{\varepsilon}}$
where $Z_{t}$ and $\varepsilon_{t}$ satisfy the following conditional
independence condition,
\[
F_{S}\left(Z_{t},\varepsilon_{t}|Z_{t-1},\varepsilon_{t-1},d_{t-1}\right)=F_{\varepsilon}\left(\varepsilon_{t}|Z_{t}\right)F_{Z}\left(Z_{t}|Z_{t-1},d_{t-1}\right).
\]
In many cases $F_{\varepsilon}\left(\varepsilon_{t}|Z_{t}\right)=F_{\varepsilon}\left(\varepsilon_{t}\right)$
in which case $\varepsilon_{t}$ is an i.i.d. sequence and so can
be thought of as idiosyncratic shocks to utility. If no shocks are
present in the model, we can always choose $\varepsilon_{t}=\emptyset$
to be an empty variable so that $S_{t}=Z_{t}$. Throughout, we will
assume that the support $\mathcal{Z}$ is a compact set. This is done
to simplify the theoretical analysis since it, for example, implies
that value functions defined below will lie in the space of bounded
functions on $\mathcal{Z}$, $\mathbb{B}\left(\mathcal{Z}\right)$,
equipped with the sup-norm, $\left\Vert v\right\Vert _{\infty}=\sup_{z\in\mathcal{Z}}\left|v\left(z\right)\right|$.
At the same time, we allow the support of the error term, $\mathcal{E}$,
be unbounded and for both countable and continuously distributed state
variables.

In the above formulation, the model is characterized by the value
function $\nu(s)$. However, it is possible to rewrite the models
in terms of either the so-called \textit{integrated value function}
or the \textit{expected value function} and solve for these instead.
These are defined as
\[
v(Z_{t})=E\left[\nu(Z_{t},\varepsilon_{t})|Z_{t}\right]=\int_{\mathcal{E}}\nu\left(Z_{t},e\right)dF_{\varepsilon}\left(e|Z_{t}\right),
\]
and
\[
V(Z_{t},d_{t})=E\left[\nu(Z_{t+1},\varepsilon_{t+1})|Z_{t},\varepsilon_{t},d_{t}\right]=E[v(Z_{t+1})|Z_{t},d_{t}]=\int_{\mathcal{E}}v(z^{\prime})dF_{Z}\left(z^{\prime}|Z_{t},d_{t}\right),
\]
respectively, where we have used the conditional independence assumption.
Observe that given $v\left(z\right)$, we can recover $V\left(z,d\right)=E[v(Z_{t+1})|Z_{t},d_{t}]$
which in turn can be used to compute $\nu(S_{t})=\max_{d\in\mathcal{D}}\left\{ u(Z_{t},\varepsilon_{t},d)+\beta V\left(Z_{t},d\right)\right\} $.
Thus, there is no loss in focusing on the integrated and expected
value function, except that in the former case we need to compute
$E[v(Z_{t+1})|Z_{t},d_{t}]$ numerically to obtain the value function.
Furthermore, in many cases the expected value function itself is of
interest. For example, the conditional choice probabilities, which
are needed for counterfactuals and for estimation, take as input the
relative expected value function,
\[
P\left(d_{t}=d|Z_{t}=z\right)=M_{u,d}(\beta\Delta V(z)|z),\:M_{u,d}(r|z)=\frac{\partial M_{u}(r|z)}{\partial r\left(d\right)},
\]
where $\Delta V(z,d)=V(z,d)-V(z,D)$, $d\in\mathcal{D}$, and$M_{u}(r|z)$
is a generalized version of the so-called social surplus function
defined as, for any $r=(r(1),...,r(D)),$ 
\begin{equation}
M_{u}(r|z)=\int_{\mathcal{E}}\max_{d\in\mathcal{D}}\left\{ u(z,e,d)+r(d)\right\} dF_{\varepsilon}(e|z).\label{eq: M_u def}
\end{equation}
It is also useful in welfare analysis of policy experiments where
we wish to see how a policy change will affect the expected present
value of lifetime utility (i.e., the expected value function).

Except for a few special cases, analytical expressions of $v$ and
$V$ are not available and so numerical approximations have to be
employed. We will here develop numerical methods for solving for either
$v$ or $V$ instead of $\nu$ for the following reasons: First, $\nu$
is a function of $s=\left(z,\varepsilon\right)$ while $V$ and $v$
are functions of $z$ alone and therefore their approximations are
lower-dimensional problems. Second, $\nu$ is non-differentiable due
to the max-function in (\ref{eq:org bellman}); in contrast, $v(z)$
and $V(z,d)$ are both smooth functions of $z$ if $F_{\varepsilon}(e|z)$
and $F_{Z}\left(z^{\prime}|z,d\right)$ are. If there is no i.i.d.
component in the model, $\varepsilon_{t}=\emptyset$, then $\nu\left(s\right)=\nu\left(z\right)=v(z)$
and so the integrated value function becomes non-smooth. In contrast,
$V(z,d)$ remains smooth even in this case. The functions $v$ and
$V$ each solves their own fixed-point problem: Taking conditional
expectations on both sides of eq. (\ref{eq:org bellman}), $V$ can
be expressed as the solution to
\begin{equation}
V(z,d)=\Gamma(V)(z,d),\label{eq: expected bellman}
\end{equation}
where, with $M_{u}$ defined in eq. (\ref{eq: M_u def}),
\begin{align*}
\Gamma(V)(z,d) & =E\left[\max_{d^{\prime}\in\mathcal{D}}\left\{ u(Z_{t+1},\varepsilon_{t+1},d^{\prime})+\beta V(Z_{t+1},d^{\prime})\right\} |Z_{t}=z,d_{t}=d\right]\\
 & =\int_{\mathcal{Z}}\int_{\mathcal{E}}\max_{d^{\prime}\in\mathcal{D}}\left\{ u(z^{\prime},e,d^{\prime})+\beta V(z^{\prime},d^{\prime})\right\} dF_{\varepsilon}(e|z^{\prime})dF_{Z}(z^{\prime}|z,d)\\
 & =\int_{\mathcal{Z}}M_{u}(\beta V(z^{\prime})|z^{\prime})dF_{Z}(z^{\prime}|z,d).
\end{align*}
 Here and in the following, we let $V\left(z\right)=\left(V(z,1),....,V(z,D)\right)^{\prime}$
denote the $D\times1$-vector of expected value function and similar
for other objects. With this notation, we can represent the fixed-point
problem in vector form, $V(z)=\Gamma(V)(z)$, where
\begin{equation}
\Gamma(V)(z)=\int_{\mathcal{Z}}M_{u}(\beta V(z^{\prime})|z^{\prime})dF_{Z}(dz^{\prime}|z).\label{bellman}
\end{equation}
Next, to derive the fixed-point problem that $v$ solves, again take
conditional expectations on both sides of eq. (\ref{eq:org bellman})
but now only condition on $Z_{t}$ to obtain
\begin{equation}
v(z)=M_{u}(\beta V(z)|z).\label{eq:int_bellman-1}
\end{equation}
Combining this with eq. (\ref{eq: expected bellman}),
\begin{equation}
v\left(z\right)=M_{u}\left(\left.\beta\int_{\mathcal{Z}}M_{u}(\beta V(z^{\prime})|z^{\prime})dF_{Z}(dz^{\prime}|z)\right|z\right)=\bar{\Gamma}(v)(z).\label{eq: integrated bellman}
\end{equation}
where
\[
\bar{\Gamma}(v)(z)=M_{u}\left(\left.\beta\int_{\mathcal{Z}}v\left(z^{\prime}\right)dF_{Z}(dz^{\prime}|z)\right|z\right).
\]
Under regularity conditions provided below, $\Gamma$ and $\bar{\Gamma}$
are contraction mappings and so $V$ and $v$ are well-defined and
unique. The above two transformations of the original problem into
the ones for either the integrated or expected value function are
particular cases of the general class of transformations analyzed
in \citet{Ma&Stachurski}.
\begin{example}
Consider the special case where $u(Z_{t+1},\varepsilon_{t+1},d)=\bar{u}(Z_{t+1},d)+\lambda\varepsilon_{t+1}\left(d\right)$
for some scale parameter $\lambda>0$ and $F_{\varepsilon}(e|z)=F_{\varepsilon}(e)$
in which case
\[
M_{u}(r|z)=\int_{\mathcal{E}}\max_{d\in\mathcal{D}}\left\{ \bar{u}(z,d)+\lambda e\left(d\right)+r(d)\right\} dF_{\varepsilon}(e)=G_{\lambda}\left(\bar{u}(z)+r\right),
\]
where $G_{\lambda}\left(r\right):=\int_{\mathcal{E}}\max_{d\in\mathcal{D}}\left\{ \lambda e\left(d\right)+r(d)\right\} dF_{\varepsilon}(e)$.
Thus,
\[
\bar{\Gamma}(v)(z)=G_{\lambda}\left(\bar{u}(z)+\beta\int_{\mathcal{Z}}v\left(z^{\prime}\right)dF_{Z}(z^{\prime}|z)\right).
\]
If $\varepsilon_{t}\left(1\right),...,\varepsilon_{t}\left(D\right)$
are mutually independent and each component follows a suitably normalized
extreme value distribution then (see, e.g., \citealp{Rustetal2002}),
\begin{equation}
G_{\lambda}(r)=\lambda\log\left[\sum_{d\in\mathcal{D}}\exp\left(\frac{r\left(d\right)}{\lambda}\right)\right].\label{eq: G social surplus}
\end{equation}
\end{example}

\section{Simulated Bellman operators\label{sec:Simulated-Bellman}}

As a first step towards a computationally feasible method for solving
for either $v$ or $V$, we here develop simulated versions of their
two Bellman operators and then introduce the smoothing device. To
allow for added flexibility and precision in the implementation and
to cover as special case a modified version of Rust's self-approximating
solution method, we employ importance sampling: Let $\Phi_{Z}\left(z^{\prime}|z,d\right)$
and $\Phi_{\varepsilon}\left(e|z\right)$ be conditional importance
sampling distribution functions as chosen by the researcher. These
have to be chosen such that $F_{Z}(\cdot|z,d)$ and $F_{\varepsilon}\left(\cdot|z\right)$
are absolutely continuous w.r.t. $\Phi_{Z}\left(\cdot|z,d\right)$
and $\Phi_{\varepsilon}\left(\cdot|z\right)$, respectively, with
Radon-Nikodym derivatives $w_{Z}\left(\cdot|z,d\right)\geq0$ and
$w_{\varepsilon}\left(\cdot|z\right)\geq0$ so that
\begin{equation}
\frac{dF_{Z}(z^{\prime}|z,d)}{d\Phi_{Z}\left(z^{\prime}|z,d\right)}=w_{Z}\left(z^{\prime}|z,d\right),\:\frac{dF_{\varepsilon}(e|z)}{d\Phi_{\varepsilon}\left(e|z\right)}=w_{\varepsilon}\left(e|z\right).\label{eq: w def}
\end{equation}
 We will throughout assume that eq. (\ref{eq: w def}) is satisfied.
In the leading case $dF_{Z}=f_{Z}d\mu_{Z}$ and $d\Phi_{Z}=\phi_{Z}d\mu_{Z}$
for some measure $\mu_{Z}$ in which case $w_{Z}=f_{Z}/\phi_{Z}$
and similar for the sampling of $\varepsilon_{t}$. The above covers
the case where $F_{Z}(z^{\prime}|z,d)$ is a continuous distribution
(in which case $\mu_{Z}$ is the Lesbesque measure), a discrete distribution
(in which case $\mu_{Z}$ is the counting measure) and the mixed case.
With discrete finite support, we could in principle compute the exact
Bellman equation and its corresponding solution and so would not need
to resort to numerical methods. But if the discrete support is large
this may still be computationally very demanding and so even in this
case the numerical methods developed below may be computationally
attractive, c.f. \citet{Arcidiaconoetal2013}.

Given the chosen importance sampler, we can rewrite $\Gamma(V)(z,d)$
as
\[
\Gamma(V)(z,d)=\int_{\mathcal{Z}}\int_{\mathcal{E}}\max_{d^{\prime}\in\mathcal{D}}\left\{ u(s^{\prime},d^{\prime})+\beta V(z^{\prime},d^{\prime})\right\} w\left(s^{\prime}|z,d\right)d\Phi(s^{\prime}|z,d)
\]
where $s^{\prime}=\left(z^{\prime},e^{\prime}\right)$ and
\[
w\left(s^{\prime}|z,d\right)=w_{\varepsilon}\left(e^{\prime}|z^{\prime}\right)w_{Z}\left(z^{\prime}|z,d\right),\:\Phi\left(s^{\prime}|z,d\right)=\Phi_{\varepsilon}(e^{\prime}|z^{\prime})\Phi_{Z}\left(z^{\prime}|z,d\right).
\]
For any given candidate $V$, we can then approximate this integral
by Monte Carlo methods: First generate $N\geq1$ i.i.d. draws, $Z_{i}\left(z,d\right)\sim\Phi_{Z}\left(\cdot|z,d\right)$
and $\varepsilon_{i}\left(z,d\right)\sim\Phi_{\varepsilon}(\cdot|Z_{i}\left(z,d\right))$,
$i=1,...,N$, and then compute 
\begin{equation}
\Gamma_{N}(V)(z,d)=\sum_{i=1}^{N}\max_{d^{\prime}\in\mathcal{D}}\left\{ u\left(S_{i}\left(z,d\right),d^{\prime}\right)+\beta V\left(Z_{i}\left(z,d\right),d^{\prime}\right)\right\} w_{N,i}\left(z,d\right),\label{eq: Gamma_N non-smooth}
\end{equation}
where $S_{i}\left(z,d\right)=\left(Z_{i}\left(z,d\right),\varepsilon_{i}\left(z,d\right)\right)$
and
\begin{equation}
w_{N,i}\left(z,d\right)=\frac{w\left(S_{i}\left(z,d\right)|z,d\right)}{\sum_{i=1}^{N}w\left(S_{i}\left(z,d\right)|z,d\right)}.\label{eq: sampling weight}
\end{equation}
Note here that we normalize the importance weights so that $\sum_{i=1}^{N}w_{N,i}\left(z,d\right)=1$.
This is done to ensure that $\Gamma_{N}$ is a contraction mapping
on $\mathbb{B}\left(\mathcal{Z}\right)^{D}$. Similarly, we approximate
$\bar{\Gamma}\left(v\right)$ by 
\begin{equation}
\bar{\Gamma}_{N}(v)(z)=\sum_{j=1}^{N}\max_{d'\in\mathcal{D}}\left\{ u(z,\varepsilon_{j}\left(z,d^{\prime}\right),d^{\prime})+\beta\sum_{i=1}^{N}v\left(Z_{i}\left(z,d^{\prime}\right)\right)w_{Z,N,i}\left(z,d^{\prime}\right)\right\} w_{\varepsilon,N,j}\left(z,d^{\prime}\right),\label{eq: Gamma_N 2 non-smooth}
\end{equation}
where again we normalize the weights to ensure $\bar{\Gamma}_{N}$
is a contraction on $\mathbb{B}\left(\mathcal{Z}\right)$,
\[
w_{Z,N,i}\left(z,d\right)=\frac{w_{Z}\left(Z_{i}\left(z,d\right)|z,d\right)}{\sum_{i=1}^{N}w_{Z}\left(Z_{i}\left(z,d\right)|z,d\right)},\:w_{\varepsilon,N,i}\left(z,d\right)=\frac{w_{\varepsilon}\left(\varepsilon_{i}\left(z,d\right)|z\right)}{\sum_{i=1}^{N}w_{\varepsilon}\left(\varepsilon_{i}\left(z\right)|z\right)}.
\]
When $\varepsilon_{t}=\emptyset$, the simulated Bellman operator
$\bar{\Gamma}_{N}$ includes as special cases the ones considered
in \citet{Rust1997curse} (who chooses $\Phi_{Z}$ as the uniform
distribution on $\mathcal{Z}$) and \citet{Pal&Stacurski2013} (who
chooses $\Phi_{Z}=F_{Z}$).

\textbf{Example 1 (continued).} Suppose we can compute the integral
w.r.t. $\varepsilon_{t}$ analytically. In this case, the following
simplified version of the simulated Bellman operator can be employed,
\begin{equation}
\bar{\Gamma}_{N}(v)(z)=G_{\lambda}\left(\bar{u}(z)+\beta\sum_{i=1}^{N}v\left(Z_{i}\left(z\right)\right)w_{Z,N,i}\left(z\right)\right),\label{eq: Gamma_N Ex 1}
\end{equation}
where $G_{\lambda}$ was defined in eq. (\ref{eq: G social surplus}).
Importantly, the $\max$-function has been replaced by its smoothed
version $G_{\lambda}\left(\cdot\right)$.

If $F_{\varepsilon}(e|z)$ and $F_{Z}(z'|z,d)$ are smooth functions
w.r.t. $z$ then $\Gamma(V)(z)$ and $\bar{\Gamma}(v)(z)$ will be
smooth functions of $z$ as well. In contrast, the general versions
of $\Gamma_{N}(V)(z,d)$ and $\bar{\Gamma}_{N}(v)(z)$ are non-smooth
due to the presence of the $\max$-function in their definitions which
does not get smoothed for finite $N$. This in turn implies that their
corresponding fixed points, $V_{N}\left(z\right)=\Gamma_{N}(V_{N})\left(z\right)$
and $v_{N}\left(z\right)=\bar{\Gamma}_{N}(v_{N})\left(z\right)$,
will be non-differentiable w.r.t. the state variables, $z$, and w.r.t.
any underlying structural parameters in the model. This is an unattractive
feature for two reasons: First, estimation and counterfactuals will
be non-smooth problems. Second, the theoretical analysis of $V_{N}$
and $v_{N}$ becomes more complicated.

To resolve this issue, we take inspiration from the additive model
in Example 1 and propose to smooth the simulated Bellman operators
by replacing the ``hard'' $\max$-function appearing in eqs. (\ref{eq: Gamma_N non-smooth})
and (\ref{eq: Gamma_N 2 non-smooth}) by its smoothed version $G_{\lambda}(r)$
defined in eq. (\ref{eq: G social surplus}). This yields the following
smoothed simulated operators,
\begin{equation}
\Gamma_{N}(V)(z;\lambda)=\sum_{i=1}^{N}G_{\lambda}\left(u\left(S_{i}\left(z,d\right)\right)+\beta V\left(Z_{i}\left(z,d\right)\right)\right)w_{N,i}\left(z\right),\label{eq: smooth Gamma_N}
\end{equation}
\begin{equation}
\bar{\Gamma}_{N}(v)(z;\lambda)=\sum_{j=1}^{N}G_{\lambda}\left(u(z,\varepsilon_{j}\left(z\right))+\beta\sum_{i=1}^{N}v\left(Z_{i}\left(z\right)\right)w_{Z,N,i}\left(z\right)\right)w_{\varepsilon,N,j}\left(z\right),\label{eq: smooth Gamma-bar_N}
\end{equation}
where $u(z,\varepsilon_{j}\left(z\right))=\left(u(z,\varepsilon_{j}\left(z,1\right),1),...,u(z,\varepsilon_{j}\left(z,D\right),D)\right)$
and other vector functions are defined similarly. Setting $\lambda=0$
in eqs. (\ref{eq: smooth Gamma_N})-(\ref{eq: smooth Gamma-bar_N}),
we recover the original non-smooth versions defined in eqs. (\ref{eq: Gamma_N non-smooth})-(\ref{eq: Gamma_N 2 non-smooth}).
Thus, the smoothed versions are generalized versions of the original
ones. The use of $G_{\lambda}(r)$ in place of $\max_{d\in\mathcal{D}}r\left(d\right)$
generates an additional bias in the approximate solutions, but this
can be controlled for by suitable choice of $\lambda$. We now interpret
$\lambda>0$ as a smoothing parameter that plays a role similar to
that of the bandwidth in kernel regression estimation. Elementary
calculations show
\begin{equation}
0\leq G_{\lambda}\left(r\right)-\max_{d\in\mathcal{D}}r\left(d\right)\leq\lambda\log D,\label{eq: G approx error}
\end{equation}
so that $G_{\lambda}\left(r\right)\rightarrow\max_{d\in D}r\left(d\right)$,
as $\lambda\rightarrow0$, uniformly in $r\in\mathbb{R}^{D}$. Thus,
the smoothing entails a bias of order $O_{P}$$\left(\lambda\right)$.
We discuss the choice of $\lambda$ in practice in the next section.

In some situations, $G_{\lambda}(r)$ appears in the Bellman operators
as an inherent feature of the model specification in which case no
smoothing bias will be present. We saw this in Example 1 and it extends
to the following class of models: Suppose that $\varepsilon_{t}=\left(\varepsilon_{t}^{\text{\ensuremath{\left(1\right)}}},\varepsilon_{t}^{\left(2\right)}\right)$
with $\varepsilon_{t}^{\text{\ensuremath{\left(1\right)}}}=\left(\varepsilon_{t}^{\text{\ensuremath{\left(1\right)}}}\left(1\right),...,\varepsilon_{t}^{\text{\ensuremath{\left(1\right)}}}\left(D\right)\right)$
are mutually independent extreme value shocks that enter the per-period
utility additively,
\begin{equation}
d_{t}=\arg\max_{d\in\mathcal{D}}\left\{ \bar{u}(Z_{t},\varepsilon_{t}^{\left(2\right)},d)+\lambda\varepsilon_{t}^{\text{\ensuremath{\left(1\right)}}}\left(d\right)+\beta V\left(Z_{t},d\right)\right\} ,\label{eq: model no bias}
\end{equation}
where as in Example 1 $\lambda>0$ is scale parameter that determines
the impact of $\varepsilon_{t}^{\text{\ensuremath{\left(1\right)}}}\left(d\right)$
on the per-period utility. By the same arguments as in Example 1,
we find that the expected value function in this case solves $\Gamma_{\lambda}\left(V\right)=V$
where
\[
\Gamma_{\lambda}(V)(z)=\int_{\mathcal{Z}}\int_{\mathcal{E}}G_{\lambda}\left(u(z^{\prime},e^{\prime}))+\beta V(z^{\prime})\right)dF_{\varepsilon^{\left(-1\right)}}(e^{\prime}|z)dF_{Z}(z'|z,d),
\]
and $\Gamma_{N,\lambda}$ in eq. (\ref{eq: smooth Gamma_N}) is clearly
an unbiased simulated version of $\Gamma_{\lambda}$. Similarly, $\bar{\Gamma}_{N,\lambda}$
is an unbiased estimator of $\bar{\Gamma}_{\lambda}$. To summarize,
if the original model of interest contains an additive extreme value
term, which is the case in many empirical papers, $G_{\lambda}$ appears
as part of the model and so no smoothing bias will be present in our
proposed simulated Bellman operators. 

The above shows that the smoothing device corresponds to adding structural
shocks to the DDP of interest. In earlier work on solving DDPs, researchers
have in some cases done the opposite and removed structural errors
in order to facilitate the numerical solution of the model; see \citet{Lumsdaineetal1992}
for one example of this. This was, however, done in the context of
discrete state variables with a small number of support points in
which case removing continuous structural errors meant that the Bellman
operators could be evaluated analytically. Our method is aimed at
models where the state variables are either continuous or have a very
large discrete support in which case simulations are required in the
first place to evaluate the Bellman operator. Once simulations are
introduced, there is little computational gains from removing shocks
from the model and instead the introduction of smoothing facilitates
solving and analyzing the corresponding solution.

\section{Approximate value functions\label{sec:Approximate-value}}

The smoothed simulated Bellman operators $\Gamma_{N}(V)\left(z,\lambda\right)$
and $\bar{\Gamma}_{N}(v)\left(z,\lambda\right)$ in eqs. (\ref{eq: smooth Gamma_N})-(\ref{eq: smooth Gamma-bar_N})
are functionals that, for given function $V$ and $v$, depend on
$\left(z,\lambda\right)$. We will here and in the following treat
them as functionals that take given function $V\left(z,\lambda\right)$
and $v\left(z,\lambda\right)$, respectively, and map them into functions
of $\left(z,\lambda\right)\in\mathcal{Z}\times\left[0,\bar{\lambda}\right]$
for some $\bar{\lambda}>0$. This simplifies the analysis of the impact
of smoothing. In particular, under suitable regularity conditions,
$\Gamma_{N}$ and $\bar{\Gamma}_{N}$ are contraction mappings on
$\mathbb{B}\left(\mathcal{Z}\times\left[0,\bar{\lambda}\right]\right)^{D}$
and $\mathbb{B}\left(\mathcal{Z}\times\left[0,\bar{\lambda}\right]\right)$,
respectively, where $\mathbb{B}\left(\mathcal{Z}\times\left[0,\bar{\lambda}\right]\right)$
denotes the space of bounded functions with domain $\mathcal{Z}\times\left[0,\bar{\lambda}\right]$.
Thus, they have unique fixed points $V_{N}\left(z,\lambda\right)$
and $v_{N}\left(z,\lambda\right)$ solving
\begin{equation}
V_{N}\left(z,\lambda\right)=\Gamma_{N}(V_{N})\left(z,\lambda\right),\:v_{N}=\bar{\Gamma}_{N,\lambda}(v_{N})\left(z,\lambda\right).\label{eq: Sim Fixed point}
\end{equation}
In practice, we will only solve for the particular value of $\lambda$
as chosen by us, but for the theory it proves helpful to treat the
solutions as mappings defined on $\left(z,\lambda\right)\in\mathcal{Z}\times\left(0,\bar{\lambda}\right)$.
However, solving these two simulated Bellman equations are not generally
feasible since these are infinite-dimensional problems. We here present
two ways to reduce the problems to finite-dimensional ones. The first
method is a generalized version of the so-called self-approximating
method proposed in \citet{Rust1997curse} while the second one uses
projection-based methods as advocated by \citet{Pal&Stacurski2013}.

\subsection{Self-approximating method}

\citet{Rust1997curse} proposed to turn the infinite-dimensional problems
in eq. (\ref{eq: Sim Fixed point}) into a finite-dimensional ones
by choosing the importance sampling to be based on marginal, instead
of conditional distributions. In our generalized version this corresponds
to restricting $\Phi_{Z}\left(z^{\prime}|z,d\right)=\Phi_{Z}\left(z^{\prime}\right)$
for some marginal distribution $\Phi_{Z}\left(\cdot\right)$ so that
the draws $Z_{i}\sim\Phi_{z}\left(\cdot\right)$ and $\varepsilon_{i}\sim\Phi_{\varepsilon}(\cdot|Z_{i})$,
$i=1,...,N$ no longer depend on $\left(z,d\right)$. In this case,
for a given value of $\lambda\in\left[0,\bar{\lambda}\right]$, the
fixed-point problems in eq. (\ref{eq: Sim Fixed point}) reduce to
the following two sets of $N$ nonlinear equations, 
\begin{equation}
V_{N,k}=\sum_{i=1}^{N}G_{\lambda}\left(u\left(S_{i}\right)+\beta V_{N,i}\right)w_{N,i}\left(Z_{k}\right),\label{eq: self-approx V}
\end{equation}
\begin{equation}
v_{N,k}=\sum_{j=1}^{N}G_{\lambda}\left(u(Z_{k},\varepsilon_{j})+\beta\sum_{i=1}^{N}v_{N,i}w_{z,i}\left(Z_{k}\right)\right)w_{\varepsilon,N,j}\left(Z_{k}\right),\label{eq: self-approx v}
\end{equation}
for $k=1,...,N$, that can be solved for w.r.t. $\left\{ V_{N,\lambda,k}:k=1,...,N\right\} $
and $\left\{ v_{N,\lambda,k}:k=1,...,N\right\} $, respectively. Here,
$V_{N,k}=V_{N}\left(Z_{k},\lambda\right)$ and $v_{N,k}=v_{N,\lambda}\left(Z_{k},\lambda\right)$,
$k=1,...,N$. Each of the two sets of equations have a unique solution
due to the contracting property of $\Gamma_{N,\lambda}$ and $\bar{\Gamma}_{N,\lambda}$.
Once, for example, eq. (\ref{eq: self-approx V}) has been solved,
the approximate expected value functions can be evaluated at any other
value $z$ by
\[
V_{N}\left(z,\lambda\right)=\sum_{i=1}^{N}G_{\lambda}\left(u\left(S_{i}\right)+\beta V_{N,i}\right)w_{N,i}\left(z\right).
\]
Note that $V_{N}\left(z,\lambda\right)$ is a smooth function even
if $\lambda=0$ as long as $w_{N,i}\left(z\right)$ is smooth and
so smoothing is not needed for this property to hold when marginal
samplers are employed. However, without smoothing, the set of equations
(\ref{eq: self-approx V}) become non-smooth w.r.t. the variables
$\left\{ V_{N,k}:k=1,...,N\right\} $ and so cannot be solved using
derivative-based methods. Thus, the numerical implementation of the
self-approximating method still benefits from smoothing.

In addition to smoothing, the above self-approximating method differs
from Rust's original proposal in two other ways: First, while \citet{Rust1997curse}
solved for the value function $\nu\left(z,\varepsilon\right)$, we
here solve for either $V\left(z,\lambda\right)$ or $v\left(z,\lambda\right)$
for a fixed value of $\lambda$. As explained in Section \ref{sec:Model},
$V$ and $v$ convey the same information as $\nu$ and at the same
time they are of lower dimension in terms of variables and are more
smooth, features which facilitate their numerical approximation. Moreover,
our formulation allows for the following generalized version of the
simulated Bellman equations for $v_{N}$,
\begin{equation}
v_{N,k}=\sum_{j=1}^{\tilde{N}}G_{\lambda}\left(u(Z_{k},\varepsilon_{j})+\beta\sum_{i=1}^{N}v_{N,i}w_{z,N,i}\left(Z_{k}\right)\right)w_{\varepsilon,N,j}\left(Z_{k}\right),\label{eq: self-approx v-1}
\end{equation}
where we allow for different number of draws from $\Phi_{\varepsilon}$
and $\Phi_{Z}$. In particular, we can choose $\tilde{N}$ as large
as we wish (thereby decreasing the variance of the problem) without
increasing the number of variables that need to be solved for ($N$).
A similar generalization of the simulated Bellman equations for $V_{N}$
is possible. Second, we here only require that the state dynamics
together with the chosen importance sampler satisfy (\ref{eq: w def});
in contrast, \citet{Rust1997curse} assumed that $S_{t}$ was continuously
distributed with compact support and chose as importance sampler the
uniform distribution with same support. Thus, our version allows for
a broader class of models and samplers. 

The self-approximating method may not always work well: First, finding
a marginal distribution $\Phi_{Z}\left(\cdot\right)$ so that (\ref{eq: w def})
holds can be difficult in some models. For example, in many specifications
with continuous dynamics, the transition density $f_{Z}\left(z^{\prime}|z,d\right)$
of $Z_{t}$ will have singularities, e.g., $\lim_{z^{\prime}\rightarrow z}f_{Z}\left(z^{\prime}|z,d\right)=+\infty$,
in which case $w_{Z}\left(z^{\prime}|z,d\right)=f_{Z}\left(z^{\prime}|z,d\right)/\phi_{Z}\left(z^{\prime}\right)$
is not well-defined no matter how we choose $\phi_{Z}\left(z^{\prime}\right)$.
And even if (\ref{eq: w def}) does hold, the use of marginal samplers
instead of conditional ones will generally lead to a larger variance
of the solutions since the ``marginal'' draws $Z_{1},...,Z_{N}$
do not adapt to the changing shape of $F_{Z}\left(\cdot|z,d\right)$
as a function of $z$. In particular, many of the draws may fall outside
of the support of $F_{Z}\left(\cdot|z,d\right)$ and so are ``wasted''
in which case a large $N$ is required to achieve a reasonable approximation;
see Section \ref{sec:Numerical} for an example of this. This issue
tends to become more severe in higher dimensions ($d_{Z}$ is ``large'')
since the volume of the support shrinks, and so the self-approximating
method will generally suffer from a built-in curse-of-dimensionality.
This curse-of-dimensionality does not appear in the subclass of models
that \citet{Rust1997curse} focused on where it was assumed that $S_{t}|S_{t-1},d_{t-1}$
has support $\left[0,1\right]^{\dim\left(S_{t}\right)}$ for all values
of $S_{t-1},d_{t-1}$.

Finally, given that $V_{N,\lambda}$ and $v_{N,\lambda}$ are solutions
to non-linear equations, a large variance in the simulated Bellman
operator translates into a large bias as is well-known from non-linear
GMM estimators. This can be controlled for by choosing $N$ large.
But large $N$ means that numerically solving either (\ref{eq: self-approx V})
or (\ref{eq: self-approx v}) becomes computationally very costly.
These issues motivate us to pursue a sieve-based solution strategy.

\subsection{Sieve-based method}

We now return to the general versions of the simulated Bellman operators
and so again allow for conditional importance samplers. Let $\mathcal{\bar{V}\subseteq\mathbb{B}\left(\mathcal{Z}\right)}$
be a suitable function space that $z\mapsto v_{N}\left(z;\lambda\right)$
defined in (\ref{eq: Sim Fixed point}) is known to lie in; see below
for more details on this. We then choose a finite-dimensional function
space (commonly called a sieve in the econometrics literature) $\mathcal{\bar{V}}_{K}=\left\{ v_{K}\left(\cdot;\alpha\right):\mathcal{Z}\mapsto\mathbb{R}|\alpha\in\mathcal{A}_{K}\right\} \subseteq\bar{\mathcal{V}}$,
where $\mathcal{A}_{K}\subseteq\mathbb{R}^{K}$ is a parameter set
with $K<\infty,$ that provides a good approximation to functions
in $\bar{\mathcal{V}}$. Similarly, we let $\mathcal{V}\subseteq\mathbb{B}\left(\mathcal{Z}\right)^{D}$
be a space of $D$-dimensional vector functions that the solution
$z\mapsto V_{N}\left(z;\lambda\right)$ to (\ref{eq: Sim Fixed point})
lie in and $\mathcal{V}_{K}=\left\{ V_{K}\left(\cdot;\alpha\right):\mathcal{Z}\mapsto\mathbb{R}^{D}|\alpha\in\mathcal{A}_{K}\right\} \subseteq\mathcal{V}$
be our sieve for this space. Let

\begin{equation}
\bar{\Pi}_{K}\left(v\right)=\arg\min_{v^{\prime}\in\mathcal{\bar{V}}_{K}}\left\Vert v-v^{\prime}\right\Vert _{\mathcal{\bar{V}}},\:\Pi_{K}\left(V\right)=\arg\min_{V^{\prime}\in\mathcal{V}_{K}}\left\Vert V-V^{\prime}\right\Vert _{\mathcal{V}},\label{eq: Projector def}
\end{equation}
be the corresponding projections for given (pseudo-) norms $\left\Vert \cdot\right\Vert _{\mathcal{\bar{V}}}$
and $\left\Vert \cdot\right\Vert _{\mathcal{V}}$ as chosen by us
as well. We then approximate $V_{N,\lambda}$ and $v_{N,\lambda}$
by the solutions to the projected Bellman equations,
\begin{equation}
\hat{v}_{N,\lambda}=\arg\min_{v\in\mathcal{\bar{V}}_{K}}\left\Vert v-\bar{\Pi}_{K}\bar{\Gamma}_{N,\lambda}(v)\right\Vert _{\mathcal{\bar{V}}},\:\hat{V}_{N,\lambda}=\arg\min_{V\in\mathcal{V}_{K}}\left\Vert V-\Pi_{K}\Gamma_{N,\lambda}(V)\right\Vert _{\mathcal{V}}.\label{eq: projected sim fixed point}
\end{equation}
These are finite-dimensional problems of size $K$. When $K$ is small
relative to $N$, which will generally be the case, the above problems
are computationally much more tractable compared to the corresponding
self-approximating ones. Note here that these projection-based approximations
are different from the least-squares approximations that would solve
$\min_{V\in\mathcal{V}_{K}}\left\Vert V-\Gamma_{N,\lambda}(V)\right\Vert _{\mathcal{V}}$
and $\min_{v\in\mathcal{\bar{V}}_{K}}\left\Vert v-\bar{\Gamma}_{N,\lambda}(v)\right\Vert _{\mathcal{\bar{V}}}$,
respectively. In particular, by suitable choice of the projection
operators, $\Pi_{K}\Gamma_{N,\lambda}$ and $\bar{\Pi}_{K}\bar{\Gamma}_{N,\lambda}$
will be contraction mappings w.r.t. $\left\Vert \cdot\right\Vert _{\infty}$
guaranteeing that $\hat{V}_{N,\lambda}$ and $\hat{v}_{N,\lambda}$
exist and are unique. The following discussion focuses on the integrated
value function approximation since it carries over with only minor
modifications to the one of the expected value function. We discuss
their numerical implementation in further detail in the subsection
below. 

The projection operator $\bar{\Pi}_{K}$ can be thought of as a function
approximator with the approximation error being $v-\bar{\Pi}_{K}\left(v\right)$
for a given function $v$. Roughly speaking, the projection-based
method approximates $v_{N}$ in (\ref{eq: Sim Fixed point}) by $\hat{v}_{N}=\bar{\Pi}_{K}\left(v_{N}\right)$
which incurs an additional sieve approximation error, $v_{N}-\bar{\Pi}_{K}\left(v_{N}\right)$.
The smoothness of $v_{N}$ here proves helpful since many well-known
sieves are able to provide good approximations of smooth functions
using a low-dimensional space (``small'' $K$). Due to these features,
our proposed projection-based solutions will generally suffer from
quite small additional biases relative to the exact simulated solution.
This is in contrast to existing projection-based solution methods,
such as the one in \citet{Pal&Stacurski2013}, that aim at approximating
the value function $\nu\left(s\right)$ which is non-differentiable.

The smoothness of $v_{N}$ here help guiding us in choosing the sieve:
It allows us to restrict $\mathcal{\bar{V}}$ to a suitable smoothness
class and then import existing approximation methods for smooth functions
as developed in the literature on numerical methods and nonparametric
econometrics. A leading example is the class of linear function approximations
where the finite-dimensional function space takes the form of $\mathcal{\bar{V}}_{K}=\left\{ \alpha^{\prime}B_{K}\left(z\right):\alpha\in\mathbb{R}^{K}\right\} $
for a set of basis functions $B_{K}\left(z\right)=\left(b_{1}\left(z\right),...,b_{K}\left(z\right)\right)^{\prime}$.
The basis functions can be chosen as, for example, Chebyshev polynomials
or B-splines that are able to approximate smooth functions well. However,
other non-linear function space are possible such as wavelets, artificial
neural networks and shrinkage-type function approximators such as
LASSO, where the additional constraints are imposed on $\alpha$;
we refer to \citet{Chen2007} for a general overview of different
function approximators and constrained sieve estimators. We also allow
for flexibility in terms of the chosen norms $\left\Vert \cdot\right\Vert _{\mathcal{\bar{V}}}$
with a leading example being $\left\Vert v\right\Vert _{\mathcal{\bar{V}}}=\sum_{i=1}^{M}v^{2}\left(z_{i}\right)$
for a set of design points $z_{1},...,z_{M}\in\mathcal{Z}.$ Very
often the $M\geq1$ design points will be chosen in conjunction with
the sieve.

The above procedure does not suffer from any of the above mentioned
issues of the self-approximating method: We can use conditional importance
samplers freely which can be designed to control the variance of the
simulated Bellman operators; and the dimension of the problem remains
$K$ irrespectively of the number of draws $N$. The main drawback
is that unique solutions to eqs. (\ref{eq: projected sim fixed point})
do not necessarily exist for a given choice of $N$ and $K$. A sufficient
condition for this to hold is that $\bar{\Pi}_{K}$ is a non-expansive
operator w.r.t $\left\Vert \cdot\right\Vert _{\infty}$, that is,
$\left\Vert \bar{\Pi}_{K}\left(v_{1}\right)-\bar{\Pi}_{K}\left(v_{2}\right)\right\Vert _{\infty}\leq\left\Vert v_{1}-v_{2}\right\Vert _{\infty}$
for any two functions $v_{1},v_{2}\in\mathcal{\bar{\mathcal{V}}}$,
since this translates into $\bar{\Pi}_{K}\Gamma_{N}$ being a contraction
mapping. However, while $\bar{\Pi}_{K}$ is non-expansive w.r.t. $\left\Vert \cdot\right\Vert _{\bar{\mathcal{V}}}$
by definition, it is not necessarily non-expansiveness w.r.t $\left\Vert \cdot\right\Vert _{\infty}$.
\citet{Pal&Stacurski2013} provide some examples of projections that
are non-expansive w.r.t. $\left\Vert \cdot\right\Vert _{\infty}$,
but these are unfortunately computationally expensive to use in general.
But $\bar{\Pi}_{K}$ will generally be close to non-expansive w.r.t
$\left\Vert \cdot\right\Vert _{\infty}$ asymptotically as $K\rightarrow\infty$
for a wide range of sieves and pseudo-norms in the sense that
\[
\left\Vert \bar{\Pi}_{K}\right\Vert _{op,\infty}:=\sup_{v\in\bar{\mathcal{V}},\left\Vert v\right\Vert =1}\left\Vert \bar{\Pi}_{K}\left(v\right)\right\Vert _{\infty}\leq\sup_{v\in\bar{\mathcal{V}},\left\Vert v\right\Vert =1}\left\Vert \bar{\Pi}_{K}\left(v\right)-v\right\Vert _{\infty}+1,
\]
where the first term in the last expression will go to zero in great
generality as $K\rightarrow\infty$ for many popular sieves (see next
subsection for details). Given that $\bar{\Gamma}_{N}$ is a contraction
with Lipschitz coefficient $\beta<1$, this in turn implies that $\bar{\Pi}_{K}\bar{\Gamma}_{N}$
will be a contraction mapping for all $K$ large enough. This will
be used in our asymptotic analysis of the algorithm. Unfortunately,
it is generally not known how large $K$ should be chosen to ensure
$\bar{\Pi}_{K}\bar{\Gamma}_{N}$ is a contraction. But in our numerical
experiments we did not experience any convergence problems.

\subsection{\label{subsec:Numerical-implementation}Numerical implementation
of the two methods}

We here discuss in more detail the numerical implementation of the
self-approximating and projection-based methods. First, the researcher
has to choose the importance sampling distributions, the smoothing
parameter $\lambda$ and, in the case of the projection-based method,
the function approximation method. Second, given these choices, either
eq. (\ref{eq: self-approx v}) or (\ref{eq: projected sim fixed point})
has to be solved for. As before, the discussion here focuses on solving
for the integrated value function since most results and arguments
for this case carries over with very minor modifications to the expected
value function. 

\subsubsection{Importance sampler}

The choices of $\Phi_{Z}$ and $\Phi_{\varepsilon}$ determine the
variance of $\bar{\Gamma}_{N}$ and should ideally be tailored to
minimize it. In the case of projection-based methods, where we can
choose $\Phi_{Z}$ and $\Phi_{\varepsilon}$ as conditional distributions,
we can rely on the already existing theory for efficient importance
sampling for how to do so; see Chapter 3 in \citet{Robert&Casella2004}
for an introduction. In our numerical experiments, we did not experiment
with different choices and throughout set $\Phi_{Z}=F_{Z}$ and $\Phi_{\varepsilon}=F_{\varepsilon}$.

In the case of the self-approximating method, the choice of $\Phi_{Z}$
is restricted to the class of marginal distributions. Generally, this
entails a large variance of the corresponding simulated Bellman operators.
It will hold in great generality that $\left(Z_{t},d_{t}\right)$
has a stationary distribution, say, $F_{S}^{*}\left(z,d\right)$.
In this case, a suitable choice would be the marginal of this, $\Phi_{Z}\left(z\right)=F_{S}^{*}\left(z,D\right)$.
However, the stationary distribution depends on the value function
and so is rarely available on closed form; so this strategy requires
an initial exploration of the model and its solution. Alternatively,
one can try to construct a good approximation of the stationary approximation
through a mixture Markov model on the form $\Phi_{Z}\left(z^{\prime}\right)=\sum_{d\in\mathcal{D}}\int\omega_{d}\left(z\right)F_{Z}\left(z^{\prime}|z,d\right)d\mu_{Z}\left(z\right)$
for a set of pre-specified mixture weights $\omega_{d}\left(z\right)\geq0$.
In the numerical experiments, we follow \citet{Rust1997curse} and
choose $\Phi_{Z}\left(z\right)$ as the uniform distribution on $\mathcal{Z}$
which we conjecture is far from optimal in many cases, and so more
research in this direction is needed.

\subsubsection{Smoothing}

The use of $G_{\lambda}\left(r\right)$ in place of $\max_{d\in\mathcal{D}}r\left(d\right)$
generally generates an additional bias in the corresponding integrated
value function of order $O\left(\lambda\right)$. At the same time,
the variance of $v_{N}$ is an increasing function of $\lambda$.
Thus, ideally we would like to choose $\lambda$ to balance these
two effects. A natural criterion would be to minimize the so-called
integrated mean-square-error, $\lambda^{*}=\arg\min_{\lambda\geq0}E\left[\int_{\mathcal{Z}}\left\Vert v_{N,\lambda}\left(z\right)-v\left(z\right)\right\Vert ^{2}dF_{Z}\left(z\right)\right]$,
where $F_{Z}\left(z\right)$ is a suitably chosen distribution such
as the stationary one of $Z_{t}$. Since $v\left(z\right)$ is unknown
and we cannot evaluate the expectations, $\lambda^{*}$ cannot be
solved for but cross-validation methods can be used instead. This
could in principle be done along the same lines as bandwidth selection
for smoothed empirical cdfs, see \citet{Bowmanetal1998}. However,
this is computationally somewhat burdensome. Moreover, in our numerical
experiments we found that the quality of the approximate value function
was quite insensitive to the choice of $\lambda$ and so in practice
we recommend using very little smoothing such as $\lambda=0.01$.

\subsubsection{\label{subsec:Function-approximation}Function approximation}

As mentioned earlier, many approximation architectures are available
in the literature. In our numerical experiments we focus on the class
of linear function approximators where $\bar{\mathcal{V}}_{K}=\left\{ \alpha'B_{K}\left(z\right):\alpha\in\mathbb{R}^{K}\right\} $
for a set of pre-specified basis functions $B_{K}\left(z\right)\in\mathbb{R}^{K}$.
For a given set of $M\geq1$ design points in $\mathcal{Z}$, $z_{1},...,z_{M}$,
eq. (\ref{eq: Projector def}) then becomes
\begin{equation}
\bar{\Pi}_{K}\left(v\right)\left(z\right)=B_{K}\left(z\right)^{\prime}\left[\sum_{i=1}^{M}B_{K}\left(z_{i}\right)B_{K}\left(z_{i}\right)^{\prime}\right]^{-1}\sum_{i=1}^{M}B_{K}\left(z_{i}\right)v\left(z_{i}\right).\label{eq: Pi least-squares}
\end{equation}
The design points may either be random or deterministic and can be
chosen relative to the basis functions to ensure that $\bar{\Pi}_{K}$
is easy to compute and provides a good approximation for a broad class
of functions. The performance of most function approximations will
depend on the smoothness of the function of interest. 

A standard smooth function class often considered in approximation
theory is the following: For any vector $a=\left(a_{1},...,a_{d_{Z}}\right)\in\mathbb{N}_{0}^{d_{X}}$,
let $D^{a}f\left(x\right)=\partial^{\left|a\right|}f\left(x\right)/\left(\partial x_{1}^{\alpha_{1}}\cdots\partial x_{d_{Z}}^{a_{d_{z}}}\right)$,
where $\left|a\right|=a_{1}+\cdots+a_{d_{z}}$, be the corresponding
partial derivative. For $\alpha>0$, let $\underline{\alpha}\geq0$
be the greatest integer smaller than $\alpha$. For any $\underline{\alpha}$
times differentiable function $f\left(x\right)$, we then define
\begin{equation}
\Vert f\Vert_{\alpha,\infty}=\max_{\left|a\right|\leq\underline{\alpha}}\Vert D^{a}f\Vert_{\infty}+\max_{\left|a\right|=\underline{\alpha}}\sup_{x_{1}\neq x_{2}}\frac{\left|D^{a}f\left(x_{1}\right)-D^{a}f\left(x_{2}\right)\right|}{\left\Vert x_{1}-x_{2}\right\Vert ^{\alpha-\underline{\alpha}}},\label{eq: hoelder norm}
\end{equation}
and let $\boldsymbol{\mathbb{C}}_{r}^{\alpha}\left(\mathcal{X}\right)$
be the space of all $\underline{\alpha}\geq0$ times continuously
differentiable functions $f:\mathcal{X}\mapsto\mathbb{R}$ with $\Vert f\Vert_{\alpha,\infty}<r$.
Due to smoothing, $v_{N}\in\boldsymbol{\mathbb{C}}_{r}^{\alpha}\left(\mathcal{Z}\right)$,
for some $r<\infty$, if $u$ and $F_{Z}$ are sufficiently smooth
(see Theorem \ref{Thm: contraction}). We can therefore import existing
results for approximation methods for functions in $\boldsymbol{\mathbb{C}}_{r}^{\alpha}\left(\mathcal{Z}\right)$:
\begin{example}
Polynomial interpolation using tensor products. Suppose we use $J$th
order Chebyshev interpolation with $M\geq J$ nodes in each of the
$d_{z}$ dimensions, or a $J$th order B-spline interpolation with
$M\geq J$ number of nodes in each of the $d_{z}$ dimensions (see
Appendix \ref{sec:Sieve-spaces} for their precise expressions). Let
$p_{1},...,p_{J}$ denote the $J$ polynomials; we then have
\[
B_{K}\left(z\right)=\left\{ p_{j_{1}}\left(z_{1}\right)\cdots p_{j_{d_{z}}}\left(z_{d_{z}}\right):j_{1},...,j_{d_{z}}=1,...,J\right\} ,
\]
 which is of dimension $K=J^{d_{Z}}$. Choosing $J\geq\underline{\alpha}$,
where $\underline{\alpha}\geq1$ denotes the number of derivatives
of $v\left(z\right)$, both interpolation schemes satisfy, for any
radius $r<\infty$,
\[
\sup_{v\in\boldsymbol{\mathbb{C}}_{r}^{\alpha}\left(\mathcal{Z}\right)}\left\Vert \bar{\Pi}_{K}\left(v\right)-v\right\Vert _{\infty}=O\left(\frac{\log\left(J\right)}{J^{\alpha}}\right)=O\left(\frac{\log\left(K\right)}{K^{\alpha/d}}\right);
\]
see p.14 in  \citet{Rivlin1990} for Chebyshev interpolation and \citet{Schumaker2007}
for B-splines. If $v\left(z\right)$ is analytic ($\underline{\alpha}=\infty$),
the above result holds for any (large) $J,\alpha<\infty$.
\end{example}
As can be seen from the above example, standard polynomial tensor
product approximations suffer from the well-known computational curse
of dimensionality: To reach a given level of error tolerance, the
total number of basis functions $K$ has to grow exponentially as
$d_{Z}$ increases. This issue can be partially resolved by using
more advanced function approximation methods:
\begin{example}
Interpolation with sparse grids. Instead of using tensor-product basis
functions to approximate a given function, where the total number
of basis function and interpolation points will have to grow exponentially
with $d_{Z}$ to control the approximation error, one can instead
use so-called Smolyak sparse grids; see, e.g., \citet{Juddetal2014}
and \citet{Brumm&Scheidegger2017}. Using these, the number of grid
points needed to obtain a given error tolerance are reduced from $O\left(M^{d_{Z}}\right)$
to $O\left(M\left(\log M\right)^{d_{Z}}\right)$ with only slightly
deteriorated accuracy. 
\end{example}
\begin{example}
Variable selection, shape constraints, shrinkage estimators, and machine
learning. An alternative way of breaking the curse of dimensionality
appearing in Example 2 is to select the basis functions judiciously.
This could, for example, be done using standard variable selection
methods; one example of this approach can be found in \citet{Chen1999}.
Alternatively, one can in some cases show that the value functions
satisfy certain shape constraints that can then be imposed on the
sieve; see, for example, \citet{Cai&Judd2013}. Other automated selection
methods include shrinkage methods where a penalization term is added
to the least-squares criterion. Again this leads to a more sparse
representation which is able to break the curse-of-dimensionality.
Finally, machine learning algorithms, such as neural networks, may
potentially be useful in approximating the value functions; see, for
example, \citet{Chen&White1999}. On the other hand, these methods
are generally computationally more expensive compared to the least-squares
projection method in (\ref{eq: Pi least-squares}) and require that
the value function satisfies certain sparsity. We will investigate
the performance of such more advanced projection operators in future
work.
\end{example}
As noted earlier, there is no guarantee that a given function approximator
is non-expansive. But this can, in principle, be examined numerically
for a given choice of $\bar{\Pi}_{K}$. For the least-squares projection,
this amounts to solving, for a given choice of basis functions and
grid points,
\begin{equation}
\left\Vert \bar{\Pi}_{K}\right\Vert _{op,\infty}=\sup_{v\in\mathbb{R}^{M},\left\Vert v\right\Vert =1}\sup_{z\in\mathit{\mathcal{Z}}}\left|B_{K}\left(z\right)^{\prime}\left[\sum_{i=1}^{M}B_{K}\left(z_{i}\right)B_{K}\left(z_{i}\right)^{\prime}\right]^{-1}\sum_{i=1}^{M}B_{K}\left(z_{i}\right)v_{i}\right|.\label{eq: Pi-norm}
\end{equation}
When $M$ and/or $\dim\mathcal{Z}$ is large this may be computationally
demanding and instead one can obtain a lower bound by restricting
$z$ to only take values on the chosen set of grid points: With $\mathbf{B}_{K,M}\in\mathbb{R}^{K\times M}$containing
the basis functions evaluated at the grid points, we can represent
$\bar{\Pi}_{K}$ when only evaluated at chosen grid points $z_{1},....,z_{M}$
in terms of
\[
\mathbf{P}_{K,M}=\mathbf{B}_{K,M}'\left[\mathbf{B}_{K,M}\mathbf{B}_{K,M}'\right]^{-1}\mathbf{B}_{K,M}\in\mathbb{R}^{M\times M}.
\]
In particular, it is easily checked that with the supremum in (\ref{eq: Pi-norm})
being only taken over $z\in\left\{ z_{1},...,z_{M}\right\} $, $\left\Vert \bar{\Pi}_{K}\right\Vert _{op,\infty}=\left\Vert \mathbf{\mathbf{P}_{K,M}}\right\Vert _{op,\infty}$.
Furthermore, $\left\Vert \mathbf{P}_{K,M}\right\Vert _{op,\infty}\leq1$
if and only if
\[
\max_{i=1,...,M}\sum_{j=1}^{M}\left|p_{ij}\right|\leq1,
\]
where $p_{ij}$ is the $\left(i,j\right)$th element of $\mathbf{P}$,
c.f. \citet{Lizotte2011}.

\subsubsection{Solving for the approximate value functions}

Computing the simulated self-approximating solution or the projection-based
one can be done using three different numerical algorithms: Successive
approximation (SA), Newton-Kantorovich (NK), or a combination of the
two. The latter corresponds to the hybrid solution method proposed
in \citet{Rust1988Siam}. We here discuss the implementation of these
algorithms with focus on the sieve-based approximation of $v$; the
implementations of the sieve approximation of $V$ and the self-approximating
solutions of either of the two follow along the same lines. The main
difference between solving for $V$ or $v$ is that the latter involves
smaller computational burden since it is a scalar function while the
former is a $D$-dimensional vector function. 

SA utilizes that (for $K$ chosen large enough), $\bar{\Pi}_{K}\bar{\Gamma}_{N,\lambda}$,
is a contraction mapping which guarantees that the following algorithm
will converge towards the solution to (\ref{eq: self-approx V}),
\begin{equation}
\hat{v}_{N}^{\left(k\right)}=\bar{\Pi}_{K}\bar{\Gamma}_{N}(\hat{v}_{N}^{\left(k-1\right)}),\label{eq: V-hat succesive}
\end{equation}
for $k=1,2,...$, given some initial guess $\hat{v}_{N}^{\left(0\right)}$.
In the leading case of (\ref{eq: Pi least-squares}), this can be
expressed as a sequence of least-squares problems that are easily
computed: $\hat{v}_{N}^{\left(k\right)}\left(z\right)=\hat{\alpha}_{k}^{\prime}B_{K}\left(z\right)$
where
\[
\hat{\alpha}_{k}=\left[\sum_{i=1}^{M}B_{K}\left(z_{i}\right)B_{K}\left(z_{i}\right)^{\prime}\right]^{-1}\sum_{i=1}^{M}B_{K}\left(z_{i}\right)\bar{\Gamma}_{N,\lambda}(\hat{\alpha}_{k-1}^{\prime}B_{K})\left(z_{i}\right)^{\prime}\in\mathbb{R}^{K},
\]
for $k=1,2,...$, given some initial guess $\hat{\alpha}_{0}$. In
the case where $\varepsilon_{t}=\emptyset$ or when the model is on
the form eq. (\ref{eq: model no bias}) with $\varepsilon_{t}^{\text{(1}}$
being extreme-valued distributed and $\varepsilon_{t}^{\left(-1\right)}=\emptyset$,
\[
\bar{\Gamma}_{N}(\alpha^{\prime}B_{K})(z_{i};\lambda)=G_{\lambda}\left(u\left(z\right)+\beta\alpha^{\prime}\sum_{j=1}^{N}B_{K}\left(Z_{j}\left(z_{i}\right)\right)w_{Z,N,j}\left(z_{i}\right)\right),
\]
and so $\sum_{j=1}^{N}B_{K}\left(Z_{j}\left(z_{i}\right)\right)w_{z,N,j}\left(z_{i}\right)$,
$i=1,...,M,$ only need to be computed once and then recycled in each
iteration; in contrast, the simulated averages appearing in $\Gamma_{N}(\hat{\alpha}_{k-1}'B_{K})\left(z_{i}\right)$,
$i=1,...,M$, have to be recomputed in each step of the SA algorithm.
Thus, in this special case, it is faster to (approximately) solve
for $v_{N,\lambda}$ instead of $V_{N,\lambda}$. While SA is guaranteed
to converge globally when $\bar{\Pi}_{K}\bar{\Gamma}_{N}$ is a contraction,
the rate of convergence will be slow with the error vanishing at rate
$\beta^{k}$,
\begin{equation}
\left\Vert \hat{v}_{N}^{\left(k\right)}-v_{N}\right\Vert _{\infty}\leq\frac{\beta^{k}\left(1+\beta\right)}{1-\beta}\left\Vert \hat{v}_{N}^{\left(0\right)}-v_{N}\right\Vert _{\infty}.\label{eq: SA conv rate}
\end{equation}

To speed up convergence, we therefore follow \citet{Rust1988Siam}
and combine SA with NK iterations since NK converges with a quadratic
rate once a given guess of the value function is close enough to the
fixed point. Moreover, in situations where $\bar{\Pi}_{K}\bar{\Gamma}_{N}$
is expansive, NK is still guaranteed to converge locally. Since both
the self-approximating and sieve-based methods solve finite-dimensional
problems, the NK algorithm for these are equivalent to Newton's method.
First consider the sieve-based method where we focus on the least-squares
projection as given in (\ref{eq: Pi least-squares}). We are then
seeking $\hat{\alpha}$ solving the following $K$ equations,
\[
\bar{S}_{N,K}\left(\alpha;\lambda\right)=0,
\]
with
\[
\bar{S}_{N,K}\left(\alpha;\lambda\right)=\alpha-\left[\sum_{i=1}^{M}B_{K}\left(z_{i}\right)B_{K}\left(z_{i}\right)^{\prime}\right]^{-1}\sum_{i=1}^{M}B_{K}\left(z_{i}\right)\bar{\Gamma}_{N}\left(\alpha^{\prime}B_{K}\right)\left(z_{i};\lambda\right).
\]
The corresponding derivatives of the left-hand side as a function
w.r.t. $\alpha$ can be expressed in terms of the Hadamard differential
of $\bar{\Gamma}_{N}$ w.r.t. $v$,
\begin{eqnarray*}
\nabla\bar{\Gamma}_{N}(v)\left[dv\right](z;\lambda) & = & \beta\sum_{d\in\mathcal{D}}\sum_{j=1}^{N}\dot{G}_{d,\lambda}\left(u\left(z,\varepsilon_{j}\left(z\right)\right)+\beta\sum_{i=1}^{N}v\left(Z_{i}\left(z\right);\lambda\right)w_{Z,N,i}\left(z\right)\right)\\
 & \times & \left(\sum_{k=1}^{N}dv\left(Z_{k}\left(z,d\right);\lambda\right)w_{Z,N,k}\left(z,d\right)\right)w_{\varepsilon,N,j}\left(z\right),
\end{eqnarray*}
where $dv:\mathcal{Z}\times\left[0,\bar{\lambda}\right]\mapsto\mathbb{R}$
is the direction and
\begin{equation}
\dot{G}_{\lambda,d}^{\left(r\right)}(r)=\frac{\partial G_{\lambda}(r)}{\partial r\left(d\right)}=\frac{\exp\left(\frac{r\left(d\right)}{\lambda}\right)}{\sum_{d^{\prime}\in\mathcal{D}}\exp\left(\frac{r\left(d^{\prime}\right)}{\lambda}\right)}.\label{eq: G r-deriv}
\end{equation}
The partial derivatives of $\bar{S}_{N,K}\left(\alpha;\lambda\right)$
then becomes
\[
\bar{H}_{N,K}\left(\alpha;\lambda\right)=I_{K}-\left[\sum_{i=1}^{M}B_{K}\left(z_{i}\right)B_{K}\left(z_{i}\right)^{\prime}\right]^{-1}\sum_{i=1}^{M}B_{K}\left(z_{i}\right)\nabla\bar{\Gamma}_{N}\left(\alpha^{\prime}B_{K}\right)\left[B_{K}\right]\left(z_{i};\lambda\right)^{\prime}\in\mathbb{R}^{K\times K}.
\]
With these definitions, the NK algorithm takes the form
\[
\hat{\alpha}_{k}=\hat{\alpha}_{k-1}-\bar{H}_{N,K}^{-1}\left(\hat{\alpha}_{k-1};\lambda\right)\bar{S}_{N,K}\left(\hat{\alpha}_{k-1};\lambda\right).
\]

The NK algorithm for the self-approximating method is on the same
form, except that we now solve directly for the value function at
the $N$ draws. With slight abuse of notation, let $v_{N}=\left\{ v_{N,\lambda}\left(Z_{k};\lambda\right):k=1,...,N\right\} $
be the vector of integrated values across the set of draws solving
$\bar{S}_{N}\left(v_{N};\lambda\right)=0$ where
\begin{equation}
\bar{S}_{N,k}\left(v_{N};\lambda\right)=v_{N,k}-\sum_{j=1}^{N}G_{\lambda}\left(u(Z_{k},\varepsilon_{j})+\beta\sum_{i=1}^{N}v_{N,i}w_{z,N,i}\left(Z_{k}\right)\right)w_{\varepsilon,N,j}\left(Z_{k}\right),\label{eq: self-approx v-2}
\end{equation}
for $k=1,...,N$. The corresponding derivatives is $\bar{H}_{N}\left(\alpha;\lambda\right)=\left(\bar{H}_{N,1}\left(\alpha;\lambda\right),...,\bar{H}_{N,N}\left(\alpha;\lambda\right)\right)^{\prime}\in\mathbb{R}^{N\times N}$
where, with $\mathbf{1}_{N}=\left(1,...,1\right)^{\prime}\in\mathbb{R}^{N}$,
\[
\bar{H}_{N,k}\left(\alpha;\lambda\right)=I_{N}-\nabla\bar{\Gamma}_{N}\left(v_{N}\right)\left[\mathbf{1}_{N}\right]\left(z_{k};\lambda\right)\in\mathbb{R}^{N}.
\]
Finally, we note that the NK algorithm for the expected value function
takes a similar form with the functional differential of $\Gamma_{N}$
given by
\begin{equation}
\nabla\Gamma_{N}(V)\left[dV\right](z;\lambda)=\beta\sum_{d\in\mathcal{D}}\sum_{i=1}^{N}\dot{G}_{\lambda,d}^{\left(r\right)}(r)\left(u\left(S_{i}\left(z\right)\right)+\beta V\left(Z_{i}\left(z\right);\lambda\right)\right)dV\left(Z_{i}\left(z\right),d\right)w_{N,i}\left(z\right),\label{eq: dGamma_N def}
\end{equation}
where $dV\left(z\right)=\left(dV\left(z,1\right),...,dV\left(z,D\right)\right)^{\prime}.$

Comparing the NK algorithm for the self-approximating and the sieve-based
method, we note that the former involves inverting a $N\times N$-matrix
while the latter a $K\times K$-matrix. As pointed out earlier, the
self-approximating method generally needs $N$ to be chosen quite
large to achieve a precise simulated version of the Bellman operator,
in particular in higher dimensions, and so the NK algorithm for this
method may become numerically infeasible in some cases. While the
projection-based method also suffers from a curse of dimensionality,
since the number of basis functions, $K$, has to be quite large in
higher dimensions to achieve a reasonable approximation, it is less
severe and is implementable for higher-dimensional models. If more
advanced function approximation methods are employed, even better
performance can be achieved.

\section{Theory\label{sec:Convergence}}

We here develop an asymptotic theory for the self-approximating and
sieve-based methods. We first establish some important properties
of the smoothed simulated Bellman operators and their exact solutions,
$v_{N}$ and $V_{N}$ defined in (\ref{eq: Sim Fixed point}). These
are then used in the asymptotic analysis of the self-approximating
solution method and the sieve-based one. This analysis will rely on
two general results for estimated solutions to fixed point problems
as stated in Theorems \ref{thm: General rate} and \ref{thm: general dist}
in the appendix. The asymptotic analysis will mostly focus on $V_{N}$
and $\hat{V}_{N}$ since our results for these easily translate into
similar results for the approximate integrated value function. For
example, $v_{N}(z;\lambda)=M_{N,u}\left(\beta V_{N}\left(z;\lambda\right)|z;\lambda\right)$,
where
\[
M_{N,u}\left(r|z;\lambda\right)=\sum_{j=1}^{N}G_{\lambda}\left(u(z,\varepsilon_{j}\left(z\right))+r\right)w_{\varepsilon,N,j}\left(z\right),
\]
and so the asymptotic results for $V_{N}$ in conjunction with the
functional Delta method can be used to obtain similar results for
$v_{N}$. 

Without loss of generality, we assume that the draws can be written
as
\begin{equation}
Z_{i}\left(z\right)=\psi_{Z}\left(U_{i};z\right)\in\mathcal{Z},\:\varepsilon_{i}\left(z\right)=\psi_{\varepsilon}\left(U_{i};z\right)\in\mathcal{E},\label{eq: psi def}
\end{equation}
for some i.i.d. draws $U_{i}\sim P_{U}$, $i=1,...,N$, and some functions
$\psi_{Z}$ and $\psi_{\varepsilon}$. We then define, with $\psi=\left(\psi_{Z},\psi_{\varepsilon}\right)$
and for any given function $V\left(z,\lambda\right)$,
\begin{equation}
u_{\psi}\left(U;z\right):=u\left(\psi\left(U;z\right)\right),\:w_{\psi}\left(U;z\right)=w\left(\psi\left(U;z\right)|z\right),\:V_{\psi}(U;z,\lambda)=V\left(\psi_{Z}\left(u;z\right),\lambda\right)\label{eq: mu def}
\end{equation}
 so that
\begin{equation}
\Gamma(V)(z,\lambda)=E\left[G_{\lambda}\left(u_{\psi}\left(U;z\right)+\beta V_{\psi}(U;z,\lambda)\right)w_{\psi}\left(U;z\right)\right]\label{eq: generalized Gamma}
\end{equation}
where expectations are taking over $U\sim P_{U}$, and
\begin{equation}
\Gamma_{N}(V)(z,\lambda)=\sum_{i=1}^{N}G_{\lambda}\left(u_{\psi}\left(U_{i};z\right)+\beta V_{\psi}(U_{i};z,\lambda)\right)w_{\psi,N}\left(U_{i};z\right),\label{eq: generalised Gamma_N}
\end{equation}
where $w_{\psi,N}\left(U_{i};z,d\right)=w_{\psi}\left(U_{i};z,d\right)/\sum_{j=1}^{N}w_{\psi}\left(U_{j};z,d\right).$
Here, and in the following, we let $V_{0}$$\left(z,\lambda\right)$
denote the exact solution to 
\begin{equation}
V_{0}\left(z,\lambda\right)=\Gamma\left(V_{0}\right)\left(z,\lambda\right)\label{eq: V_0 def}
\end{equation}
As explained earlier, we here define the two operators to take a given
function $V\left(z,\lambda\right)$, $\left(z,\lambda\right)\in\mathcal{Z}\times\left[0,\bar{\lambda}\right]$,
for some given $\bar{\lambda}>0$, and map them into another function
with domain $\mathcal{Z}\times\left[0,\bar{\lambda}\right]$. In particular,
$V_{0}$$\left(z,0\right)$ and $V_{N}\left(z,0\right)$ are the non-smoothed
($\lambda=0$) exact and simulated solutions, respectively. We then
impose the following regularity conditions on the model and chosen
importance sampler, where we recall the function norm defined in eq.
(\ref{eq: hoelder norm}) and the function set $\boldsymbol{\mathbb{C}}_{r}^{\alpha}\left(\mathcal{Z}\right)$
defined below this equation:
\begin{assumption}
\label{assu: mom bound}The support $\mathcal{Z}$ is a compact set;
$\bar{u}_{\psi}\left(u\right):=\sup_{z\in\mathcal{Z}}\left\Vert u_{\psi}\left(u;z\right)\right\Vert $
and $\bar{w}_{\psi}\left(u\right):=\sup_{z\in\mathcal{Z}}w_{\psi}\left(u;z\right)$
satisfy $E\left[\bar{u}_{\psi}^{2}\left(U\right)\bar{w}_{\psi}^{2}\left(U\right)\right]<\infty$. 
\end{assumption}
\begin{assumption}
\label{Ass: smoothness}For some $\alpha>0$, $z\mapsto u_{\psi}\left(U;\cdot\right)$
and $z\mapsto w_{\psi}\left(U;z\right)$ belong to $\boldsymbol{\mathbb{C}}_{\infty}^{\alpha}\left(\mathcal{Z}\right)$
$P_{U}$-almost surely with $E\left[\Vert u_{\psi}\left(U;\cdot\right)\Vert_{\alpha,\infty}^{2}\Vert w_{\psi}\left(U;\cdot\right)\Vert_{\alpha,\infty}^{2}\right]<\infty$.
\end{assumption}
Assumption \ref{assu: mom bound} share some similarities with the
regularity conditions found in \citet{Rust1988Siam} who considered
an additive version of our general model. Importantly, we only require
that $Z_{t}$ has bounded support while $\varepsilon_{t}$ can have
potentially unbounded support. This is in contrast to \citealp{Pal&Stacurski2013}
and \citealp{Rust1997curse} who require both components to be bounded.
We conjecture that the subsequent results can be generalized to also
hold in the case of $\mathcal{Z}$ unbounded but then our conditions
and arguments would have to be changed. For example, the existence
of unique fixed points would have to be verified in a function space
equipped with a weighted $\sup$-norm and with additional moment conditions
on $Z_{t}$, see, e.g., \citet{Norets2010differentiability}. Similarly,
our empirical process results would need to be established using bracketing
conditions with weighted norms and additional moment conditions, see,
e.g., Section 2.10.4 in \citealp{VW}. Similar to the results in \citet{Rust1988Siam},
Assumption \ref{assu: mom bound} implies $\Gamma(V)\in\boldsymbol{\mathbb{B}}\left(\mathcal{Z}\times\left(0,\bar{\lambda}\right)\right)^{D}$
for all $V\in\mathcal{\boldsymbol{\mathbb{B}}}\left(\mathcal{Z}\times\left(0,\bar{\lambda}\right)\right)^{D}$,
see below. This particular result actually holds under the weaker
requirement that $E\left[\Vert u_{\psi}\left(U;\cdot\right)\Vert_{0,\infty}\Vert w_{\psi}\left(U;\cdot\right)\Vert_{,\infty}\right]<\infty$
but the existence of the second order moment is needed for the subsequent
asymptotic analysis of $V_{N}\left(z;\lambda\right)$ and so we impose
this restriction throughout. 

Assumption \ref{Ass: smoothness} impose smoothness conditions on
the model and the chosen samplers in terms of the state variables
$Z_{t}$. These conditions imply that the expected and integrated
value functions will be smooth too. Note here that the degree of smoothness
$\alpha$ is left unrestricted at this stage and so the functions
are not required to be differentiable, merely Lipschitz. While the
focus is on models with continuous state variables our theory also
covers models with discrete state space. In this case, we can dispense
of Assumption \ref{Ass: smoothness} and instead rely on the more
general Theorem \ref{thm: Master} which implies that the subsequent
results still go through when $Z$$_{t}$ is discrete.

We first establish existence and uniqueness of the (generally) infeasible
simulated solutions and show that they inherit the smoothness properties
of $z\mapsto u_{\psi}\left(U;\cdot\right)$ and $z\mapsto w_{\psi}\left(U;z\right)$.
This feature of the approximate solutions is important for two reasons:
First, it allows us to show uniform convergence of certain functionals
as part of our proof of weak convergence. Second, we can control the
approximation error due to the use of sieves later on.
\begin{thm}
\label{Thm: contraction}Suppose Assumption \ref{assu: mom bound}
holds and, for a given $N\geq1$, $\inf_{z\in\mathcal{Z}}\sum_{i=1}^{N}w_{\psi}\left(U_{i};z\right)>0$.
Then the operators $\Gamma$ and $\Gamma_{N}$ in eqs. (\ref{eq: generalized Gamma})
and (\ref{eq: generalised Gamma_N}) are almost surely contraction
mappings on $\boldsymbol{\mathbb{B}}\left(\mathcal{Z}\times\left(0,\bar{\lambda}\right)\right)^{D}$
and so $V_{0}:\mathcal{Z}\times\left(0,\bar{\lambda}\right)\mapsto\mathbb{R}^{D}$
and $V_{N}:\mathcal{Z}\times\left(0,\bar{\lambda}\right)\mapsto\mathbb{R}^{D}$
exist and are unique. If furthermore Assumption \ref{Ass: smoothness}
holds then $V_{0}\in\boldsymbol{\mathbb{C}}_{r_{0}}^{\alpha}\left(\mathcal{Z}\times\left(0,\bar{\lambda}\right)\right)$
for some constant $r_{0}<\infty$ while $V_{N}\in\boldsymbol{\mathbb{C}}_{r_{N}}^{\alpha}\left(\mathcal{Z}\times\left(0,\bar{\lambda}\right)\right)$
for some $r_{N}<\infty$ $P_{U}$-almost surely.
\end{thm}
The above is a fixed $N$ result with the bound on $V_{N}$, $r_{N}$,
being random since it depends on the particular set of draws. We derive
a deterministic bound on $r_{N}$ as $N\rightarrow\infty$ below.
The condition $\inf_{z\in\mathcal{Z}}\sum_{i=1}^{N}w_{\psi}\left(U_{i};z\right)>0$
will hold with probability approaching 1 (w.p.a.1) as $N\rightarrow\infty$
and so can be dropped in our asymptotic analysis. Next, we analyze
the effect of smoothing on the exact and simulated value function:
\begin{thm}
\label{thm: smoothing}Under the conditions of Theorem \ref{Thm: contraction},
the following hold: $\left\Vert V_{0}\left(\cdot;\lambda\right)-V_{0}\left(\cdot;0\right)\right\Vert _{\infty}=O\left(\lambda\right)$
and $\left\Vert V_{N}\left(\cdot;\lambda\right)-V_{N}\left(\cdot;0\right)\right\Vert _{\infty}=O_{P}\left(\lambda\right)$
for any given $N\geq1.$
\end{thm}
This shows that the smoothing can be controlled for by suitable choice
of $\lambda$ both asymptotically ($N=+\infty$) and for any finite
number of simulations ($N<\infty$). Also note that the above result
holds independently of the smoothness properties of the unsmoothed
exact and simulated solutions.

\subsection{Self-approximating method}

In this section, we provide an analysis of the smoothed fixed point,
$V_{N}$, to $\Gamma_{N}$ defined in eq. (\ref{eq: generalised Gamma_N})
thereby allowing for general importance samplers. As a special case,
we obtain an asymptotic theory for the self-approximating solution
method (where $\Phi_{z}\left(z^{\prime}|z,d\right)$ is restricted
to be marginal distribution). The results for $V_{N}$ will then in
turn be used in the analysis of the corresponding sieve-based methods
in the next section.
\begin{thm}
\label{Thm: Rust approx rate}Suppose Assumptions \ref{assu: mom bound}-\ref{Ass: smoothness}
hold for some $\alpha>0$. Then $V_{N}$ solving $\Gamma_{N}(V_{N})=V_{N}$
satisfies $\Vert V_{N}-V_{0}\Vert_{\infty}=O_{P}(1/\sqrt{N})$. 

If Assumption \ref{Ass: smoothness} holds with $\alpha\geq1$, then
$V_{0},V_{N}\in\boldsymbol{\mathbb{C}}_{r}^{1}\left(\mathcal{Z}\times\left[0,\bar{\lambda}\right]\right)$
w.p.a.1 for some constant $r<\infty$. Moreover, $\sup_{z\in\mathcal{Z}}\left\Vert \partial V_{N}\left(z,\lambda\right)/\left(\partial z_{j}\right)-\partial V_{0}\left(z,\lambda\right)/\left(\partial z_{j}\right)\right\Vert =O_{P}\left(\sqrt{N}/\lambda\right)$
uniformly over $\lambda\in$$\left(0,\bar{\lambda}\right)$.
\end{thm}
The first part of this theorem is similar to results found in \citet{Rust1997curse}
and \citet{Pal&Stacurski2013} who also show $\sqrt{N}$-convergence
of their value function approximation. Importantly, the convergence
result holds uniformly over the smoothing parameter $\lambda$ and
so there is no first-order effect from smoothing if $\lambda$ vanishes
sufficiently fast. Specifically, for any sequence $\lambda_{N}$ satisfying
$\sqrt{N}\lambda_{N}\rightarrow0$, Theorems \ref{thm: smoothing}
and \ref{Thm: Rust approx rate} yield $\sup_{z\in\mathcal{Z}}\left\Vert V_{N}\left(z,\lambda_{N}\right)-V\left(z,0\right)\right\Vert =O_{p}(1/\sqrt{N})$.
This is similar to convergence of smoothed empirical cdf where the
indicator function is replaced by a smoothed version; this also does
not affect the convergence rate as long as the smoothing bias is controlled
for. The second part of the theorem appears to be a new result and
shows that if the problem is smooth enough, the first-order partial
derivatives of $V_{N}\left(z,\lambda\right)$ also converge uniformly
over $z$ with rate $\sqrt{N}/\lambda$. Since we need $\lambda\rightarrow0$
to kill the smoothing bias, this could seem to imply that the first-order
derivatives converge with slower than $\sqrt{N}$-rate. However, we
conjecture that the derived rate is not sharp and that $\sqrt{N}$-convergence
does actually hold. The proof of this appears to require a more delicate
and refined arguments, however, and so we leave this for future research.

The above result is then in turn used to derive the asymptotic distribution
of $V_{N}\left(z,\lambda\right)$ uniformly in $\left(z,\lambda\right)\in\mathcal{Z}\times\left(0,\bar{\lambda}\right)$.
Here, the smoothing proves important since it allows us to generalize
the standard arguments used in the analysis of finite-dimensional
extremum estimators to our setting: We first expand the ``first-order
condition'', $V_{N}-\Gamma_{N}(V_{N})=0$, around $V_{0}=\Gamma(V_{0})$
to obtain, with $\nabla\Gamma_{N}$ defined in (\ref{eq: dGamma_N def}),
\begin{eqnarray*}
0 & = & \Gamma(V_{0})-\Gamma_{N}(V_{0})+\left\{ I-\nabla\Gamma_{N}(V_{0})\right\} \left[V_{N}-V_{0}\right]+o_{P}\left(1/\sqrt{N}\right),
\end{eqnarray*}
where the rate of the remainder term follows from Theorem \ref{Thm: Rust approx rate}.
Next, employing empirical process theory, we show that $\sqrt{N}\left\{ \Gamma(V_{0})-\Gamma_{N}(V_{0})\right\} \rightsquigarrow\mathbb{G}$
in $\boldsymbol{\mathbb{B}}\left(\mathcal{Z}\times\left(0,\bar{\lambda}\right)\right)^{D}$
for a Gaussian process $\mathbb{G}\left(z,\lambda\right)$ with covariance
kernel
\begin{equation}
\Omega\left(z_{1},\lambda_{1},z_{2},\lambda_{2}\right)=E_{U}\left[g\left(U;z_{1},\lambda_{1}\right)g\left(U;z_{2},\lambda_{2}\right)^{\prime}\right],\label{eq: Omega def}
\end{equation}
\begin{equation}
g\left(U;z,\lambda\right)=\left\{ G_{\lambda}\left(u_{\psi}\left(U;z\right)+\beta V_{0}\left(\psi_{Z}\left(U;z\right),\lambda\right)\right)-\Gamma(V_{0})\left(z,\lambda\right)\right\} w_{\psi}\left(U;z\right).\label{eq: g def}
\end{equation}
Finally, we show that $\nabla\Gamma_{N}(V_{0})\left[dV\right]\rightarrow^{P}\nabla\Gamma(V_{0})\left[dV\right]$
uniformly over $dV\in\boldsymbol{\mathbb{C}}_{2r}^{1}\left(\mathcal{Z}\times\left[0,\bar{\lambda}\right]\right)^{D}$
with $r<\infty$ given in Theorem \ref{Thm: Rust approx rate}. Since
$V_{N}-V_{0}\boldsymbol{\in\mathbb{C}}_{2r}^{1}\left(\mathcal{Z}\times\left[0,\bar{\lambda}\right]\right)^{D}$,
we conclude that:
\begin{thm}
\label{Thm: Rust normal}Suppose Assumptions \ref{assu: mom bound}
and \ref{Ass: smoothness} hold with $\alpha\geq1$. Then, $\sqrt{N}\{V_{N}-V_{0}\}\rightsquigarrow\mathbb{G}_{V}$
on $\boldsymbol{\mathbb{B}}\left(\mathcal{Z}\times\left(0,\bar{\lambda}\right)\right)^{D}$
where $\mathbb{G}_{V}\left(z,\lambda\right)=\left\{ I-\nabla\Gamma(V_{0})\right\} ^{-1}\left[\mathbb{G}\right]\left(z,\lambda\right)$
is a $D$-dimensional Gaussian process.
\end{thm}
The above result implies, for example, that $\sqrt{N}\left\{ V_{N}\left(z,\lambda\right)-V_{0}\left(z,\lambda\right)\right\} \rightarrow^{d}N\left(,\Omega_{V}\left(z,\lambda,z,\lambda\right)/N\right)$
as $N\rightarrow\infty$ for any given $\left(z,\lambda\right)$,
where
\[
\Omega_{V}\left(z_{1},\lambda_{1},z_{2},\lambda_{2}\right)=\int\int r^{*}\left(z_{1}^{\prime},\lambda_{1}^{\prime}|z_{1},\lambda_{1}\right)\Omega\left(z_{1}^{\prime},\lambda_{1}^{\prime},z_{2},\lambda_{2}\right)r^{*}\left(z_{2}^{\prime},\lambda_{2}^{\prime}|z_{2},\lambda_{2}\right)^{\prime}d\left(z_{1}^{\prime},\lambda_{1}^{\prime}\right)d\left(z_{2}^{\prime},\lambda_{2}^{\prime}\right),
\]
where $r^{*}$ is the Riesz representer of $dV\mapsto\left\{ I-\nabla\Gamma(V_{0})\right\} ^{-1}\left[dV\right]\left(z,\lambda\right)$.
Thus, it allows us to construct (pointwise or uniform) confidence
bands for the expected value function. We expect that the result will
also be useful in analyzing the impact of value function approximation
when used in estimation. This could be done by combining the above
weak convergence result with, e.g., the results for approximate estimators
found in \citet{Kristensen&Salanie2017}.

The proof of Theorem \ref{Thm: Rust normal} proceeds by verifying
the two high-level conditions of the ``master'' Theorem \ref{thm: Master}
where the same weak convergence result is obtained under more general
conditions. Theorem \ref{thm: Master} allows us to replace the smoothness
conditions in Assumption \ref{Ass: smoothness} with some other conditions
implying that $V_{N}-V_{0}$ is situated in a function set with finite
entropy. One example would be to impose restrictions on $u$ and $w$
so that the value function and its estimator are both monotone functions,
c.f. \citealp{Pal&Stacurski2013}, in which case we could then appeal
to Theorem 2.7.5 in \citealp{VW} to obtain the results of Theorem
\ref{Thm: Rust normal}.

We conjecture that a similar weak convergence result will hold for
the non-smoothed value function approximation ($\lambda=0$). However,
the proof of such a result would require different arguments and seemingly
stronger assumptions. In particular, the current proof only requires
the empirical process $\left(z,\lambda\right)\mapsto\Gamma_{N}(V_{0})\left(z,\lambda\right)$
to converge weakly. To allow for non-smooth value function approximation,
we conjecture that we would now need to show that the empirical process
$\left(V,z\right)\mapsto\Gamma_{N}(V)\left(z,0\right)$ converges
weakly over a suitable function set that the estimated non-smooth
solution, $V_{N}$, would be situated in. For this to hold, the uniform
entropy of the function set would need to be finite. Standard choices
of function sets are smooth classes, but $V_{N}$ and its limit $V_{0}$
are both non-smooth now and so the proof appears to be rather delicate.

Finally, for a complete analysis that takes into account the smoothing
bias, we state the following corollary to Theorem \ref{Thm: Rust normal}:
For any $\lambda_{N}\rightarrow0$ such that $\lambda_{N}\sqrt{N}\rightarrow0$,
$\sqrt{N}\{V_{N}\left(\cdot;\lambda_{N}\right)-V_{0}\left(\cdot;0\right)\}\rightsquigarrow\mathbb{G}_{V}\left(\cdot,0\right).$

\subsection{\label{subsec:Projection-based-approximation}Sieve-based approximation
of value functions}

We now proceed to analyze the asymptotic properties of the sieve-based
approximate value function, $\hat{V}_{N}$. To this end, we use the
following decomposition of the over-all error,
\begin{equation}
\hat{V}_{N}-V_{0}=\left\{ \hat{V}_{N}-V_{N}\right\} +\left\{ V_{N}-V_{0}\right\} ,\label{eq: V-hat decomp}
\end{equation}
where the second term converges weakly towards a Gaussian process,
c.f. Theorem \ref{Thm: Rust normal}. What remains is to control the
first term which is due to the sieve approximation; this is done by
imposing the following high-level assumption on the projection operator
when applied to a function set $\mathcal{V}$ which is chosen so that
$V_{N}\in\mathcal{V}$ w.p.a.1.:
\begin{assumption}
\label{ass: projection 1}The projection operator $\Pi_{K}$ satisfies
$\sup_{v\in\boldsymbol{\mathbb{C}}_{r_{0}}^{\alpha}\left(\mathcal{Z}\right)}\left\Vert \Pi_{K}\left(v\right)-v\right\Vert _{\infty}=O_{P}\left(\rho_{K}\right)$
for some sequence $\rho_{K}\rightarrow0$, where $\alpha$ is given
in Assumption \ref{Ass: smoothness} and $r_{0}<\infty$ in Theorem
\ref{Thm: contraction}.
\end{assumption}
This is a high-level condition that requires the chosen function approximation
method to have a uniform error rate over the function class $\boldsymbol{\mathbb{C}}_{r}^{\alpha}\left(\mathcal{Z}\right)$
which we know $z\mapsto V_{0}\left(z,\lambda\right)$ belongs uniformly
in $\lambda$ under Assumption \ref{Ass: smoothness}. As discussed
earlier, one could replace Assumption \ref{Ass: smoothness} with
other regularity conditions that ensure $z\mapsto V_{0}\left(z,\lambda\right)$
is sufficiently regular (e.g., monotonic) in which case $\boldsymbol{\mathbb{C}}_{r}^{\alpha}\left(\mathcal{Z}\right)$
in Assumption \ref{ass: projection 1} should be modified accordingly.
Assumption \ref{ass: projection 1} is satisfied for standard polynomial
approximators with $\rho_{K}=\log\left(K\right)/K^{\left(s+1\right)/d}$,
c.f. Section \ref{subsec:Numerical-implementation}. Compared to results
on sieve approximations of value functions found elsewhere in the
literature, our rate is better since we are here seeking to approximate
the expected value function that is situated in $\boldsymbol{\mathbb{C}}_{r}^{\alpha}\left(\mathcal{Z}\right)$.
In contrast, sieve-based approximations developed in other papers,
such as \citet{Munos&Szepesvari2008} and \citet{Pal&Stacurski2013},
try to approximate the value function which is at most Lipschitz and
for such functions the approximation error will be larger in general.
In the case of $\mathcal{Z}$ being finite, we have $\sup_{V\in\mathcal{V}}\left\Vert \Pi_{K}\left(V\right)-V\right\Vert _{\infty}=0$
for $K>\left|\mathcal{Z}\right|$ under great generality and so there
will be no asymptotic bias component due to sieve approximations in
this case. 

The second part of Theorem \ref{thm: General rate} together with
the fact that $\Gamma_{N,\lambda}(V_{\lambda})-\Gamma_{\lambda}(V_{\lambda})=O_{P}\left(1/\sqrt{N}\right)$,
c.f. Proof of Theorem \ref{Thm: Rust approx rate}, now yield the
following result:
\begin{thm}
\label{thm: proj rate}Suppose Assumptions \ref{assu: mom bound}-\ref{ass: projection 1}
hold. Then $\hat{V}_{N}$, defined as the solution to $\Pi_{K}\Gamma_{N}(\hat{V}_{N})=\hat{V}_{N}$,
satisfies $\Vert\hat{V}_{N}-V_{0}\Vert_{\infty}=O_{p}(1/\sqrt{N})+O_{P}\left(\rho_{K}\right)$.
Suppose in addition that $\alpha\geq1$ in Assumption \ref{Ass: smoothness}.
Then, if $\sqrt{N}\rho_{K}\rightarrow0$, $\sqrt{N}\{\hat{V}_{N}-V_{0}\}\rightsquigarrow\mathbb{G}_{V}.$
\end{thm}
The discussions following Theorems \ref{Thm: Rust approx rate} and
\ref{Thm: Rust normal} carry over to the above result. In particular,
the rate result still goes through when no smoothing is employed ($\lambda=0$)
but the current proof of the asymptotic distribution result requires
smoothing ($\lambda>0$). Compared to the rate results for $V_{N,\lambda}$,
the projection-based method suffers from an additional error due to
the sieve approximation, $O_{P}\left(\rho_{K}\right)$. This can be
interpreted as a bias term, while $O_{p}(1/\sqrt{N})$ is its variance
component which is shared with $V_{N,\lambda}$. The requirement that
$\sqrt{N}\rho_{K}\rightarrow0$ is used to kill the sieve bias term
so that $\hat{V}_{N}$ is centered around $V_{0}$.

The above result provides a refinement over existing results where
a precise rate for the bias is not available; see, e.g., Lemma 5.2
in \citet{Pal&Stacurski2013}. It shows that there is an inherent
computational curse-of-dimensionality built into our projection-based
value function approximation when polynomial interpolation is employed:
In high-dimensional models, a large number of basis functions are
needed which in turn increases the computational effort. In the case
of polynomial approximations, the rate condition becomes $\sqrt{N}\log\left(K\right)/K^{\left(s+1\right)/d}\rightarrow0$
and so, as $d$ increases, we need $K$ to increase faster with $N$
to kill the sieve bias component. However, also note that $K$ as
no first-order effect on the variance and so there is no bias-variance
trade-off present. In particular, we can let $K$ increase with $N$
as fast as we wish and so our procedure should in principle also work
for models with high-dimensional state space. That is, we can achieve
$\sqrt{N}$-rate regardless of the dimension of the problem and so
our method does not suffer from any statistical curse-of-dimensionality.
However, this requires choosing $K$ large enough in order to control
the sieve approximation bias which will increase computation time
as the dimension grows. Thus, there is a potential computational curse-of-dimensionality.

\section{\label{sec:Numerical}Numerical results}

In this section we examine the numerical performance of the proposed
solution algorithms with focus on how the theoretical results derived
in the previous sections translate into practice and how different
features of model and implementation affect their performances. 

We focus exclusively on approximating the integrated value function,
$v\left(z\right)$, and measure the performance of a given approximate
solution, say, $\tilde{v}\left(z\right)$ in terms of its pointwise
bias, variance and mean-square error (MSE) defined as $Bias\left(z\right):=E[\tilde{v}\left(z\right)]-v\left(z\right)$,
$Var\left(z\right):=Var\left(\tilde{v}\left(z\right)\right)=E\left[(\tilde{v}\left(z\right)-E[\tilde{v}\left(z\right)])^{2}\right]$
and $MSE\left(z\right)=Bias^{2}\left(z\right)+Var\left(z\right)$,
respectively. As overall measures we use uniform bias, variance and
MSE, $\left\Vert Bias\right\Vert $$_{\infty}=\sup_{z\in\mathcal{Z}}\left|Bias\left(z\right)\right|$,
$\left\Vert Var\right\Vert $$_{\infty}=\sup_{z\in\mathcal{Z}}\left|Var\left(z\right)\right|$
and $\left\Vert MSE\right\Vert $$_{\infty}=\sup_{z\in\mathcal{Z}}\left|MSE\left(z\right)\right|$.
Given that the exact solution $v\left(z\right)$ is unknown, we replace
this by a very precise approximate solution computed in the following
way: First, instead of using simulations in the computation of the
Bellman operator, we utilize that the state transitions follow a Beta
distribution in the chosen model (see below) and so we can use nodes
and weights based on Jacobi polynomials to compute it using numerical
integration. We then implement the sieve method using $K=60$ Chebyshev
polynomials and $N=60$ sets of quadrature nodes and weights. The
``exact'' solution was computed by successive approximation until
a contraction tolerance of machine precision was reached. We approximate
the point-wise bias and variance of a given method through $S\geq1$
independent replications of it: Let $\tilde{v}_{1}\left(z\right),....,\tilde{v}_{S}\left(z\right)$
be the solutions obtained across the $S$ replications, where $S$
generally was chosen to 2,000. We then approximate the mean by $\hat{E}[\tilde{v}\left(z\right)]=\frac{1}{S}\sum_{s=1}^{S}\tilde{v}_{s}\left(z\right)$
which in turn is used to obtain the following pointwise bias and variance
estimates, $\hat{Bias}\left(z\right)=\hat{E}[\tilde{v}\left(z\right)]-v_{0}\left(z\right)$,
and $\hat{Var}\left(z\right)=\frac{1}{S}\sum_{s=1}^{S}(\tilde{v}_{s}\left(z\right)-\hat{E}[\tilde{v}\left(z\right)])^{2}$.
Based on these, we approximate $\left\Vert Bias\right\Vert $$_{\infty}$
and $\left\Vert Var\right\Vert $$_{\infty}$ by the maximum pointwise
biases and variances over a uniform grid over $[0,1000]$ of size
500 in the univariate case and a uniform grid over $[0,1000]^{2}$
of size 250.

To implement the sieve-based method, we need to choose the sieve space
used in constructing $\Pi_{K}$. We here focus on Chebyshev basis
functions and B-Splines as discussed in Section \ref{subsec:Numerical-implementation}\footnote{For more details on their implementation, see appendix \ref{sec:Sieve-spaces}
.}.

\subsection{\label{subsec: Rust model}A model of optimal replacement}

To provide a test bed for comparison of the sieve-based approximation
method, we use the well-known engine replacement model by \citet{Rust1987}.
Rust's model has become the basic framework for modeling dynamic discrete-choice
problems and has been extensively used in other studies to evaluate
the performance of alternative solution algorithms and estimators.
While the model and its solution is well described in many papers,
for completeness we briefly describe our variation of it below. 

We consider the optimal replacement of a durable asset (such as a
bus engine) whose controlled state $Z_{t}\in\mathbb{R}_{+}$ is summarized
by the accumulated utilization (mileage) since last replacement. In
each period, the decision maker faces the binary decision $d_{t}\in\mathcal{D=\mathrm{\left\{ 0,1\right\} }}$
whether to keep ($d_{t}=0)$ or replace $(d_{t}=1$) the durable asset
with a fixed replacement cost $RC>0$. If the asset is replaced, accumulated
usage $Z_{t}$ regenerates to zero. The maintenance/operating costs
are assumed to be linear in usage $Z_{t}$, $c(Z_{t})=\theta_{c}\cdot0.001\cdot Z_{t}$.
The state and decision dependent per period utility is then given
by $\bar{u}(Z_{t},d_{t})+\varepsilon_{t}(d_{t})$ where $\bar{u}(Z_{t},d_{t})=\left(RC+c(0)\right)I\left\{ d_{t}=0\right\} +c\left(Z_{t}\right)I\left\{ d_{t}=1\right\} $
and the utility shocks $\varepsilon_{t}=\left(\varepsilon_{t}(0),\varepsilon_{t}(1)\right)$
are i.i.d. extreme value and fully independent of $Z_{t}$. This specification
is a special case of Example 1 with $\lambda=\sigma_{\varepsilon}$
and so the simulated Bellman operator takes the form (\ref{eq: Gamma_N Ex 1})
where $G_{\sigma_{\varepsilon}}\left(\cdot\right)$ appears as part
of the model. Thus, there is no smoothing bias present in the baseline
model. In Section \ref{subsec:smoothng}, we investigate the effect
of smoothing by pretending that we are not able to integrate out $\varepsilon_{t}$
analytically in the baseline model and instead we simulate both $Z_{t}$
and $\varepsilon_{t}$ and then include our smoothing device in the
computation of the simulated Bellman operator.

We assume that $Z_{t}$ (in absence of the replacement decision) follows
a mixture of a discrete distribution with a probability mass $\pi>0$
at zero and a linearly transformed Beta distribution with shape parameters
$a$ and $b$ and scale parameter $\sigma_{\varepsilon}>0$. Thus,
\begin{equation}
F_{Z}\left(z^{\prime}|z,d\right)=\pi I\left\{ z^{\prime}=z\right\} +(1-\pi)F_{+}\left(z^{\prime}|z,d\right),\label{eq: milage transition}
\end{equation}
where $F_{+}\left(z^{\prime}|z,d\right)$ has density $f_{+}\left(z^{\prime}|z,d\right)=f_{\beta}\left((z^{\prime}-z)/\sigma_{Z};a,b\right)/\sigma_{Z}$,
$\pi>0$ is the probability of no usage and $f_{\beta}(x;a,b)$ is
the probability density function of the Beta distribution with shape
parameters $a,b$. Note here that it has bounded support $\left(0,\sigma_{Z}\right)$
so that $f_{+}\left(z^{\prime}|z,d\right)=0$ for $z^{\prime}<z$
or $z^{\prime}-z>\sigma_{Z}$. This is in line with the discretized
model in the original formulation in \citet{Rust1987} where monthly
mileage were only allowed to take a few discrete values and monthly
mileage is naturally bounded above and below (busses never drives
backwards and there are limits how far a bus can drive within a month).\textbf{
}We introduce probability mass $\pi$ at \textbf{$z^{\prime}=z$ }to
allow for the possibility that the asset is not used in a given period
and thereby can end in the same state with positive probability when
$\pi>0$. As explained below, this feature turns out to be quite important
for the applicability of the self-approximating method of Rust (1997).
Note that the support of $Z_{t}$ is unbounded (the positive half
line) and therefore the theory does not apply directly, since we throughout
assumed bounded support. However, we expect that the theory extends
to the unbounded case after suitable modifications, c.f. discussion
following Assumptions \ref{assu: mom bound}-\ref{Ass: smoothness}.

In the numerical illustrations below we use the following set of benchmark
parameter values unless otherwise specified: We set replacement cost
to $RC=10$ and the cost function parameter to $\theta_{c}=2$ so
that $RC$ is 5 times as large as $c(1000)$. This implies a large
variation in the probability of replacement over $Z_{t}$ compared
to \citet{Rust1987} and a more curved value function. The parameters
indexing the transition density $f_{+}\left(z^{\prime}|z,d\right)$
are set to $\sigma_{Z}=15$, $a=2$, $b=5$ and $\pi=10^{-10}$ as
default. This implies a quite sparse transition density, which is
similar to the fitted model in \citet{Rust1987}. In Figure \ref{fig: exact_solution_1d}
we plot the corresponding ``exact'' solution as described earlier.
 Importantly, since the transition density is an analytic function
the value function is also analytic and so well-approximated by polynomial
interpolation methods.

\begin{figure}
\caption{Fine Approximation as ``Exact'' Solution\label{fig: exact_solution_1d}}

\includegraphics[width=0.5\textwidth]{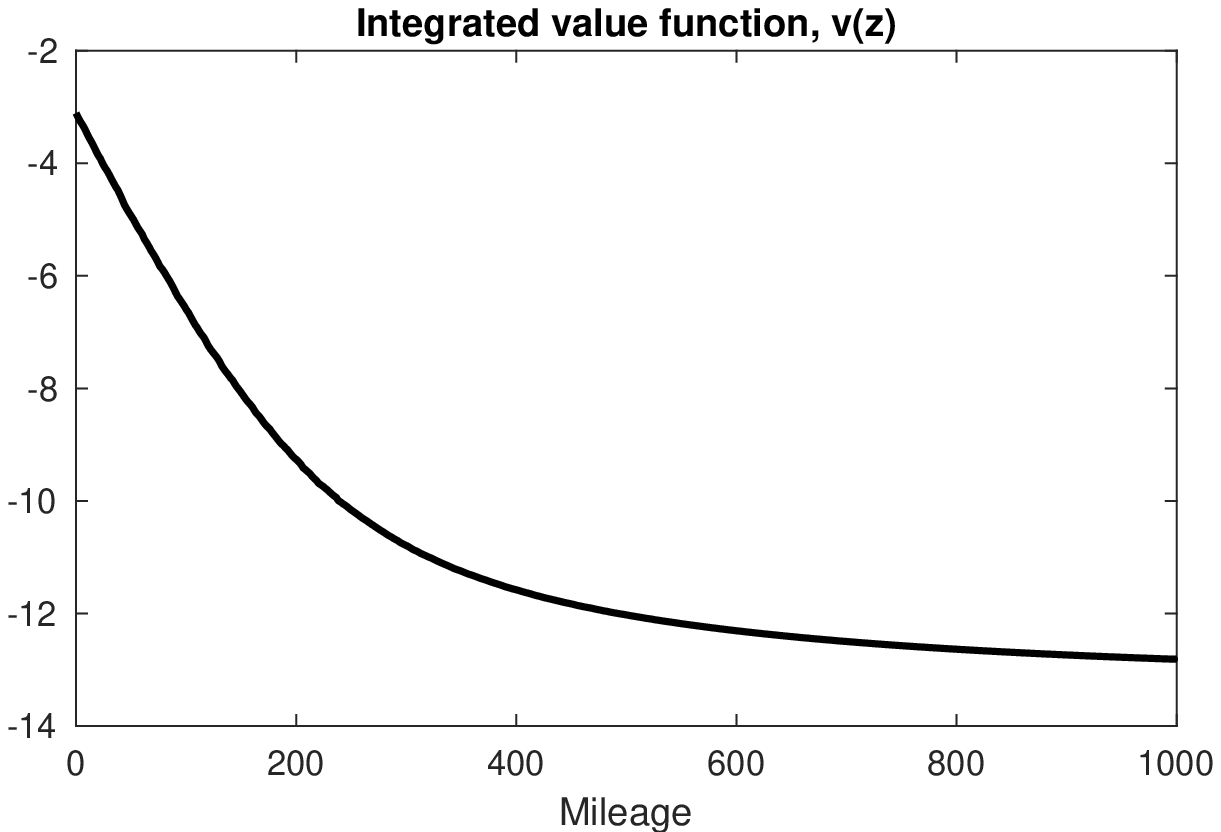}\includegraphics[width=0.5\textwidth]{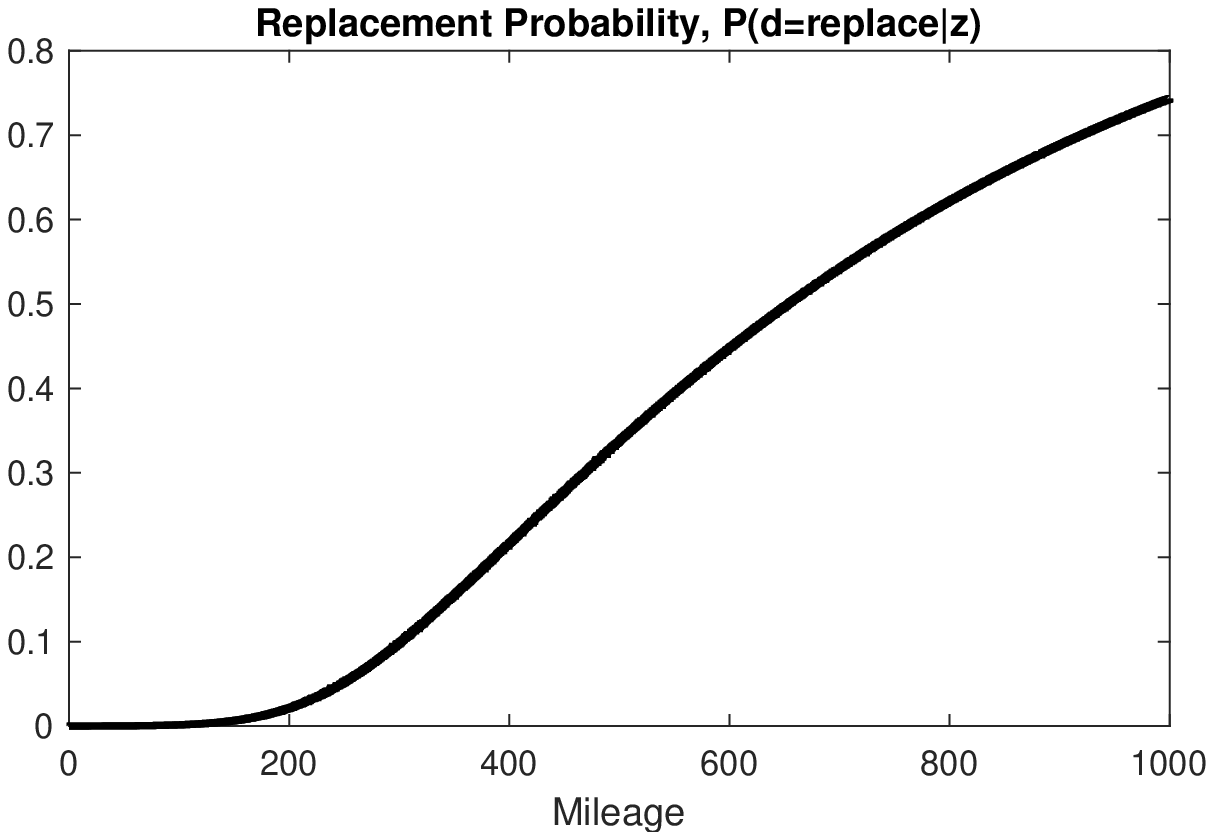}

{\footnotesize{}Notes: Discount factor is $\text{\ensuremath{\beta}}=0.95$,
utility function parameters are $\theta_{c}=2$, $RC=1$, $\lambda=1$
and transition parameters are $\sigma_{Z}=15$, $a=2$, $b=5$ and
$\pi=0$.000000001. }{\footnotesize\par}
\end{figure}

\subsection{\label{subsec:Implementation-of-Bellman}Numerical implementation
of simulated Bellman operators}

The simulated Bellman operators in (\ref{eq: smooth Gamma_N}) and
(\ref{eq: smooth Gamma-bar_N}) require the user to choose an importance
sampling distribution. For the self-approximating solution method
we need to choose a marginal sampler, $d\Phi_{Z}\left(z^{\prime}|z,d\right)=\phi_{Z}\left(z^{\prime}\right)dz^{\prime}$.
We follow \citet{Rust1997curse} and choose $\phi_{Z}\left(z^{\prime}\right)=I\left\{ 0<z^{\prime}<z^{\max}\right\} $
as a uniform density with support support $\left[0,z^{\max}\right]$
for some truncation point $0<z^{\max}<\infty$ chosen by us. First
note that this entails that the simulated Bellman operator used for
the self-approximating value function is biased since we do not sample
from the full support $\mathcal{Z}=\mathbb{R}_{+}$; however, this
bias can be controlled by choosing $z^{\max}$ large enough. We will
explain below why we do not choose $\phi_{Z}\left(z^{\prime}\right)$
as a density with support $\mathbb{R}_{+}$. Using a uniform sampler,
the corresponding Radon-Nikodym derivative takes the form $w_{Z}\left(z^{\prime}|z,d\right)=\pi\delta\left(z^{\prime}-z\right)+(1-\pi)f_{+}\left(z^{\prime}|z,d\right)$
where $\delta\left(\cdot\right)$ denotes Dirac's delta function.
We approximate this by $\hat{w}_{Z}\left(z^{\prime}|z,d\right)=\pi I\left\{ z^{\prime}=z\right\} +(1-\pi)f_{+}\left(z^{\prime}|z,d\right)$
which entails another small approximation error. For the sieve-based
version, we simply choose $\Phi_{Z}\left(z^{\prime}|z,d\right)=F_{Z}\left(z^{\prime}|z,d\right)$
and so $w_{Z}\left(z^{\prime}|z,d\right)=1$. 

As explained in Section \ref{subsec:Numerical-implementation}, using
a marginal importance sampler creates issues since it fails to adapt
to the particular shape of the support of $F_{Z}\left(z^{\prime}|z,d\right)$.
In particular, for a given choice of $z$, many of the draws from
$\phi_{Z}\left(z^{\prime}\right)$ will tend to fall outside the support
of $f_{Z}\left(z^{\prime}|z,d\right)$ and so will not contribute.
In contrast, when $\phi_{Z}\left(z^{\prime}|z,d\right)=f_{Z}\left(z^{\prime}|z,d\right)$,
the draws from $\phi_{Z}$ will by construction fall within the support
of $f_{Z}\left(z^{\prime}|z,d\right)$. This can be seen in Figure
\ref{fig:Ggrids} where we have plotted the random draws obtained
from the two different importance samplers used for the sieve-based
and self-approximating solutions together with the actual support
of $f_{Z}\left(z^{\prime}|z,d\right)$. In the left-hand side panel
we have plotted pairs of the uniform draws, $\left(Z_{i},Z_{j}\right)$
for $i,j=1,...,N$, used for Rust's self-approximating method with
$N=400$ and $z^{max}=1,000$, while in the right-hand side we have
plotted $\left(z_{i},Z_{j}\left(z_{i},d\right)\right)$ where $z_{i}$
are uniform draws and $Z_{j}\left(z,d\right)\sim f_{Z}\left(\cdot|z,d\right)$.
In both cases, we have marked the pairs for which the corresponding
density, $f_{Z}\left(Z_{j}|Z_{i},d\right)$ and $f_{Z}\left(Z_{j}\left(z_{i},d\right)|z_{i},d\right)$,
respectively, is positive. Clearly, the use of a marginal importance
sampling density leads to very poor coverage of the actual support
of $f_{Z}\left(z^{\prime}|z,d\right)$ as $z$ varies while by construction
$\Phi_{Z}\left(z^{\prime}|z,d\right)=F_{Z}\left(z^{\prime}|z,d\right)$
does an excellent job. This translates into the former simulated Bellman
operator exhibiting much larger variance compared to the latter.

\begin{figure}
\caption{Random Grids\label{fig:Ggrids}}

\includegraphics[width=0.5\textwidth]{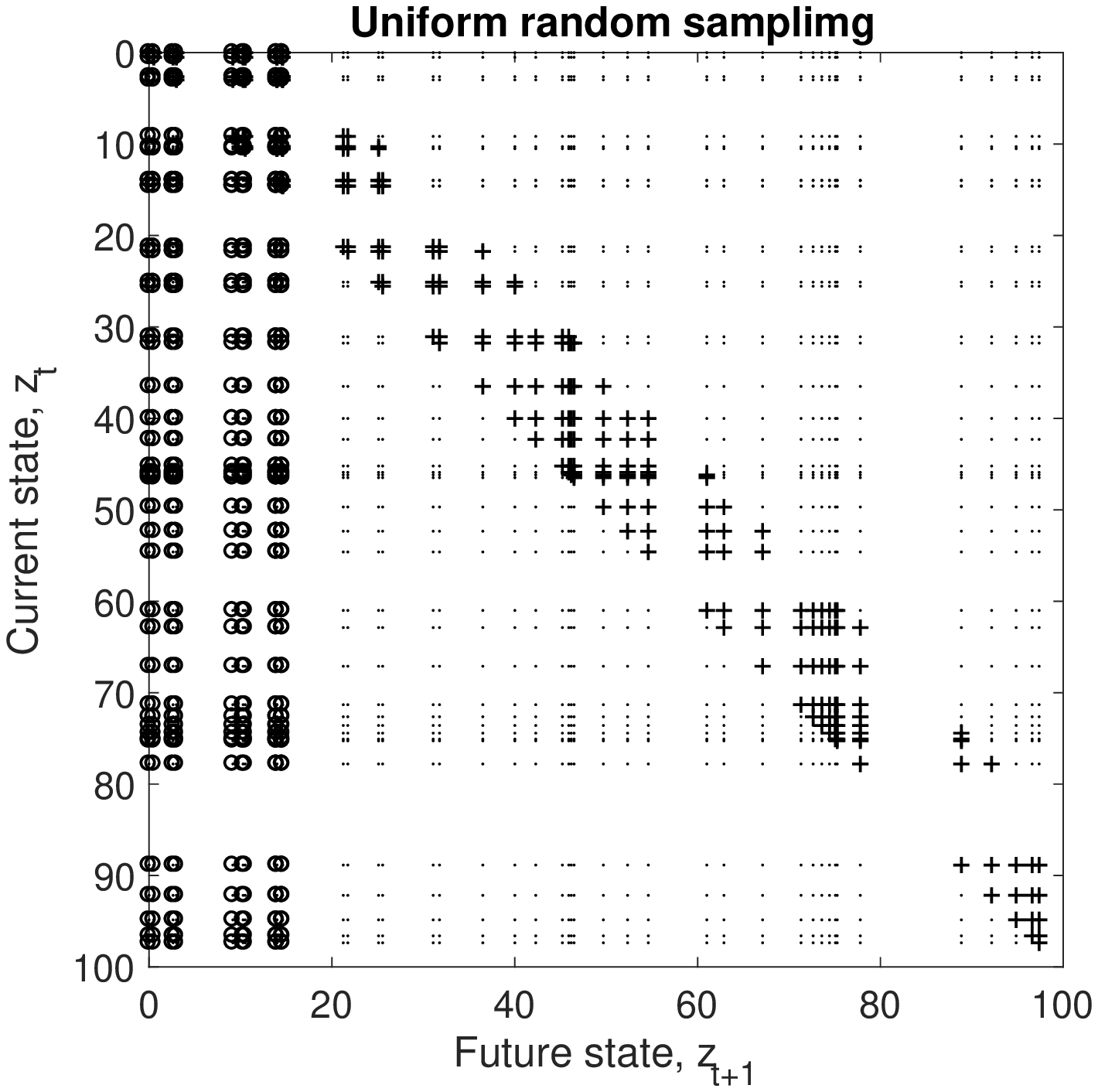}\includegraphics[width=0.5\textwidth]{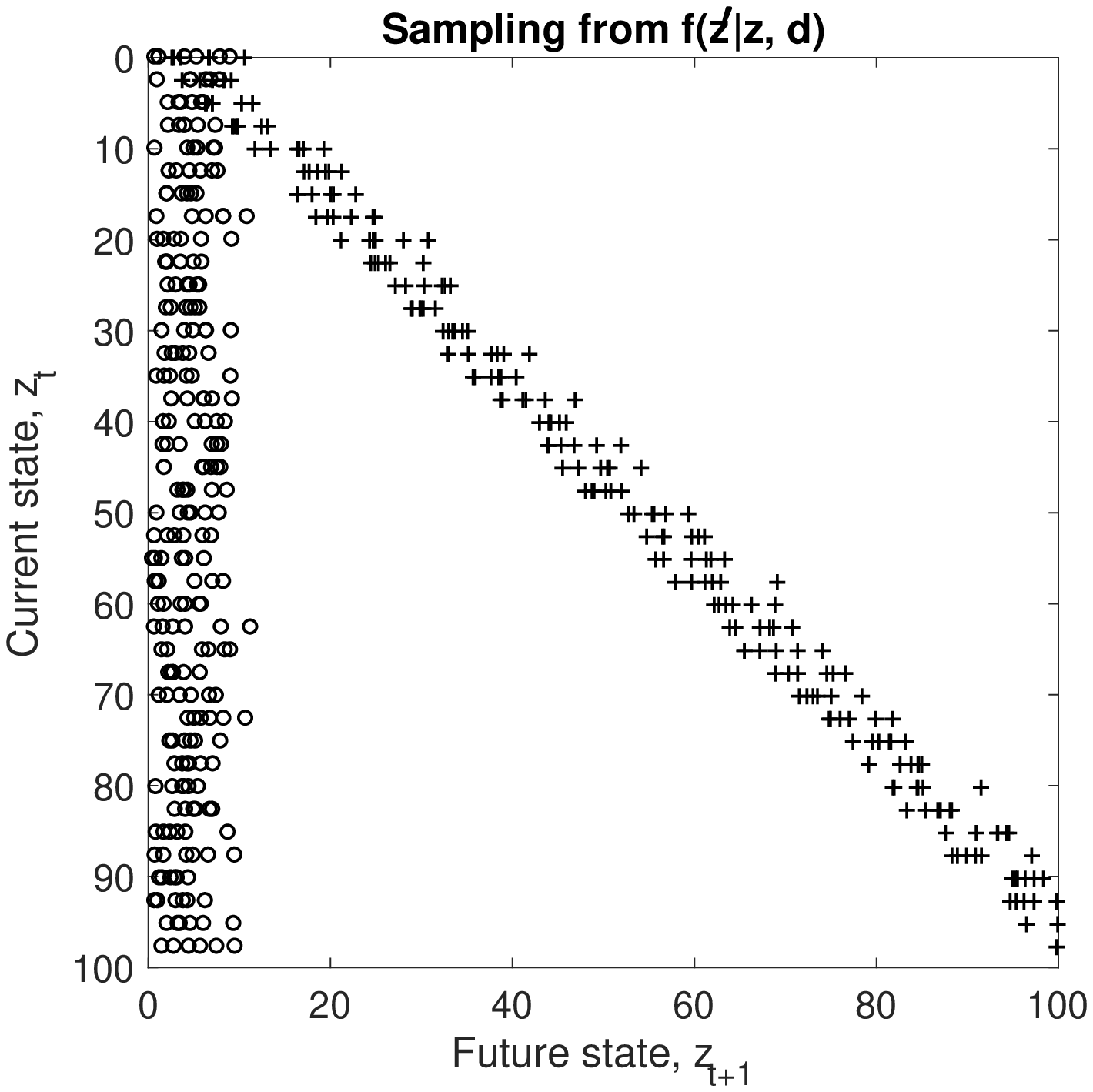}

{\footnotesize{}Notes: In the left panel we present the grids used
for the self-approximate random Bellman operator. We have uniformly
sampled a random grid, $\{Z_{1},...,Z_{N}\}$ on the interval $[0;1000]$
with $N=$400. Dots (.) mark sampled grid points in $R^{2}$: $Z_{N}\times Z_{N}$,
plus (+) mark grid points where $f(z_{j}|z_{i},d=0)>0$ and circles
(o) mark points where $f(z_{j}|z_{i},d=1)>0$. In the right panel,
we plot the grid the projected random Bellman operator, where we have
sampled directly from the conditional transition density in each of
the $M=400$ uniformly spaced evaluation points. To have equally many
grid-points with non-zero transition density we only need $N=400*\sigma_{Z}/\max(Z_{N})=9$
random grids for each of the $M=400$ evaluation points. Both figures
show only a subset of the state space, $(z,z')\in[0;100]^{2}$. Parameters
are $\sigma_{Z}=15$, $a=2$, $b=5$ and $\pi=0.0000000001$.}{\footnotesize\par}
\end{figure}

This issue is further amplified when we introduce the normalization
given in eq. (\ref{eq: sampling weight}): Suppose that we had not
included a discrete component $\pi I\left\{ z^{\prime}=z\right\} $
in the model. Then, with $Z_{i}\sim U\left[0,z^{\max}\right]$, $w_{N,Z,i}(Z_{j},d)=f_{+}\left(Z_{i}|Z_{j},d\right)/{\textstyle \sum}_{k=1}^{N}f_{+}\left(Z_{k}|Z_{j},d\right).$
Since $f_{+}(z^{\prime}|z,d)$ has bounded support, it often happens
that ${\textstyle \sum}_{k=1}^{N}f_{+}\left(Z_{k}|Z_{j},d\right)=0$
for even large values of $N$ and so the simulated Bellman operator
is not even well-defined. This issue will vanish as $N\rightarrow\infty$,
but this on the other hand increases the computational burden since
the self\textendash approximating method require us to solve for the
value function at the $N$ draws. Introducing the discrete component
in the model resolves this issue since now $w_{Z,N,i}(Z_{j},d)=\hat{w}_{Z}\left(Z_{i}|Z_{j},d\right)/{\textstyle \sum}_{k=1}^{N}\hat{w}_{Z}\left(Z_{k}|Z_{j},d\right),$
where ${\textstyle \sum}_{k=1}^{N}\hat{w}_{Z}\left(Z_{k}|Z_{j},d\right)>0$
for all $j=1,...,N$ by construction. Thus, $\pi>0$ functions as
a regularization device.

\begin{figure}
\noindent \begin{centering}
\caption{Truncation bias due to $z^{max}$ being too low.\label{fig:Truncation-bias}}
\par\end{centering}
\noindent \centering{}\input{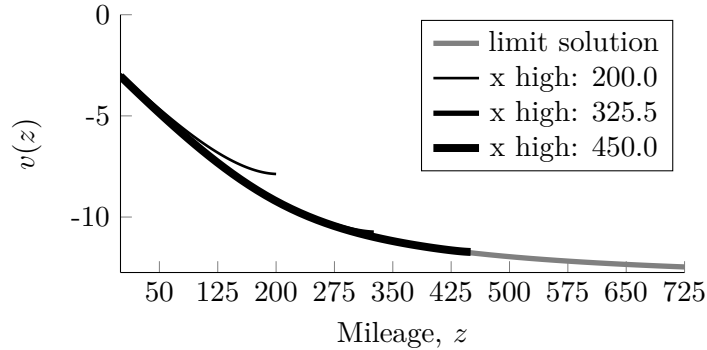}
\end{figure}

Why not choose $\phi_{Z}\left(z\right)$ as a density with unbounded
support in order to avoid the issue of truncation? In our initial
experimentation, we did try out sampling from distributions with unbounded
support, but the above numerical issues became even more severe in
this case since the resulting draws are even more dispersed. Figure
\ref{fig:Truncation-bias} shows how the solution depends on $z^{\max}$.
The effect of the truncation $z^{\max}$ will be model specific and
in practice experimentation is required. If we, for example, simply
set $z^{\max}=1,000,000$, the variance of the simulated Bellman operator
becomes very large for a given $N$ due to the issue with undefined
sample weights $w_{N,i}\left(z,d\right)$ mentioned above. At the
same time, choosing $z^{\max}$ too small leads to a large bias. To
balance the bias and variance, we ended up using $z^{\max}=1000$
which all subsequent numerical results for the self-approximating
method is based on. Finally, we would like to stress that none of
these issues appear for the sieve-based method.

\subsection{Convergence properties and computation times}

We first investigate the convergence properties of our solution methods
for given choice of $K$ and $N$. Do they converge and if so how
fast?

\subsubsection*{Global convergence properties of sieve method}

As demonstrated in Theorem \ref{Thm: contraction}, the simulated
Bellman operators are always contraction mappings and so the self-approximating
method is guaranteed to converge using successive approximations.
In contrast, $\Pi_{K}\bar{\Gamma}_{N,\lambda}$ is not necessarily
a contraction and so global convergence of the sieve method may fail,
c.f. discussion in Section \ref{subsec:Projection-based-approximation}.
A sufficient condition for global convergence is $||\Pi_{K}||_{op,\infty}<1/\beta$
and we saw that $||\mathbf{P}_{K}||_{op,\infty}>1$ implies $||\Pi_{K}||_{op,\infty}>1$.
However, even if $||\mathbf{P}_{K}||_{op,\infty}>1$, successive approximation
may still converge: Across various parameter values of model, choices
of sieve spaces and number of simulations, we did not encounter any
failure of the sieve method to converge and the resulting approximate
solution was well-behaved. This finding held across various initializations
of the solution algorithms (initial choice of sieve coefficients).
For example, we implemented the sieve method using $M=64$ evaluation
points and using either $K=1$ or $K=4$ Chebyshev basis functions.
We found that $||\mathbf{P}_{1}||_{op,\infty}=1$ while $||\mathbf{P}_{4}||_{op,\infty}>1.78$
and so the sieve method was guaranteed to converge for $K=1$ but
not for $K=4$. Nevertheless, the method of successive approximations
did in fact converge to a tolerance of $10^{-12}$ for both $K=1$
and $K=4$.

\subsubsection*{Successive approximation versus Newton-Kantorovich}

In Section \ref{subsec:Numerical-implementation} we advocated a hybrid
of successive approximation (SA) and Newton-Kantorovich (NK) where
we start with SA to ensure global convergence, and switch to NK iterations
once the domain of attraction has been reached since NK generally
converges faster. We illustrate this attractive feature of the NK
algorighm in Figure \ref{fig:DISCOUNT-1} where we have plotted the
log residual error of the current value function approximation (relative
to the ``exact'' solution) against the iteration count for the SA
and NK algorithms, respectively, for four different values of $\beta$.
As expected, the convergence of the SA algorithm requires a very large
number of iterations ($>1000$) with computation time increasing in
$\beta$, where as NK converges after less than 10 iterations and
with the value of $\beta$ having little effect on its performance.

\begin{figure}
\caption{Convergence and discount factor \label{fig:DISCOUNT-1}}

\noindent \begin{centering}
\includegraphics[width=0.5\textwidth]{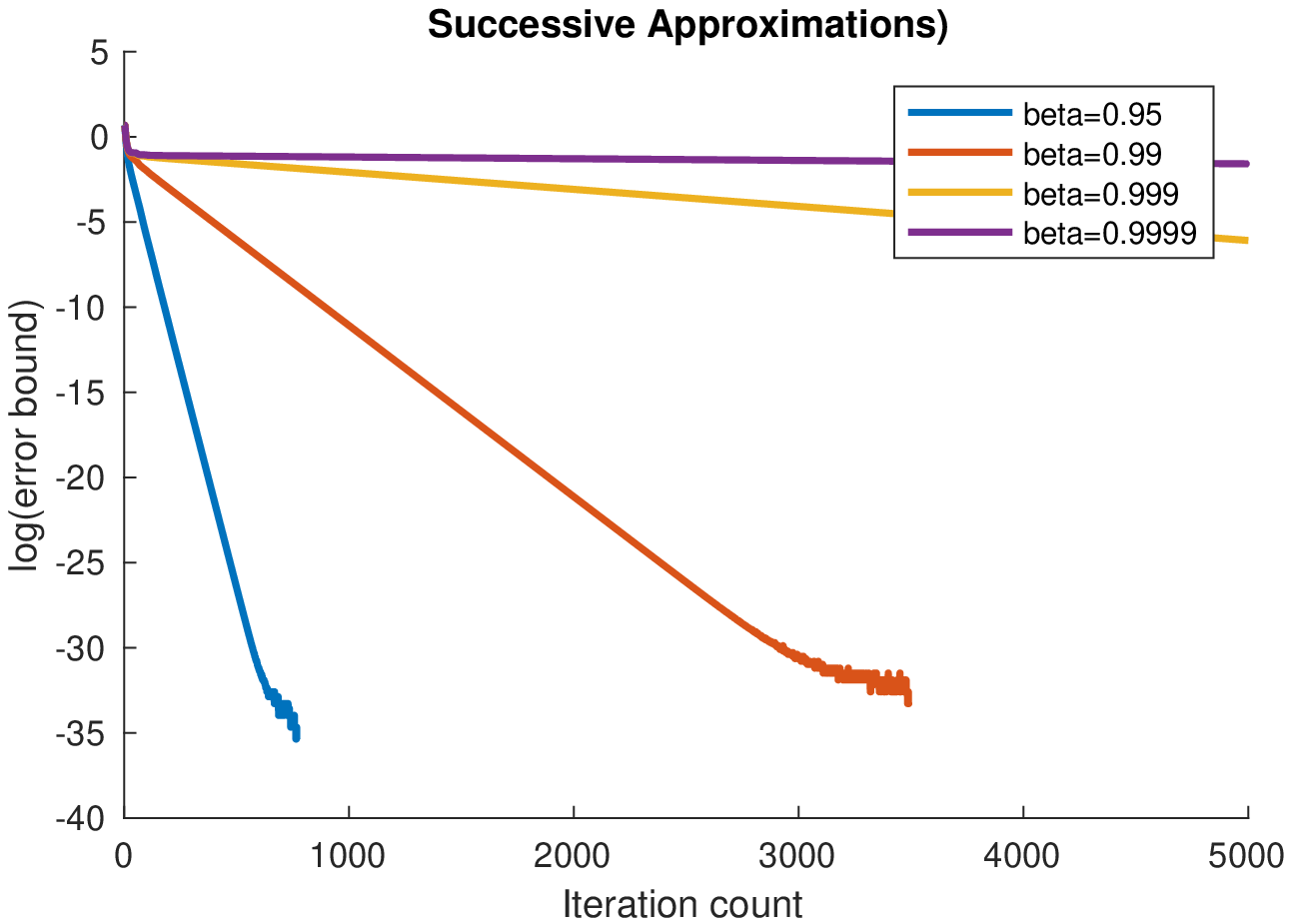}\includegraphics[width=0.5\textwidth]{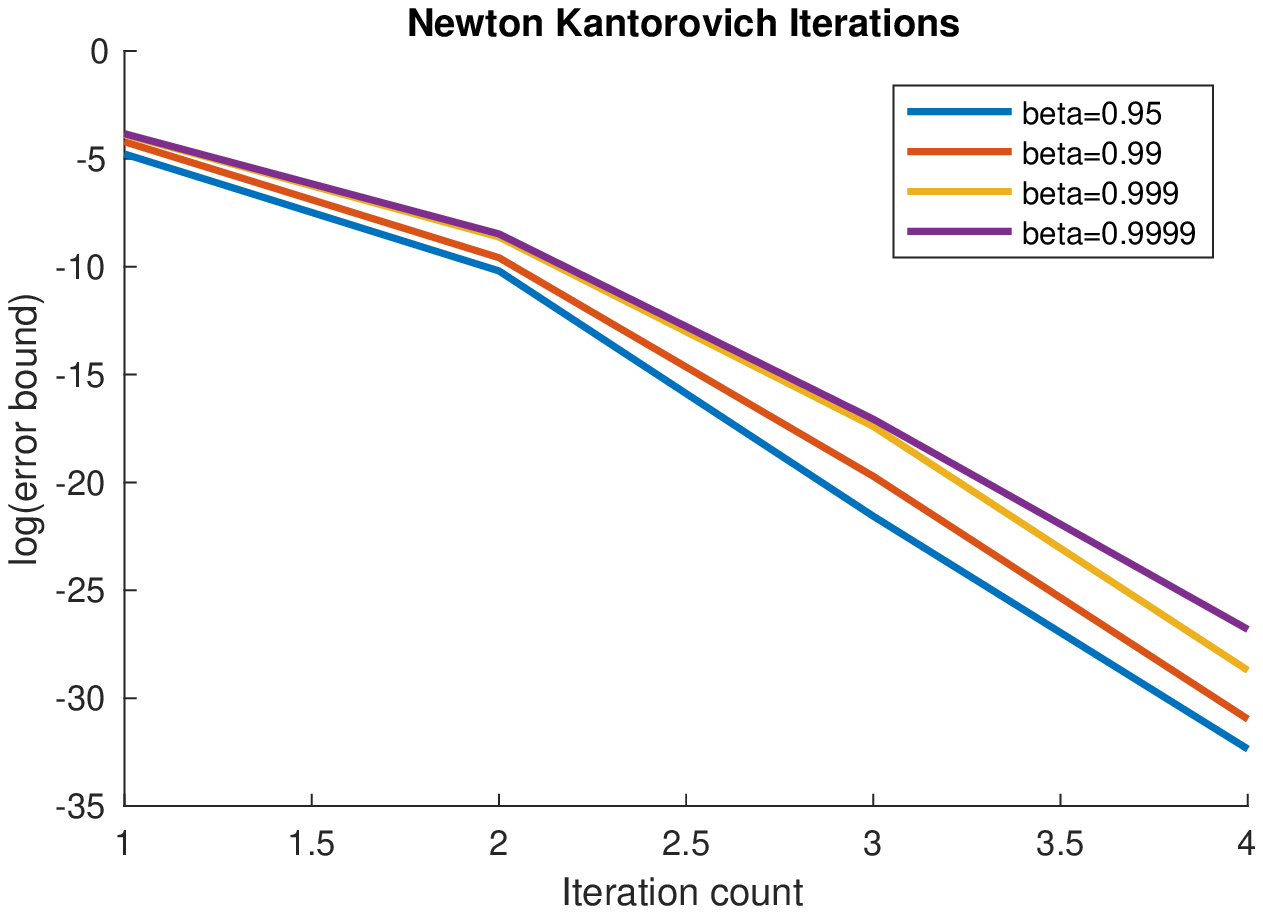}
\par\end{centering}
{\footnotesize{}Notes: Discount factor is $\text{\ensuremath{\beta}}\in\{0.95,0.99,0.999,0.9999\}$,
utility function parameters are $\theta_{c}=2$, $RC=1$, $\lambda=1$
and transition parameters are $\sigma_{\varepsilon}=15$, $a=2$,
$b=5$ and $\pi=0$.000000001.}{\footnotesize\par}
\end{figure}

Figure \ref{fig:DISCOUNT-1} is silent about the over-all computation
time of SA relative to NK. Compared to SA, each NK iteration is more
expensive since the former only requires computing the simulated Bellman
operator evaluated at the value function obtained in the previous
step while the latter, in addition, requires computing its functional
derivative and inverting a $K\times K$ dimensional matrix for the
integrated value function and a $KD\times KD$ dimensional matrix
for the expected value function, c.f. Section \ref{subsec:Numerical-implementation}.
With $K$ large, one could therefore fear that NK would become computationally
too expensive. 

In Figure \ref{fig:speed-1-1} we report best of 10 run-times for
various levels of $K$ and $\beta$ and tolerance levels of SA and
NK where we also include set-up time (time spent on initial computations
before starting the actual algorithm). As expected, we find that NK
is the faster of the two algorithms when $\beta$ is relatively large
and $K$ is relatively small. With $K=5$ NK is faster across all
levels of $\beta$ while for $K=100$ and $K=500$, SA is faster for
moderate values of $\beta$. However, as we shall subsequently see,
with $K=5$ the sieve method carries almost no bias and so choosing
$K$ larger (such as 100 or 500) is actually unnecessary here and
is only included here to illustrate potential issues with NK for models
where a large number of sieve terms are needed to obtain a good approximation
of the value function. Moreover, in most empirical applications, $\beta$
is chosen to be larger than 0.99 in which case NK still dominates
SA even with $K=500$. 

\begin{figure}
\caption{Run-times (incl setup times) for SA (dotted lines) and NK (drawn lines)
algorithms.}
\label{fig:speed-1-1}\input{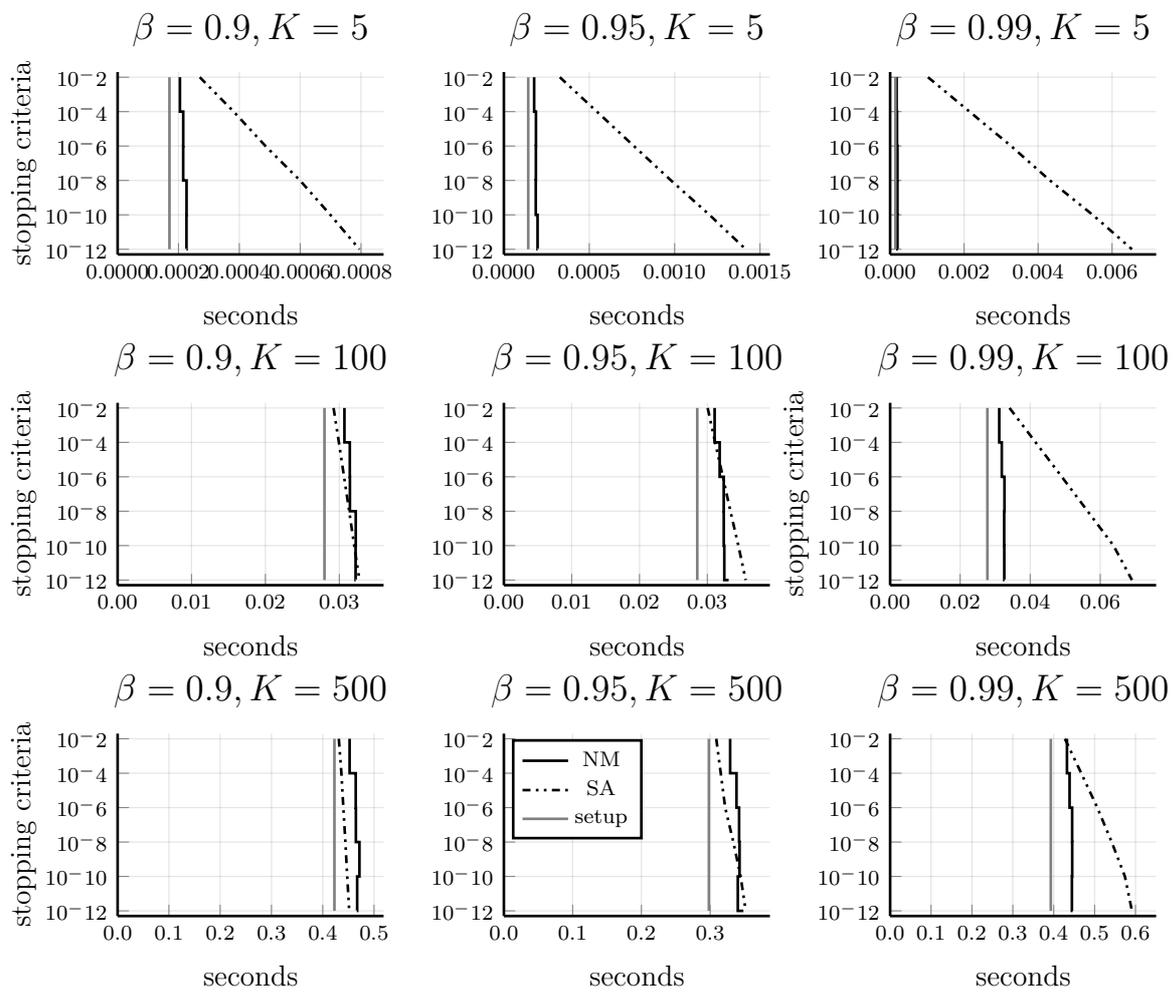}
\end{figure}

\subsection{Approximation quality}

We here investigate how the approximate value function is affected
by the number of draws and the chosen projection basis. The goal is
to demonstrate the rate results of the theoretical sections, and to
compare the two types of basis functions spaces that we described
above. We will take a partial approach and first fix $N$ to study
the role of $K$, and then afterwards fix $K$ to study the role of
$N$. All subsequent results are for the case of $\beta=0.95$. This
is to save space. We implemented the methods for other values of $\beta$
and since the numerical results were qualitatively the same, we have
left these out. The main difference in the numerical results is that
higher values of $\beta$ tend to shift the overall level of the value
function upwards and add more curvature to it. This in turn generally
leads to an increase in the absoute bias and variance numbers. However,
in terms of percentage bias and variance, the performance of the methods
were very similar across different values of $\beta$.

\subsubsection*{Effect of varying $K$ for projection-based value function approximation}

\begin{figure}
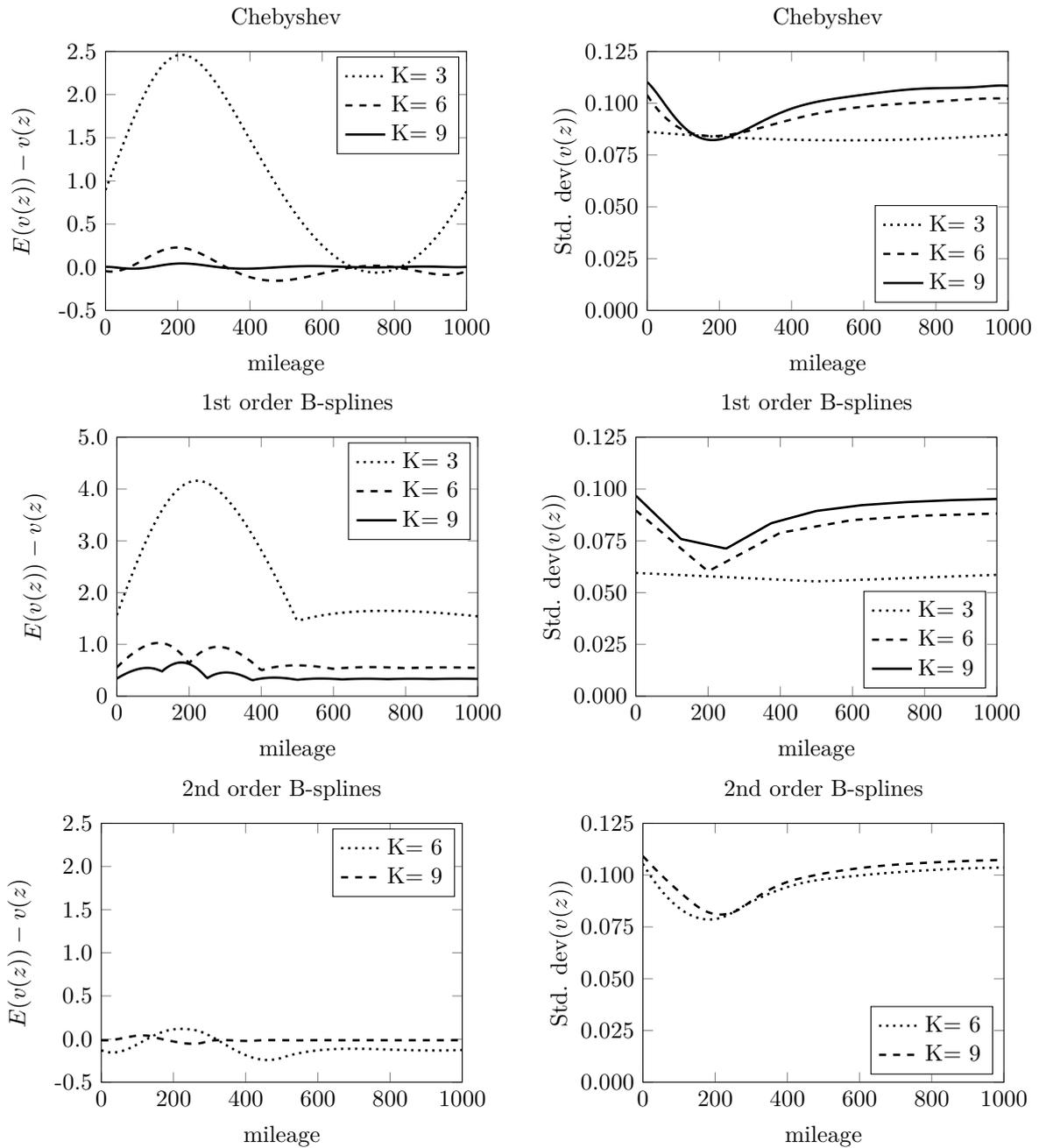

\caption{Point-wise bias and standard deviation of solutions for various choices
of $K$ using different interpolation schemes, $N=200$, $S=200$,
$\sigma_{Z}=15$.\label{fig:Point-wise-bias}}

\noindent \begin{centering}
\input{"fig/chebyshev_biasstd_200.tex"}
\par\end{centering}
\noindent \begin{centering}
\input{"fig/bspline2_biasstd_200.tex"}
\par\end{centering}
\noindent \centering{}\input{"fig/bspline3_biasstd_200.tex"}
\end{figure}

The theory for the projection-based value function approximation informs
us that the choice of the basis functions will have a first-order
effect on the bias while only a second-order effect on the variance.
In particular, we expect $Bias\left(z\right)$, as defined in the
beginning of this section, to satisfy $Bias\left(z\right)\cong\Pi_{K}\left(v\right)\left(z\right)-v\left(z\right)$,
c.f. discussion following Theorem \ref{thm: proj rate}, while $Var\left(z\right)$
should be much less affected by $K$. The actual size of the bias
obviously depends on the curvature and smoothness of $v_{0}$ and
the particular choice of basis functions. But we know that $v$ is
an analytic function and with only moderate curvature, c.f. Figure
\ref{fig: exact_solution_1d} and so expect it to be well-approximated
by a small number of polynomial basis functions. 

This is confirmed by the pointwise bias and standard deviation, $\sqrt{Var\left(z\right)}$,
reported in Figure \ref{fig:Point-wise-bias}: First, as can be seen
in the left-hand side panel of Figure \ref{fig:Point-wise-bias},
first-order B-splines lead to significantly larger point-wise bias
compared to the other two sieve bases, namely second-order B-splines
and Chebyshev polynomials. This is accordance with theory since we
know that a smooth function is better approximated by higher-order
polynomials, c.f. the error rates reported in Example 1 as a function
of $s$. At the same time, second-order B-splines and Chebyshev polynomials
exhibit very similar biases for a given choice of $K$.

The right-hand side panel of Figure \ref{fig:Point-wise-bias} shows
the point-wise standard deviation across different choices of $K$
for the three different sieve bases. Consistent with the theory, the
standard deviation of the value function approximation is not very
sensitive to the particular choice of the sieve basis and the number
of basis functions uses. That is, the sieve basis mostly affect the
bias with only minor impact on the variance.

Finally, we examine how the bias behaves as we further increase $K$.
Figure \ref{fig:sup-Bias} plots $\left\Vert Bias\right\Vert _{\infty}$
as a function of $K$. Similar to Figure \ref{fig:Point-wise-bias}
we see much more rapid convergence when smooth basis functions are
used, and with little improvement for $K$ greater than 9. This is
not surprising given the reported shape of $v$. The second-order
B-splines and Chebyshev basis functions produce very similar fits,
even if they are evaluated on different grids and the B-splines have
very different properties compared to Chebyshev polynomials. Indeed,
the curves are practically overlapping. This is in accordance with
the asymptotic theory that predicts that higher-order B-splines and
Chebyshev polynomials should lead to similar biases. Moreover, the
theory informs us that if $v$ is analytic, and this is the case in
this particular implementation, we should expect the bias to vanish
with rate $O(K^{-K})$ when using polynomial interpolation. The bias
indeed does go to zero very quickly and so the numerical results support
the theory.

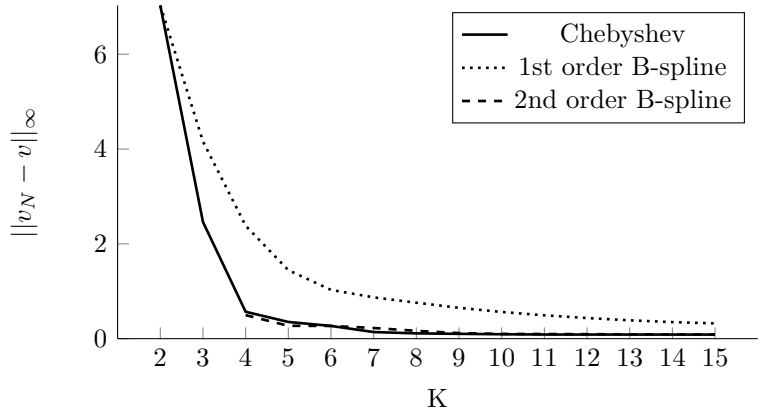
\begin{figure}
\caption{Sup-norm of bias of solutions for various choices of $K$ using various
interpolation schemes, $N=200$, $S=200$.\label{fig:sup-Bias}}

\noindent \centering{}\begin{tikzpicture}[]

    \tikzstyle{every node}=[font=\small]
\begin{axis}[height = {60.0mm}, ylabel = {$||v_{N}-v||_\infty$}, xmin = {1}, xmax = {16}, ymax = {7.028193399866775}, xlabel = {K}, {unbounded coords=jump, scaled x ticks = false, xticklabel style={rotate = 0}, xmajorgrids = false, xtick = {2.0,3.0,4.0,5.0,6.0,7.0,8.0,9.0,10.0,11.0,12.0,13.0,14.0,15.0}, xticklabels = {2,3,4,5,6,7,8,9,10,11,12,13,14,15}, xtick align = inside, axis lines* = left, scaled y ticks = false, yticklabel style={rotate = 0}, ymajorgrids = false, ytick = {0.0,2.0,4.0,6.0}, yticklabels = {0,2,4,6}, ytick align = inside,     xshift = 0.0mm,
    yshift = 0.0mm,
    axis background/.style={fill={rgb,1:red,1.00000000;green,1.00000000;blue,1.00000000}}
}, ymin = {0}, width = {100.0mm}]\addplot+ [color = {rgb,1:red,0.00000000;green,0.00000000;blue,0.00000000},
draw opacity=1.0,
line width=1,
solid,mark = none,
mark size = 2.0,
mark options = {
    color = {rgb,1:red,0.00000000;green,0.00000000;blue,0.00000000}, draw opacity = 1.0,
    fill = {rgb,1:red,0.00000000;green,0.00000000;blue,0.00000000}, fill opacity = 1.0,
    line width = 1,
    rotate = 0,
    solid
}]coordinates {
(2.0, 7.028193399866751)
(3.0, 2.461881956928099)
(4.0, 0.567558878125906)
(5.0, 0.352882106071749)
(6.0, 0.2688068197407547)
(7.0, 0.13954473308480128)
(8.0, 0.11194480349633702)
(9.0, 0.10228151324126658)
(10.0, 0.09322206222850193)
(11.0, 0.08760069278015038)
(12.0, 0.08704772197385634)
(13.0, 0.08689549971762307)
(14.0, 0.08654074673356835)
(15.0, 0.08650186339201507)
};
\addlegendentry{Chebyshev}
\addplot+ [color = {rgb,1:red,0.00000000;green,0.00000000;blue,0.00000000},
draw opacity=1.0,
line width=1,
dotted,mark = none,
mark size = 2.0,
mark options = {
    color = {rgb,1:red,0.00000000;green,0.00000000;blue,0.00000000}, draw opacity = 1.0,
    fill = {rgb,1:red,0.00000000;green,0.00000000;blue,0.00000000}, fill opacity = 1.0,
    line width = 1,
    rotate = 0,
    solid
}]coordinates {
(2.0, 7.028193399866775)
(3.0, 4.159644341957215)
(4.0, 2.382268862862936)
(5.0, 1.4447884084798222)
(6.0, 1.0291656822383892)
(7.0, 0.8713446088911839)
(8.0, 0.7586581480594561)
(9.0, 0.6510941651145218)
(10.0, 0.5633563436505626)
(11.0, 0.4927340840715479)
(12.0, 0.4352408434493137)
(13.0, 0.38744508990743265)
(14.0, 0.34937837078985395)
(15.0, 0.32232349357482554)
};
\addlegendentry{1st order B-spline}
\addplot+ [color = {rgb,1:red,0.00000000;green,0.00000000;blue,0.00000000},
draw opacity=1.0,
line width=1,
dashed,mark = none,
mark size = 2.0,
mark options = {
    color = {rgb,1:red,0.00000000;green,0.00000000;blue,0.00000000}, draw opacity = 1.0,
    fill = {rgb,1:red,0.00000000;green,0.00000000;blue,0.00000000}, fill opacity = 1.0,
    line width = 1,
    rotate = 0,
    solid
}]coordinates {
(4.0, 0.49653939640422246)
(5.0, 0.2715050342965163)
(6.0, 0.26871184823151717)
(7.0, 0.2256258082327038)
(8.0, 0.1653127873050981)
(9.0, 0.11802822877789032)
(10.0, 0.1008402786874866)
(11.0, 0.09794024801473636)
(12.0, 0.09497083000946764)
(13.0, 0.09197189334189011)
(14.0, 0.08960400861697)
(15.0, 0.08816108755673163)
};
\addlegendentry{2nd order B-spline}
\end{axis}

\end{tikzpicture}
\end{figure}

\subsubsection*{Simulation errors, rates of convergence and asymptotic normality}

We now compare the errors due to simulations and the rates with which
these vanish for the two solution methods. For both methods, theory
tells us that $N$ should have a first-order effect on the variance
of the approximate value function which is supposed to vanish at rate
$1/N$, c.f. Theorems \ref{thm: proj rate} and \ref{Thm: Rust approx rate}.
Our asymptotic theory is, on the other hand, silent about the size
of simulation bias and the rate with which it should vanish with.
However, we can think of both the sieve-based and self-approximating
method as a nonlinear GMM-estimator where the simulated Bellman operator
defines the sample moments. Importing results for GMM estimators,
see, e.g., \citet{Newey&Smith2004}, we should expect the simulation
bias to be of order $1/N$.

In Figure \ref{fig:RANDOMIZED} we investigate this prediction by
plotting $\left\Vert Bias\right\Vert $$_{\infty}$ and $\left\Vert \sqrt{Var}\right\Vert $$_{\infty}$
for the sieve-based method (left panels) and for the self-approximating
method (right panels) for two different choices of $\sigma_{\varepsilon}$
and for across different values of $N$. To examine the rate with
which the simulation bias and variance vanish we estimate the following
an exponential regressions by NLS $||\sqrt{Var}||_{\infty}=\exp(\alpha_{SD}+\rho_{SD}\ln(N))$
and $||Bias||_{\infty}=\exp(\alpha_{Bias}+\rho_{Bias}\ln(N))$ where
$\rho_{SD}$ and $\rho_{Bias}$ measures the rate that $||\sqrt{Var}||_{\infty}$
and $||Bias||_{\infty}$ vanishes with respectively. The resulting
regression fit estimates are reported in both Figure \ref{fig:RANDOMIZED}
as well as in Table \ref{tab:ConvRates_by_K}. In Table \ref{tab:ConvRates_by_K}
we present bias and standard deviation for $N=500$ as well as their
rates of convergence both methods; with various values of $K$ for
the sieve approximation method.

\begin{figure}
\caption{Convergence results\label{fig:RANDOMIZED}}

\noindent \begin{centering}
\includegraphics[width=0.49\textwidth]{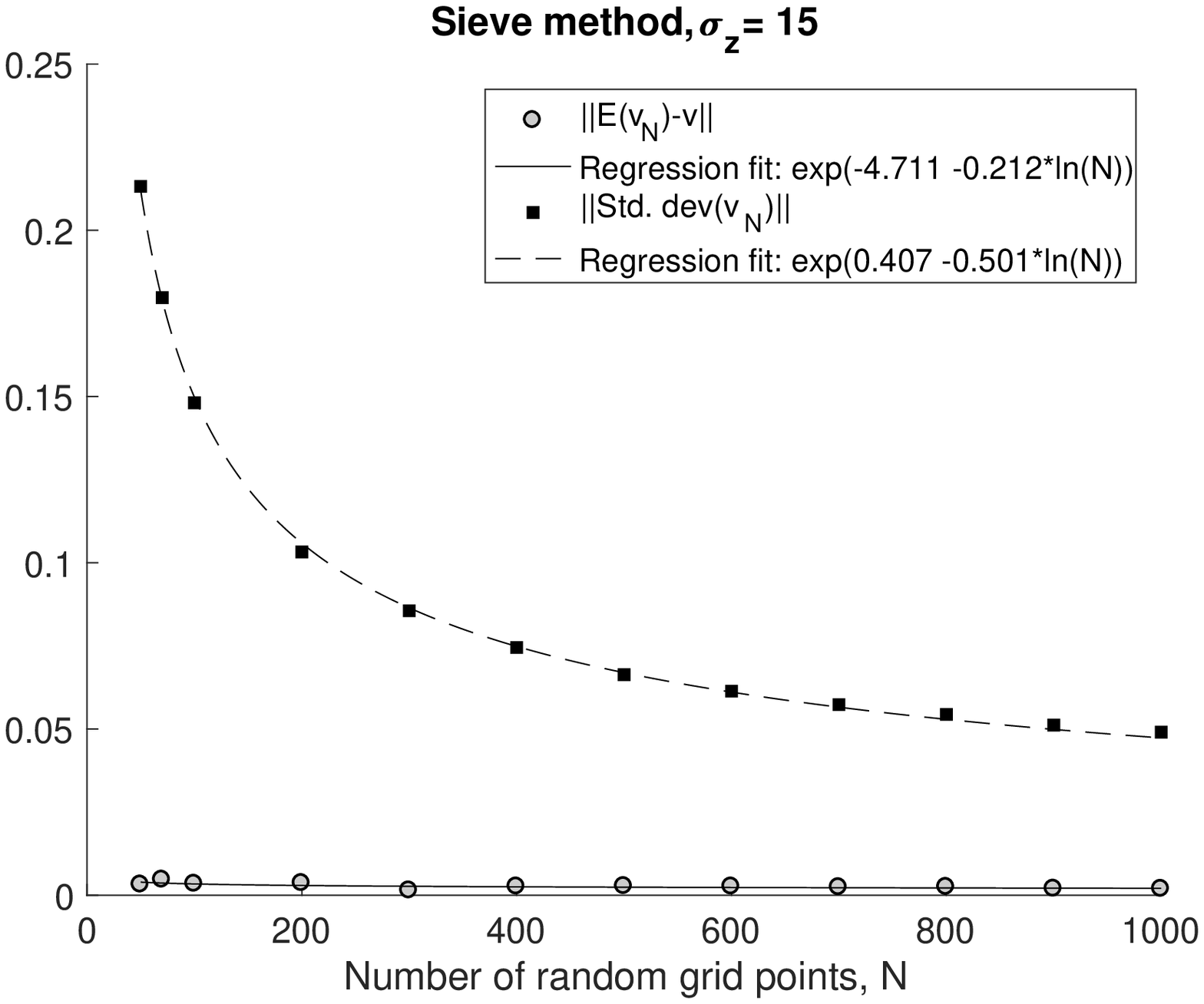}\includegraphics[width=0.49\textwidth]{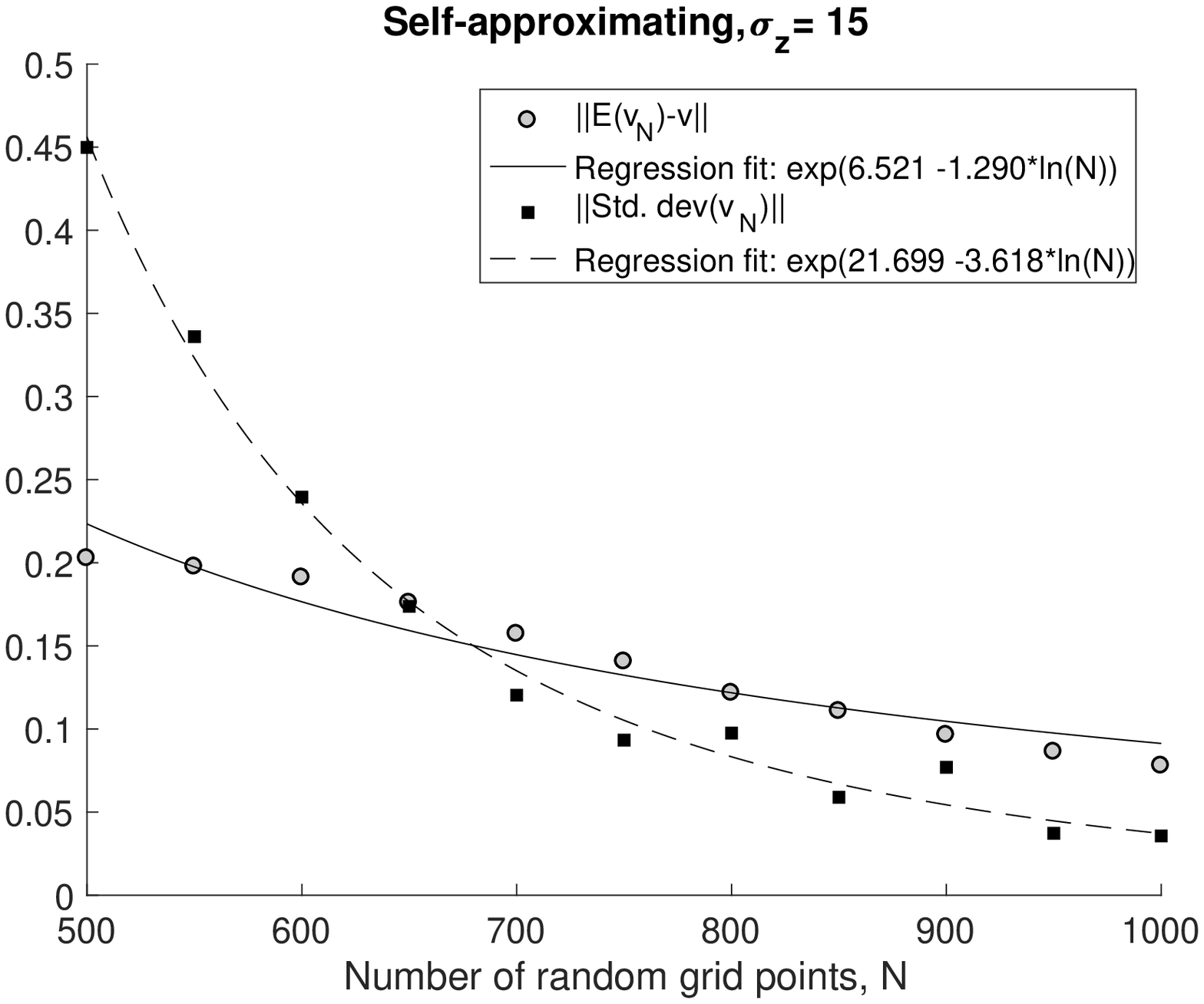}
\par\end{centering}
\noindent \begin{centering}
\includegraphics[width=0.49\textwidth]{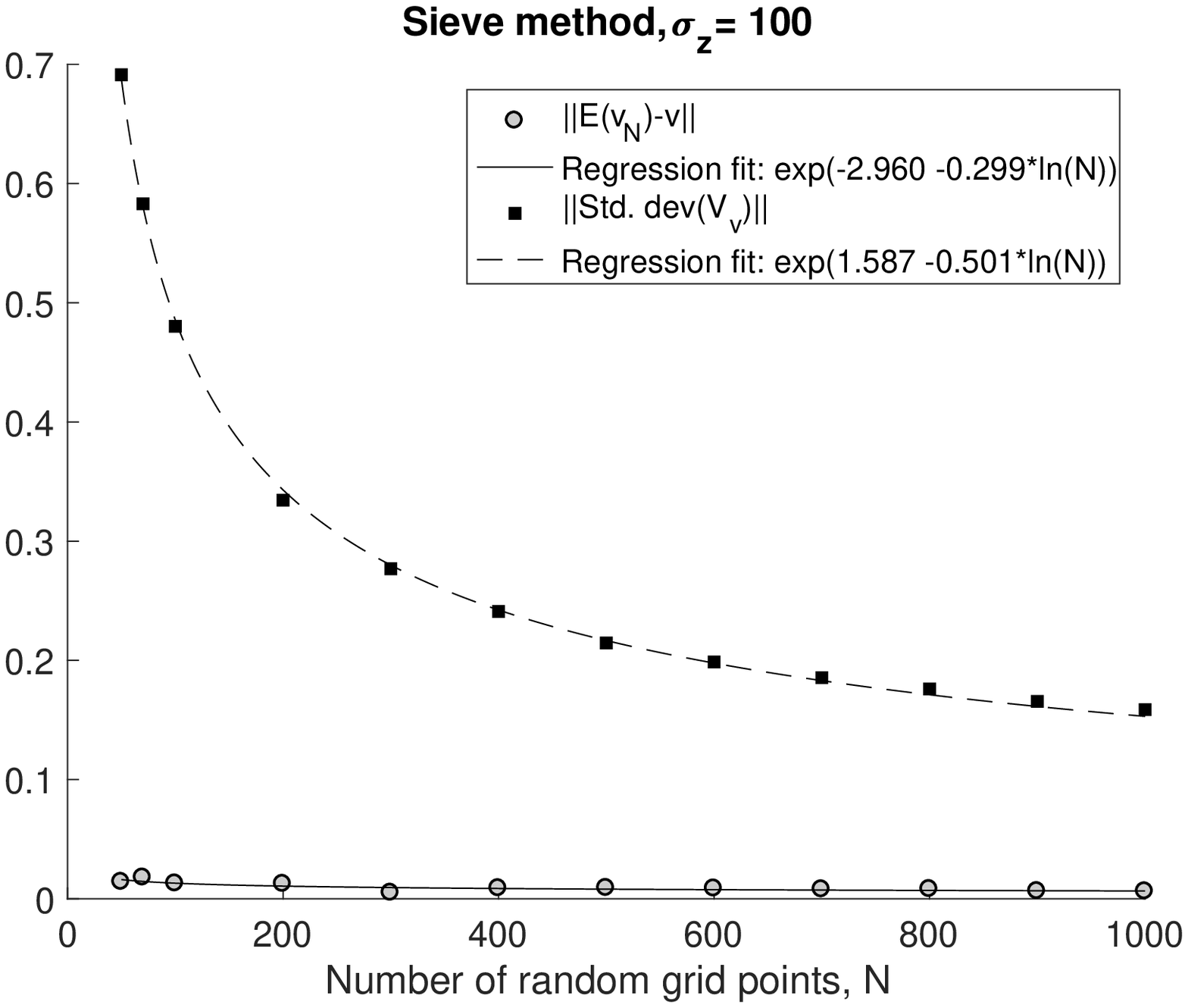}\includegraphics[width=0.49\textwidth]{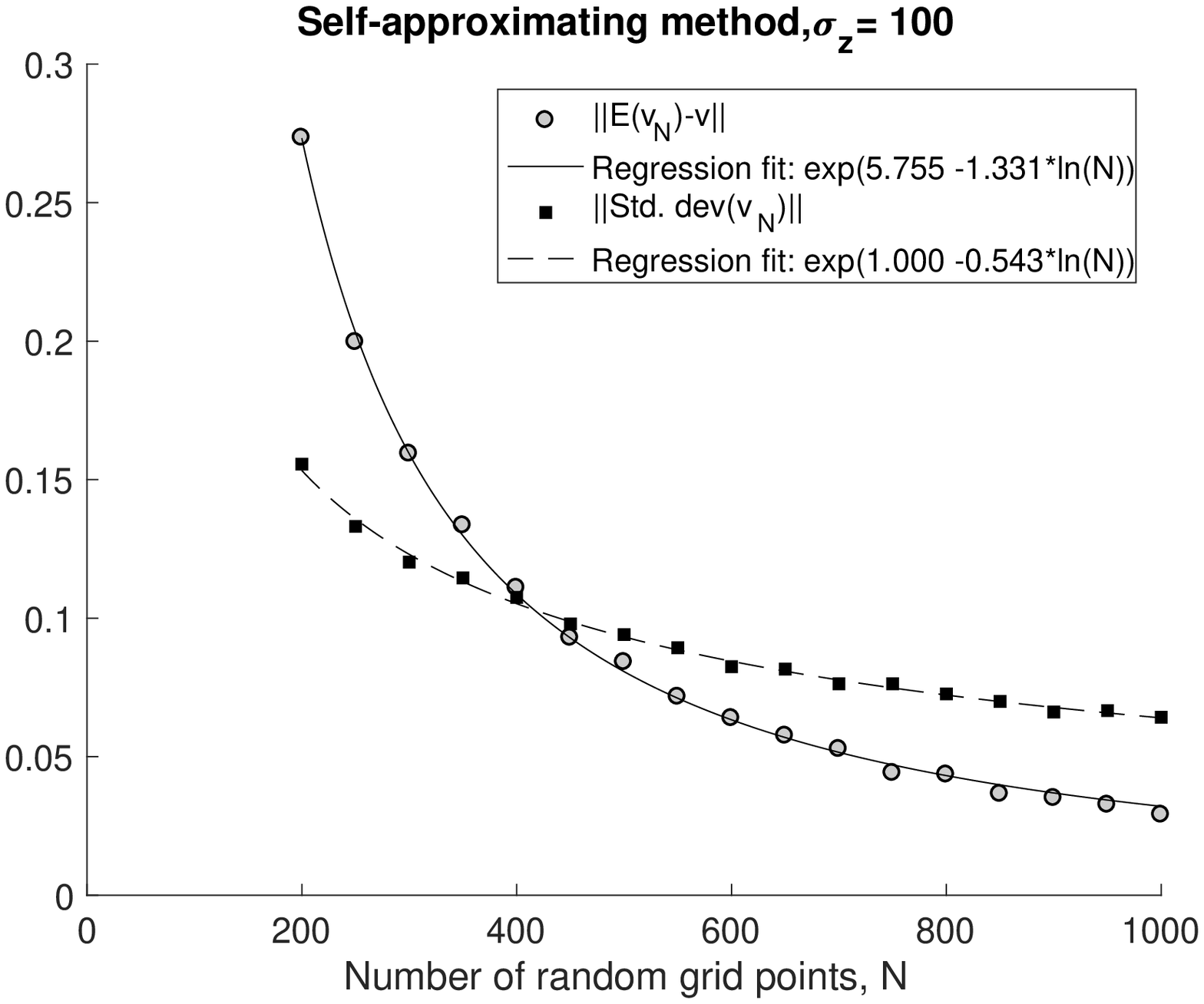}
\par\end{centering}
{\footnotesize{}Notes: Discount factor is $\text{\ensuremath{\beta}}=0.95$,
utility function parameters are $\theta_{c}=2$, $RC=10$, and transition
parameters are $a=2$, $b=5$ and $\pi=0.000000001$. Uniform bias
and variance were estimated using $500$ evaluation points and $S=2000$
implementations.}{\footnotesize\par}
\end{figure}

According to the theory, the variance should vanish with rate $1/N$
for both methods and we therefore expect $\rho_{SD}=-0.5$ so that
$||\sqrt{Var}||_{\infty}$ vanish with $1/\sqrt{N}$. For the projection
based method, we see that the rate with which the standard deviation
shrinks to zero is indeed close to $-0.5$ for all values of $K>1$
and irrespectively of the value of $\sigma_{\varepsilon}$. For the
self-approximating method we estimate the rate to $\rho_{SD}=-0.541$
when $\sigma_{\varepsilon}=100$, which is in line with the theory.
However, $||\sqrt{Var}||_{\infty}$ is found to vanish with rate $1/N^{3.6}$
for $\sigma_{\varepsilon}=15$. This seems to indicate that the asymptotic
theory developed in Theorems \ref{Thm: Rust approx rate} and \ref{Thm: Rust normal}
do not provide a very accurate approximation of the performance of
the self-approximating method for small and moderate choices of $N$
when the support of $Z_{t}|\left(Z_{t-1}=z,d_{t}=1\right)$ is small
($\sigma_{Z}=15$). We conjecture that the discrepancy between theoretical
predictions and numerical results for the self-approximating method
is due to the aforementioned issues with the marginal importance sampler
discussed in Section \ref{subsec:Implementation-of-Bellman}: Many
of the draws are not used in the computation of the simulated Bellman
operator because they fall outside the support of $Z_{t}|Z_{t-1}=z$
for a given choice of $z$. Thus, the effective number of draws is
smaller than $N$ and changes as $z$ varies.

\begin{table}
\caption{Bias, variance, and rates of convergence for various values of K\label{tab:ConvRates_by_K}}

\begin{tabular}{lrrrrrrr} 
\hline\hline 
& \multicolumn{5}{c}{Sieve Method} & \multicolumn{1}{c}{Self-approx.} \\ 
\# of basis functions, K &1& 2& 5& 10& 15& \multicolumn{1}{c}{method} \\ 
\hline 
& \multicolumn{6}{c}{$\sigma_{z}$=15} \\ 
$||Bias||_\infty$ for $N=500$ & 12.743 & 7.029 & 0.348 & 0.016 & 0.003 & 0.203\\ 
$||\sqrt{Var}||_\infty$ for $N=500$  & 0.000 & 0.020 & 0.063 & 0.066 & 0.066 & 0.450\\ 
Convergence rate for $||Bias||_\infty$ & 0.000 & 0.000 & 0.002 & -0.012 & -0.212 & -1.290\\ 
Convergence rate for $||\sqrt{Var}||_\infty$ & 0.169 & -0.500 & -0.501 & -0.501 & -0.501 & -3.618\\
\\
& \multicolumn{6}{c}{$\sigma_{z}$=100} \\ 
$||Bias||_\infty$ for $N=500$ &22.446 & 10.937 & 0.112 & 0.009 & 0.009 & 0.084 \\ 
$||\sqrt{Var}||_\infty$ for $N=500$ & 0.000 & 0.128 & 0.218 & 0.215 & 0.215 & 0.094\\ 
Convergence rate for $||Bias||_\infty$ & 0.000 & 0.000 & -0.027 & -0.299 & -0.299 & -1.331\\ 
Convergence rate for $||\sqrt{Var}||_\infty$ & 0.169 & -0.500 & -0.501 & -0.501 & -0.501 & -0.543\\ 
\hline\hline 
\end{tabular}
\end{table}

For the projection based method, the main source of bias is due to
the sieve projection. From Figure \ref{fig:Point-wise-bias}, we see
that, with $N=200$ and $K=9$, the sieve-based methods using second-order
B-splines or Chebyshev polynomials have virtually no bias, and both
Figure \ref{fig:RANDOMIZED} as well as in Table \ref{tab:ConvRates_by_K}
also confirms that we practically eliminate by approximate the value
function using Chebychev polynomials with $K=20$. However, there
still remains a small bias that vanishes as $N$ grows. For small
$K$, we see that the bias is roughly independent of $N$. As $K$
increases so does the dependence on $N$. However, even for $K=20$
where we estimate $\rho_{Bias}$ to be $0.21$ and $0.30$ for $\sigma_{Z}=15$
and $\sigma_{Z}=100$ respectively, the rate of convergence is far
from $1/N$. This is probably due to the presence of higher-order
bias components that our asymptotic theory does not account for.

For the self-approximating method, there is no sieve projection bias
but a larger simulation induced bias that decreases with $N$. We
obtain rate estimates of $1/N^{1.7}$ and $1/N^{1.4}$ for the bias
when $\sigma_{Z}=15$ and $\sigma_{Z}=100$ respectively; these are
slightly faster than expected but not too far from the theoretical
predictions of $1/N$. For the self-approximating method, bias constitute
more than half of RMSE when $N<600$ for $\text{\ensuremath{\sigma_{Z}=15}}$
(or $N<400$ for $\text{\ensuremath{\sigma_{Z}=100}}$), but since
$||Bias||_{\infty}$ decays faster than $||\sqrt{Var}||_{\infty}$
, the simulation bias eventually becomes second order for large $N$.

Comparing $||MSE||_{\infty}$ for $N=500$ we find that the sieve-based
method clearly dominates the self-approximating method when $\sigma_{Z}=15$,
whereas the self-approximating method performs best when $\sigma_{Z}=100$.
This is not entirely surprising since a large value of $\sigma_{\varepsilon}$
implies a large conditional support of $Z_{t}$ in which case the
draws of the marginal sampler are more likely to fall within the support,
c.f. the discussion in Subsection \ref{subsec:Implementation-of-Bellman}.
Thus, the over-all error of the self-approximating method will tend
to be smaller when $\sigma_{Z}$ is large. The opposite is the case
for the sieve based method which becomes more precise for smaller
value of $\sigma_{Z}$ since the variance of the simulated Bellman
operator used for this method gets smaller as $\sigma_{Z}$ gets smaller.
This shows that there is considerably scope for improving the performance
of the sieve-based method by more careful design of the sampling method.

\begin{figure}[h]
\caption{Asymptotic Normality\label{fig:NORMALITY_RAND}}

\noindent \begin{centering}
\includegraphics[width=0.49\textwidth]{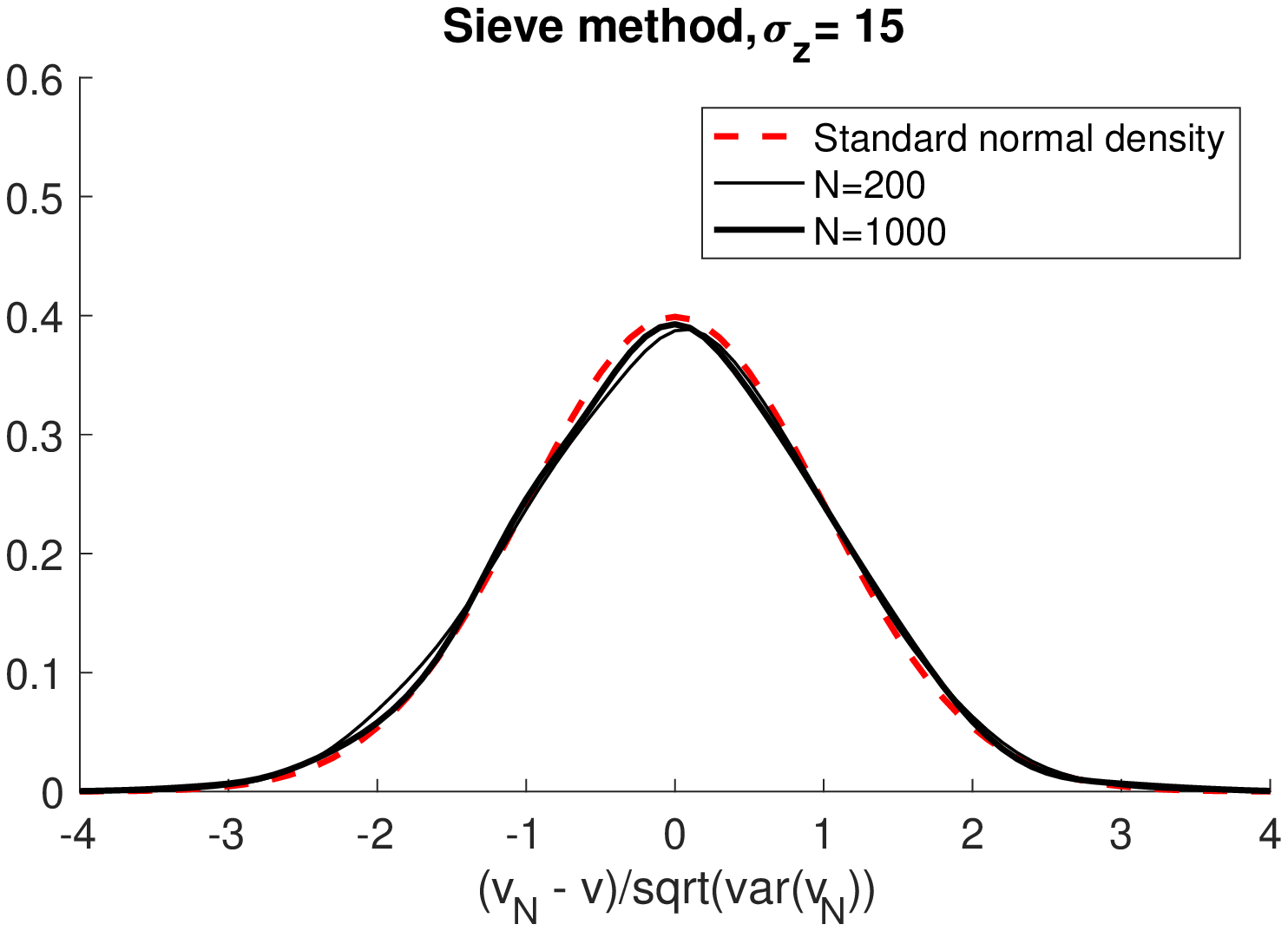}\includegraphics[width=0.49\textwidth]{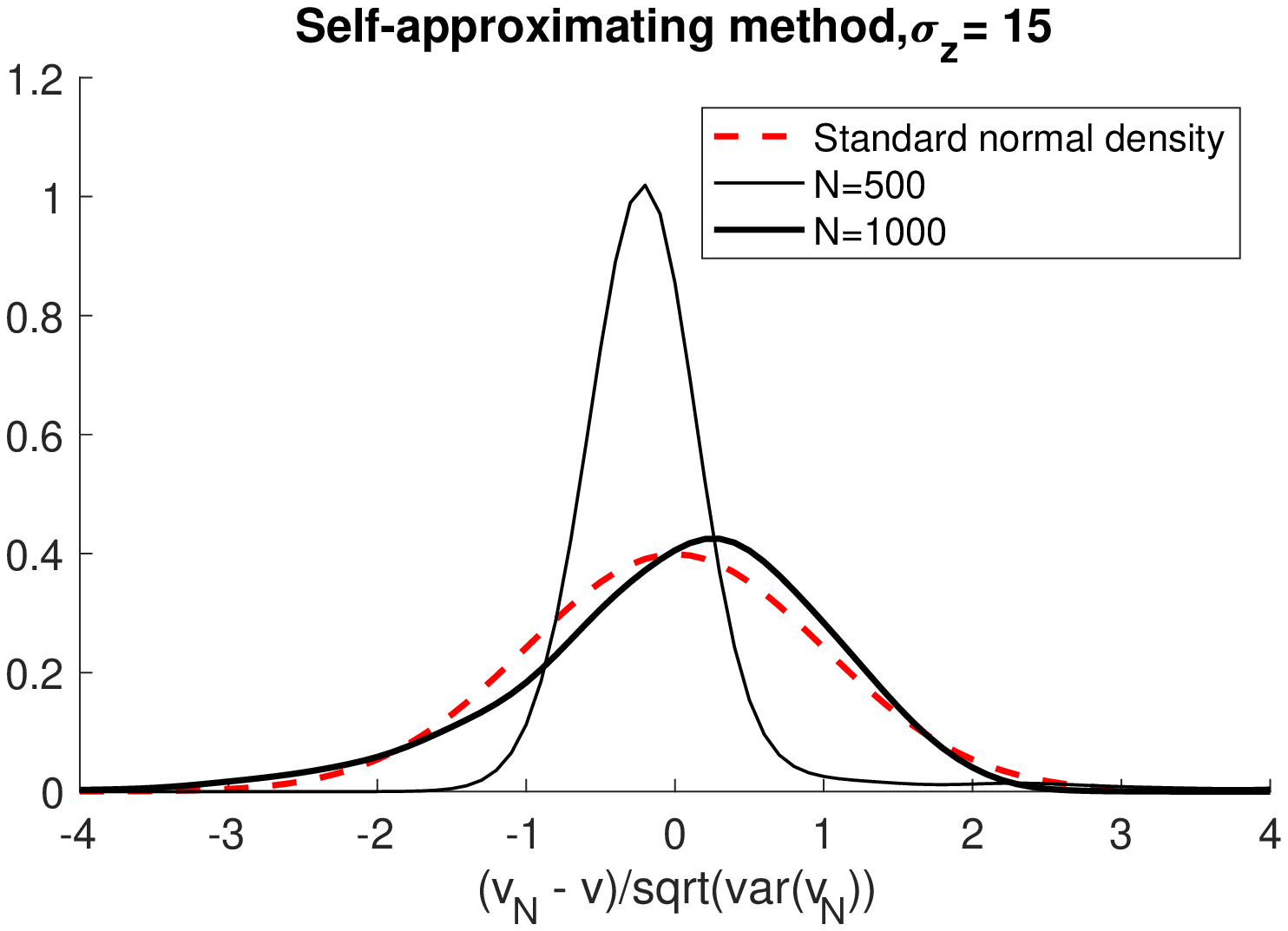}
\par\end{centering}
\noindent \begin{centering}
\includegraphics[width=0.49\textwidth]{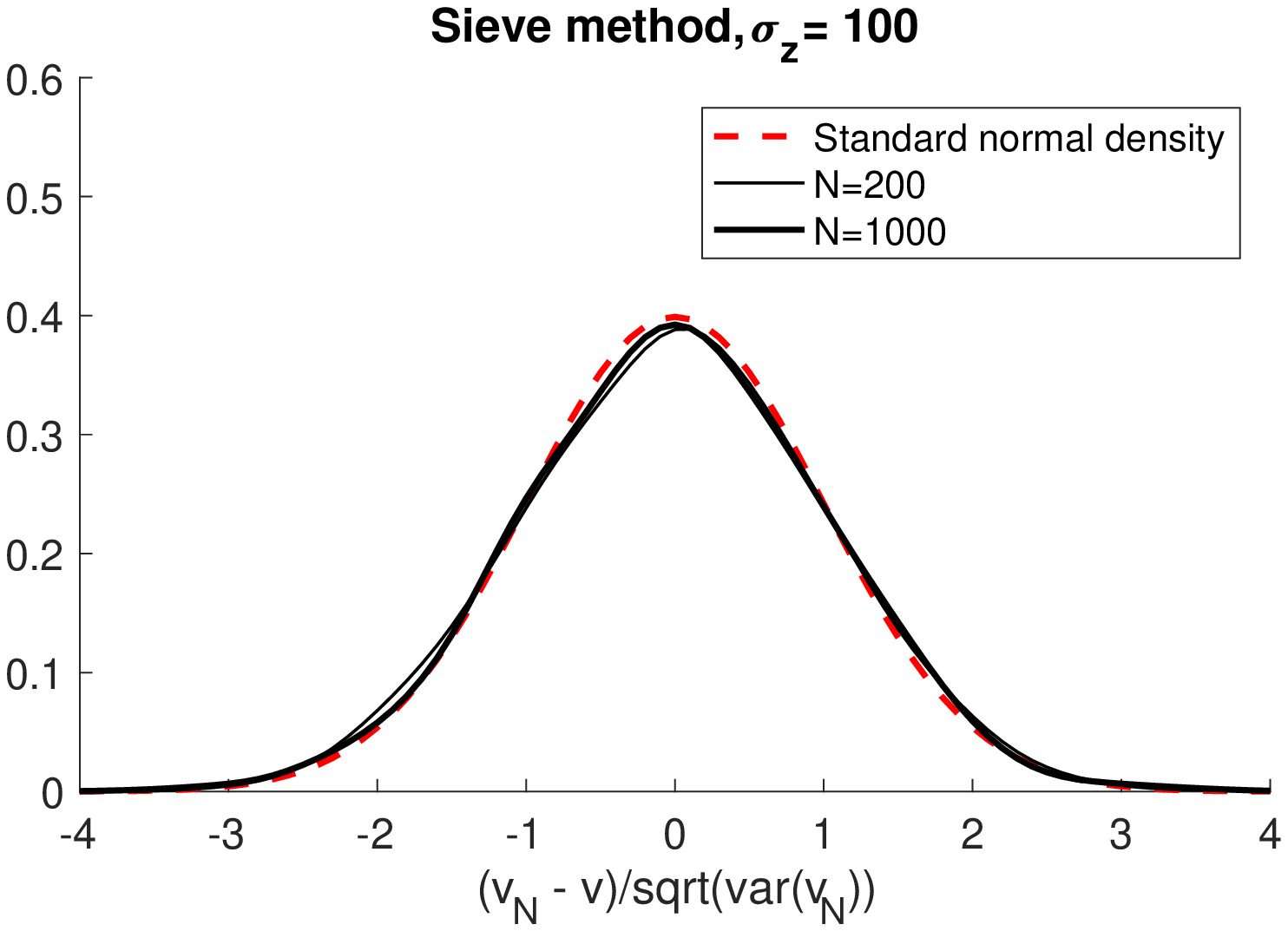}\includegraphics[width=0.49\textwidth]{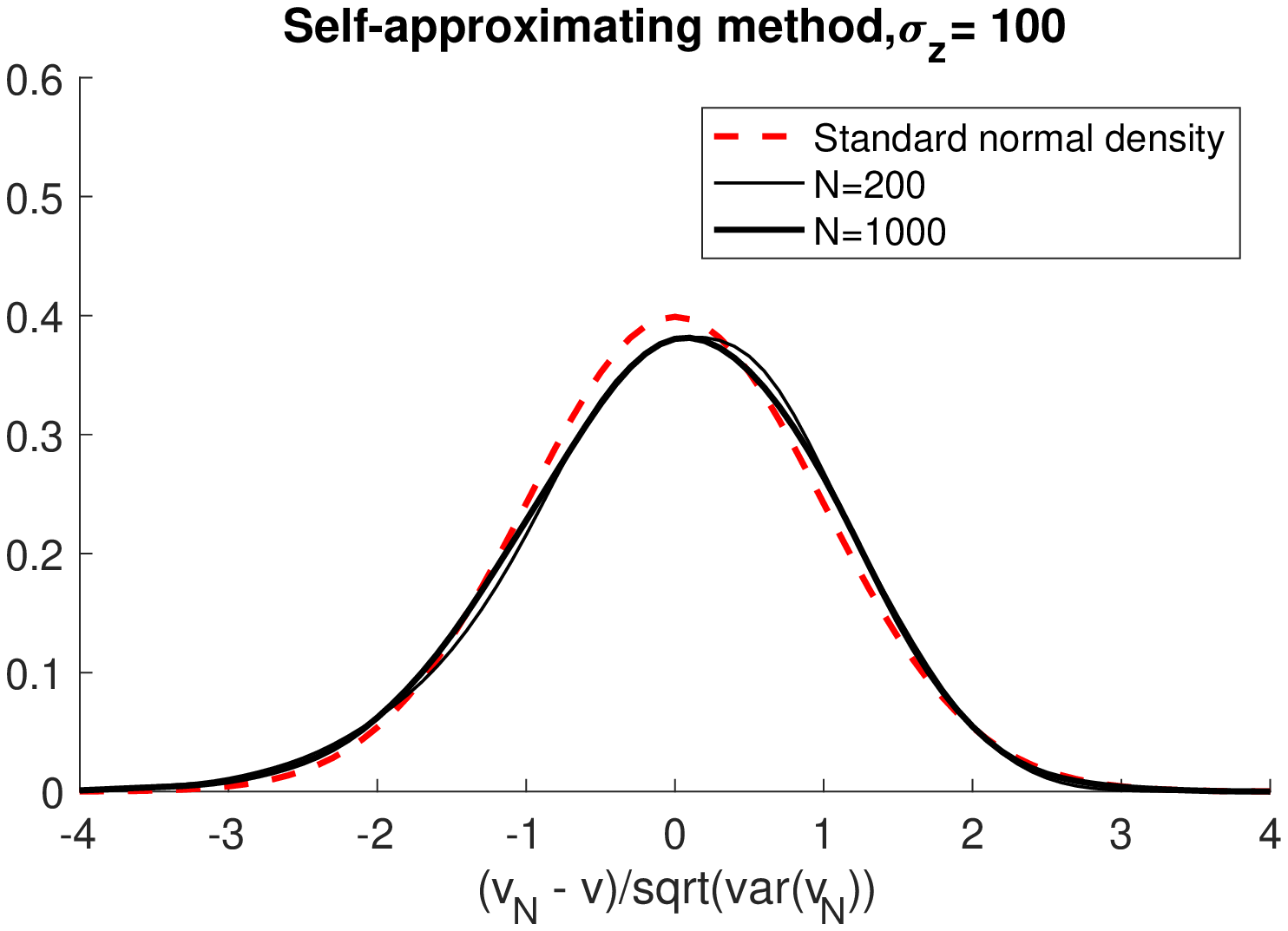}
\par\end{centering}
{\footnotesize{}Notes: Each panel shows kernel density estimates of
$(\widehat{v}_{N}(z)-E[\widehat{v}_{N}(z)])/\sqrt{var(\hat{v}_{N}(z))}$
for $z=500$ based on $S=2000$ solutions for each sample size $N$.
Discount factor is $\text{\ensuremath{\beta}}=0.95$, utility function
parameters are $\theta_{c}=2$, $RC=10$, and transition parameters
are $\sigma_{Z}=100$, $a=2$, $b=5$ and $\pi=0$.000000001.}{\footnotesize\par}
\end{figure}

Theorems \ref{Thm: Rust normal} and \ref{thm: proj rate} state that
when $N$ is large, the approximate value functions should be normally
distributed. We here investigate whether this asymptotic approximation
is useful in practice by looking at the pointwise distribution of
the approximate solutions obtained through both methods. In Figure
\ref{fig:NORMALITY_RAND}, we plot the distribution of $\left(\tilde{v}\left(z\right)-E\left[\tilde{v}\left(z\right)\right]\right)/\sqrt{Var\left(\tilde{v}\left(z\right)\right)}$
for $z=500$, where $\tilde{v}$ denotes a given approximation method,
together with the standard normal distribution. It is here important
to note we do not center the estimate around $v\left(z\right)$ but
instead around $E[\tilde{v}(z)]$ ; this is due to the sizable bias
of the self-approximating method. For the sieve-based method, we see
that its normalized distribution is quite close to the standard normal
irrespectively of the value of $\sigma_{Z}$. In contrast, the normal
distribution is a poor approximation for the self-approximating method
when $\sigma_{Z}=15$ when $N=500$; we expect this is due to the
fact that the effective number of draws is quite small and so the
asymptotic approximation is poor in this case. As expected the approximation
gets better as $N$ and/or $\sigma_{Z}$ increases.

\subsection{\label{subsec:smoothng}Effect of smoothing}

\begin{figure}
\caption{Sup-norm MSE of solutions to Bellman operators with simulated taste
shocks and state transitions for varying levels of smoothing, for
$N=100$.}
\label{fig:tastesmoothing}\begin{tikzpicture}[]
\begin{axis}[height = {50.8mm}, legend pos = {north west}, ylabel = {$||MSE||_\infty$}, xlabel={Smoothing, $\lambda$}, xmin = {-0.03}, xmax = {1.03}, ymax = {5.065717913977283}, , unbounded coords=jump,scaled x ticks = false,xlabel style = {font = {\fontsize{11 pt}{14.3 pt}\selectfont}, color = {rgb,1:red,0.00000000;green,0.00000000;blue,0.00000000}, draw opacity = 1.0, rotate = 0.0},xmajorgrids = true,xtick = {0.0,0.25,0.5,0.75,1.0},xticklabels = {$0.00$,$0.25$,$0.50$,$0.75$,$1.00$},xtick align = inside,xticklabel style = {font = {\fontsize{8 pt}{10.4 pt}\selectfont}, color = {rgb,1:red,0.00000000;green,0.00000000;blue,0.00000000}, draw opacity = 1.0, rotate = 0.0},x grid style = {color = {rgb,1:red,0.00000000;green,0.00000000;blue,0.00000000},
draw opacity = 0.1,
line width = 0.5,
solid},axis lines* = left,x axis line style = {color = {rgb,1:red,0.00000000;green,0.00000000;blue,0.00000000},
draw opacity = 1.0,
line width = 1,
solid},scaled y ticks = false,ylabel style = {font = {\fontsize{11 pt}{14.3 pt}\selectfont}, color = {rgb,1:red,0.00000000;green,0.00000000;blue,0.00000000}, draw opacity = 1.0, rotate = 0.0},ymajorgrids = true,ytick = {4.800000000000001,4.8500000000000005,4.9,4.95,5.0,5.050000000000001},yticklabels = {$4.80$,$4.85$,$4.90$,$4.95$,$5.00$,$5.05$},ytick align = inside,yticklabel style = {font = {\fontsize{8 pt}{10.4 pt}\selectfont}, color = {rgb,1:red,0.00000000;green,0.00000000;blue,0.00000000}, draw opacity = 1.0, rotate = 0.0},y grid style = {color = {rgb,1:red,0.00000000;green,0.00000000;blue,0.00000000},
draw opacity = 0.1,
line width = 0.5,
solid},axis lines* = left,y axis line style = {color = {rgb,1:red,0.00000000;green,0.00000000;blue,0.00000000},
draw opacity = 1.0,
line width = 1,
solid},    xshift = 0.0mm,
    yshift = 0mm,
    axis background/.style={fill={rgb,1:red,1.00000000;green,1.00000000;blue,1.00000000}}
,legend style = {color = {rgb,1:red,0.00000000;green,0.00000000;blue,0.00000000},
draw opacity = 1.0,
line width = 1,
solid,fill = {rgb,1:red,1.00000000;green,1.00000000;blue,1.00000000},font = {\fontsize{8 pt}{10.4 pt}\selectfont}},colorbar style={title=}, ymin = {4.788232625982812}, width = {76.2mm}]\addplot+[draw=none, color = {rgb,1:red,0.00000000;green,0.00000000;blue,0.00000000},
draw opacity = 1.0,
line width = 0,
solid,mark = *,
mark size = 2.0,
mark options = {
    color = {rgb,1:red,0.00000000;green,0.00000000;blue,0.00000000}, draw opacity = 1.0,
    fill = {rgb,1:red,0.00000000;green,0.00000000;blue,0.00000000}, fill opacity = 1.0,
    line width = 1,
    rotate = 0,
    solid
}] coordinates {
(0.0, 4.796085983190203)
(0.05, 4.796492241401615)
(0.1, 4.797448072907024)
(0.15, 4.798767044433014)
(0.2, 4.800431024840986)
(0.25, 4.8024450177151214)
(0.3, 4.804803860401239)
(0.35, 4.807497304003023)
(0.4, 4.810524076556262)
(0.5, 4.817673368727176)
(0.6, 4.826694435794043)
(0.7, 4.838424850270665)
(0.8, 4.854038184982356)
(0.9, 4.874989386443136)
(1.0, 5.057864556769893)
};
\addlegendentry{K=4}
\end{axis}
\begin{axis}[height = {50.8mm}, legend pos = {north west}, , ylabel = {$||MSE||_\infty$}, xlabel={Smoothing, $\lambda$},  xmin = {-0.03}, xmax = {1.03}, ymax = {2.8802126747033445}, unbounded coords=jump,scaled x ticks = false,xlabel style = {font = {\fontsize{11 pt}{14.3 pt}\selectfont}, color = {rgb,1:red,0.00000000;green,0.00000000;blue,0.00000000}, draw opacity = 1.0, rotate = 0.0},xmajorgrids = true,xtick = {0.0,0.25,0.5,0.75,1.0},xticklabels = {$0.00$,$0.25$,$0.50$,$0.75$,$1.00$},xtick align = inside,xticklabel style = {font = {\fontsize{8 pt}{10.4 pt}\selectfont}, color = {rgb,1:red,0.00000000;green,0.00000000;blue,0.00000000}, draw opacity = 1.0, rotate = 0.0},x grid style = {color = {rgb,1:red,0.00000000;green,0.00000000;blue,0.00000000},
draw opacity = 0.1,
line width = 0.5,
solid},axis lines* = left,x axis line style = {color = {rgb,1:red,0.00000000;green,0.00000000;blue,0.00000000},
draw opacity = 1.0,
line width = 1,
solid},scaled y ticks = false,ylabel style = {font = {\fontsize{11 pt}{14.3 pt}\selectfont}, color = {rgb,1:red,0.00000000;green,0.00000000;blue,0.00000000}, draw opacity = 1.0, rotate = 0.0},ymajorgrids = true,ytick = {2.6500000000000004,2.7,2.75,2.8000000000000003,2.85},yticklabels = {$2.65$,$2.70$,$2.75$,$2.80$,$2.85$},ytick align = inside,yticklabel style = {font = {\fontsize{8 pt}{10.4 pt}\selectfont}, color = {rgb,1:red,0.00000000;green,0.00000000;blue,0.00000000}, draw opacity = 1.0, rotate = 0.0},y grid style = {color = {rgb,1:red,0.00000000;green,0.00000000;blue,0.00000000},
draw opacity = 0.1,
line width = 0.5,
solid},axis lines* = left,y axis line style = {color = {rgb,1:red,0.00000000;green,0.00000000;blue,0.00000000},
draw opacity = 1.0,
line width = 1,
solid},    xshift = 76.2mm,
    yshift = 0.0mm,
    axis background/.style={fill={rgb,1:red,1.00000000;green,1.00000000;blue,1.00000000}}
,legend style = {color = {rgb,1:red,0.00000000;green,0.00000000;blue,0.00000000},
draw opacity = 1.0,
line width = 1,
solid,fill = {rgb,1:red,1.00000000;green,1.00000000;blue,1.00000000},font = {\fontsize{8 pt}{10.4 pt}\selectfont}},colorbar style={title=}, ymin = {2.637865408625582}, width = {76.2mm}]\addplot+[draw=none, color = {rgb,1:red,0.00000000;green,0.00000000;blue,0.00000000},
draw opacity = 1.0,
line width = 0,
solid,mark = *,
mark size = 2.0,
mark options = {
    color = {rgb,1:red,0.00000000;green,0.00000000;blue,0.00000000}, draw opacity = 1.0,
    fill = {rgb,1:red,0.00000000;green,0.00000000;blue,0.00000000}, fill opacity = 1.0,
    line width = 1,
    rotate = 0,
    solid
}] coordinates {
(0.0, 2.677144934699648)
(0.05, 2.676825145511437)
(0.1, 2.6756882707789655)
(0.15, 2.6738174269468815)
(0.2, 2.6712729219778604)
(0.25, 2.6681016162234856)
(0.3, 2.664383152967328)
(0.35, 2.660264082897299)
(0.4, 2.656047322315762)
(0.5, 2.6483382027836524)
(0.6, 2.6447242935145754)
(0.7, 2.650853170736444)
(0.8, 2.673730139721892)
(0.9, 2.7241009122398214)
(1.0, 2.873353789814351)
};
\addlegendentry{K=15}
\end{axis}

\end{tikzpicture}
\end{figure}
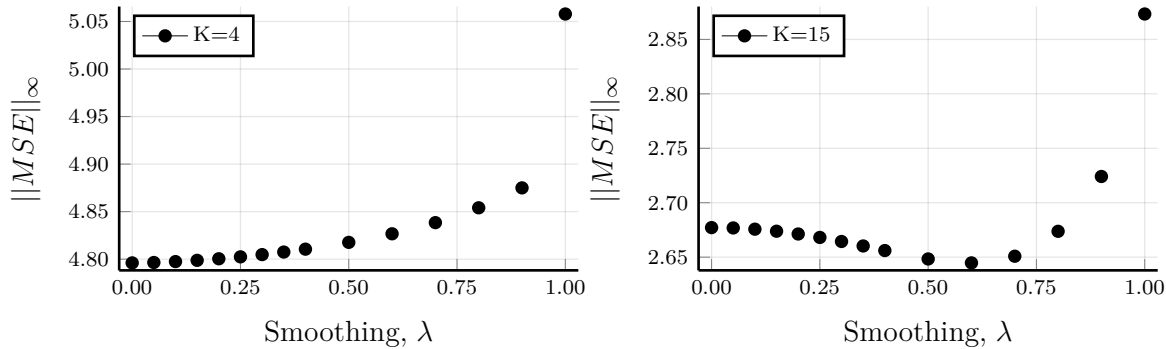

The results reported above did not involve any smoothing bias. We
now numerically study the effect of smoothing. This is done by, instead
of integrating out the i.i.d. extreme value taste shocks $\varepsilon_{t}$
analytically as we have done so far, using Monte Carlo simulations
to evaluate this part of the integral and then introducing the smoothing
device to ensure that the simulated Bellman operator remains smooth.
While this may appear somewhat artificial, the merit of doing this
exercise is that we can use the same ``exact'' solution as benchmark
as used above.

In Figure \ref{fig:tastesmoothing} we plot the sup-norm of the mean
squared error, $||\text{MSE}||_{\infty}$ as a function of $\lambda$
(the smoothing scale parameter) for the sieve-based method using $K=4$
or $8$ Chebyshev polynomials (similar results were obtained for the
self-approximating method and so are left out). For $K=4$, the MSE
increases monotonically as a function of $\lambda$ while for $K=15$
the bias due to smoothing is non-monotonic in $\lambda$. In both
cases, at $\lambda=0$, any remaining biases are due to either sieve-approximation
or simulations. Importantly, the bias due to smoothing is negiglible
(relative to the other biases) for small and moderate values of $\lambda$
while the variance is largely unaffected. We have no theory or heuristics
for choosing an optimal $\lambda$ to optimally balance bias and variance
due to smoothing but the current numerical results indicate that choosing
a quite small $\lambda$ value works well.

\subsection{Performance in the bivariate case}

We now examine how the solution methods perform in the bivariate case
($d_{Z}=\dim\left(Z_{t}\right)=2$) in order to see if there is any
curse of dimensionality built into the two methods. We do this for
two different models as described below.

\subsubsection*{An Additive DDP}

We here follow the approach of \citet{Arcidiaconoetal2013} and \citet{rust1997comparison}
and build a $d_{Z}$-dimensional model by adding up $d_{Z}$ independent
versions of the univariate model considered so far. That is, we choose
the utilities and state dynamics as $\bar{u}(z,\varepsilon,d)=\sum_{i=1}^{d_{Z}}u(z_{i},\varepsilon_{i},d_{i})$
and $\bar{F}_{Z}(z^{\prime}|z,d)=\prod_{i=1}^{d_{Z}}F_{Z}(z_{i}^{\prime}|z_{i},d_{i})$,
where $z=\left(z_{1},...,z_{d_{z}}\right)$ and $d=\left(d_{1},...,d_{d_{Z}}\right)$,
with $F_{Z}(z_{i}^{\prime}|z_{i},d_{i})$ and $u(z_{i},\varepsilon_{i},d_{i})$
denoting the state transition and per-period utility in the univariate
case as described in Section \ref{subsec: Rust model}. Note here
that $\left(Z_{t,i},\varepsilon_{t,i}\right)$ and $\left(Z_{t,j},\varepsilon_{t,j}\right)$
are fully independent of each other, $i\neq j$ and the number of
alternatives are $2^{d_{Z}}$, where $d_{Z}$ Thus, the model considers
the joint replacement decision of $d_{Z}$ assets whose stochastic
usages $\left(Z_{t,1},...,Z_{t,d_{Z}}\right)$ are mutually independent.
Conveniently, the integrated value function of this multidimensional
problem, $\bar{v}(z_{1},...,z_{d_{Z}})$, is simply the sum of the
solutions to each of the underlying univariate models, $\bar{v}(z_{1},...,z_{d_{Z}})=\sum_{i=1}^{d_{Z}}v(z_{i})$,
where $v(z_{i})$ is the solution to the univariate model in Section
\ref{subsec: Rust model}. This is a rather simplistic multivariate
model but it comes with the major advantage that we can obtain a very
accurate approximation of the exact solution by simply adding up the
``exact'' solution found for the univariate case. With a more complicated
multidimensional structure, the computational cost of finding the
``exact'' solution is much higher. However, when implementing our
solution methods, we forgo forgo the knowledge of the additive structure
of the solution and so treat the above model as a ``proper'' multivariate
problem.

\subsubsection*{Simulation error}

Given the issues with the self-approximating method for small values
of $\sigma_{Z}=15$, we here focus exclusively on the case $\sigma_{Z}=100$.
To get a sense of the pointwise performance of the self-approximating
method, we plot the pointwise bias and standard deviation for this
method with $N=3000$ in Figure \ref{fig:RANDOMIZED-2D-ERROR} together
with the pointwise errors of the corresponding replacement (choice)
probabilities. The overall shape and level of the integrated value
function is quite well captured, and the same is true for the policy.
However, the approximation errors tend to get larger out in the tails
of the distribution and some of this comes from the fact that the
issues with the marginal sampler used for the self-approximating method
are amplified here. The problems are especially present in the off-grid
evaluations, where we often have very few draws in a given region
where we want to evaluate the value function or policies.

\begin{figure}
\caption{Approximation errors of self-approximating method, bivariate DDP \label{fig:RANDOMIZED-2D-ERROR}}

\includegraphics[width=0.49\textwidth]{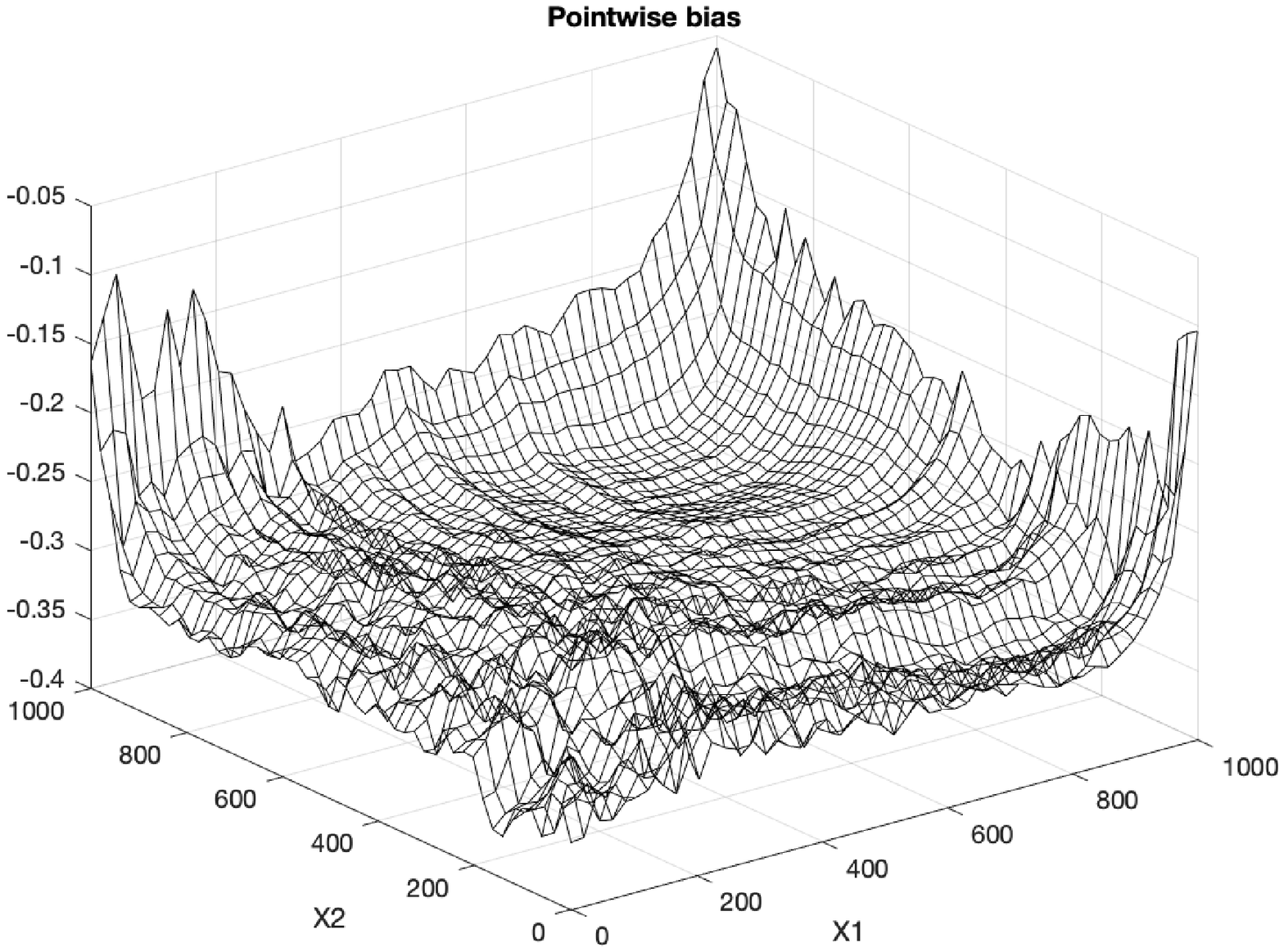}\includegraphics[width=0.49\textwidth]{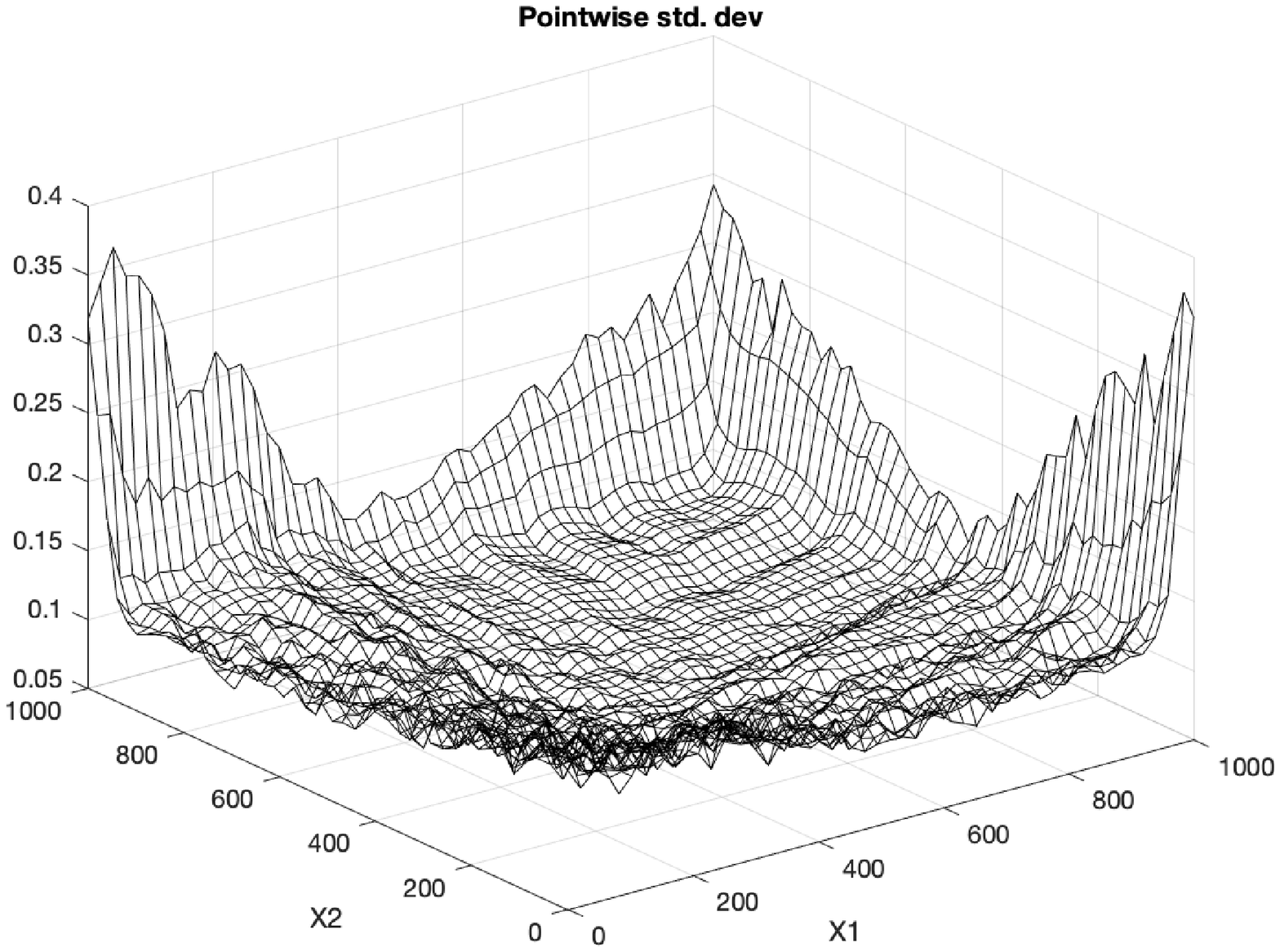}

{\footnotesize{}Notes: Discount factor is $\text{\ensuremath{\beta}}=0.95$,
utility function parameters are $\theta_{c}=2$, $RC=1$0, $\lambda=1$
and transition parameters are $\sigma_{Z}=100$, $a=2$, $b=5$ and
$\pi=0$.000000001. The ``exact'' solution was computed by averaging
over $S=100$ solutions, each found using the smoothed random Bellman
operator with $N=3000$ pseudo random draws. Each fixed point was
found using a contraction tolerance of machine precision. }{\footnotesize\par}
\end{figure}

Next, we examine $||Bias||_{\infty}$ and $||\sqrt{Var}||_{\infty}$
for both methods as we increase $N$. These are plotted in Figure
\ref{fig:RANDOMIZED-2D-CONVERGENCE} where it should be noted that
the reported range of $N$ reported on the $x$-axis of the two figures
differ substantially. This is due to the fact that the self\textendash approximating
method became numerically unstable for $N$ smaller than 1,400 while
no such issues were present for the sieve-based method. As in the
univariate case, both bias and variance of the two methods vanish
as $N$ increases. However, comparing Figures \ref{fig:RANDOMIZED-2D-CONVERGENCE}
and \ref{fig:RANDOMIZED}, while the errors of the sieve-based method
in the bivariate case is of a similar magnitude as in the univariate
case, the errors of the self-approximating method are much larger
in the bivariate case. This seems to indicate a certain type of curse-of-dimensionality
in this particular application of the self-approximating method. This
is caused by the issues with the marginal importance sampler employed
for this method.

\begin{figure}
\caption{Simulation errors for bivariate additive DDP \label{fig:RANDOMIZED-2D-CONVERGENCE}}

\noindent \begin{centering}
\includegraphics[width=0.5\textwidth]{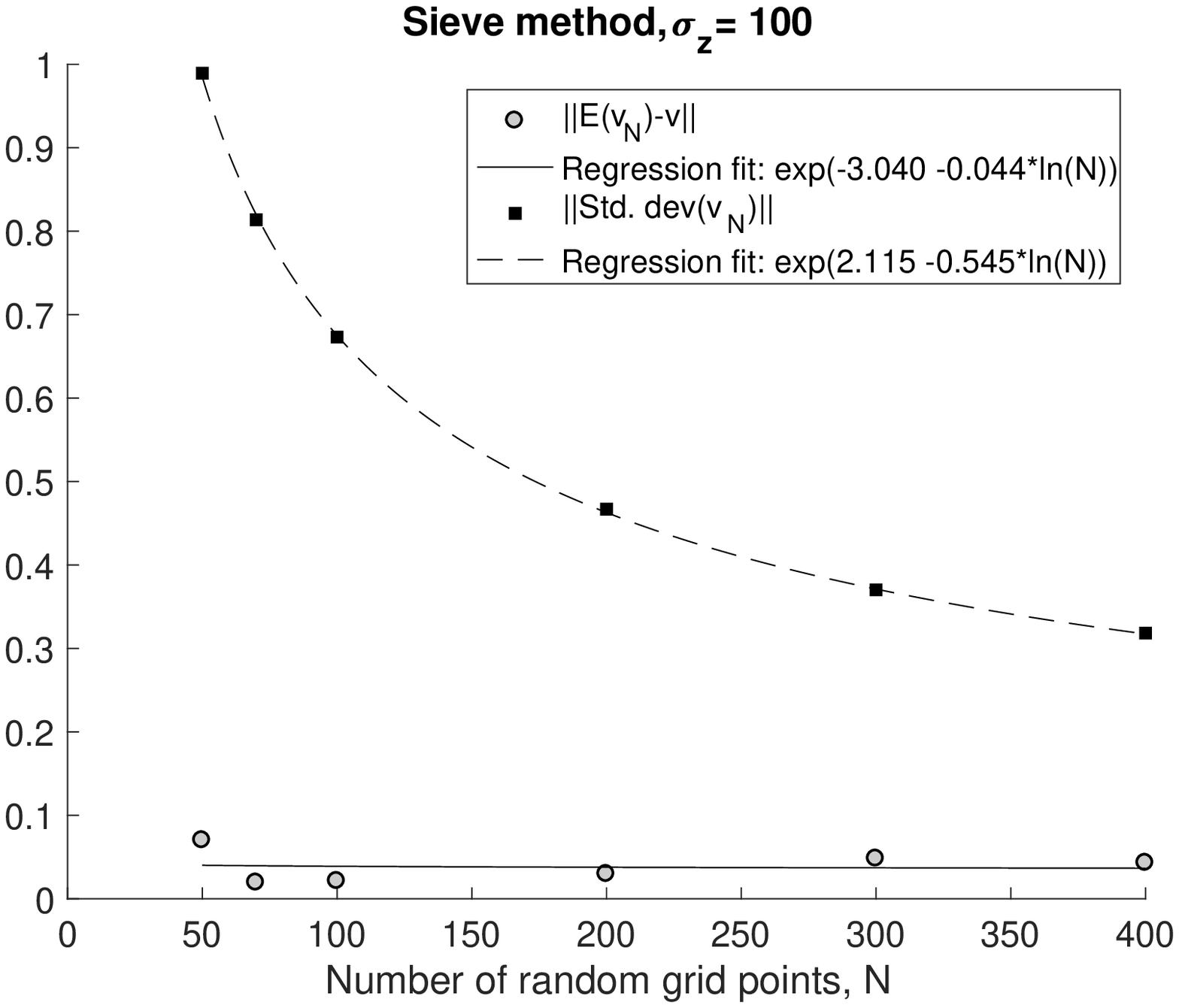}\includegraphics[width=0.5\textwidth]{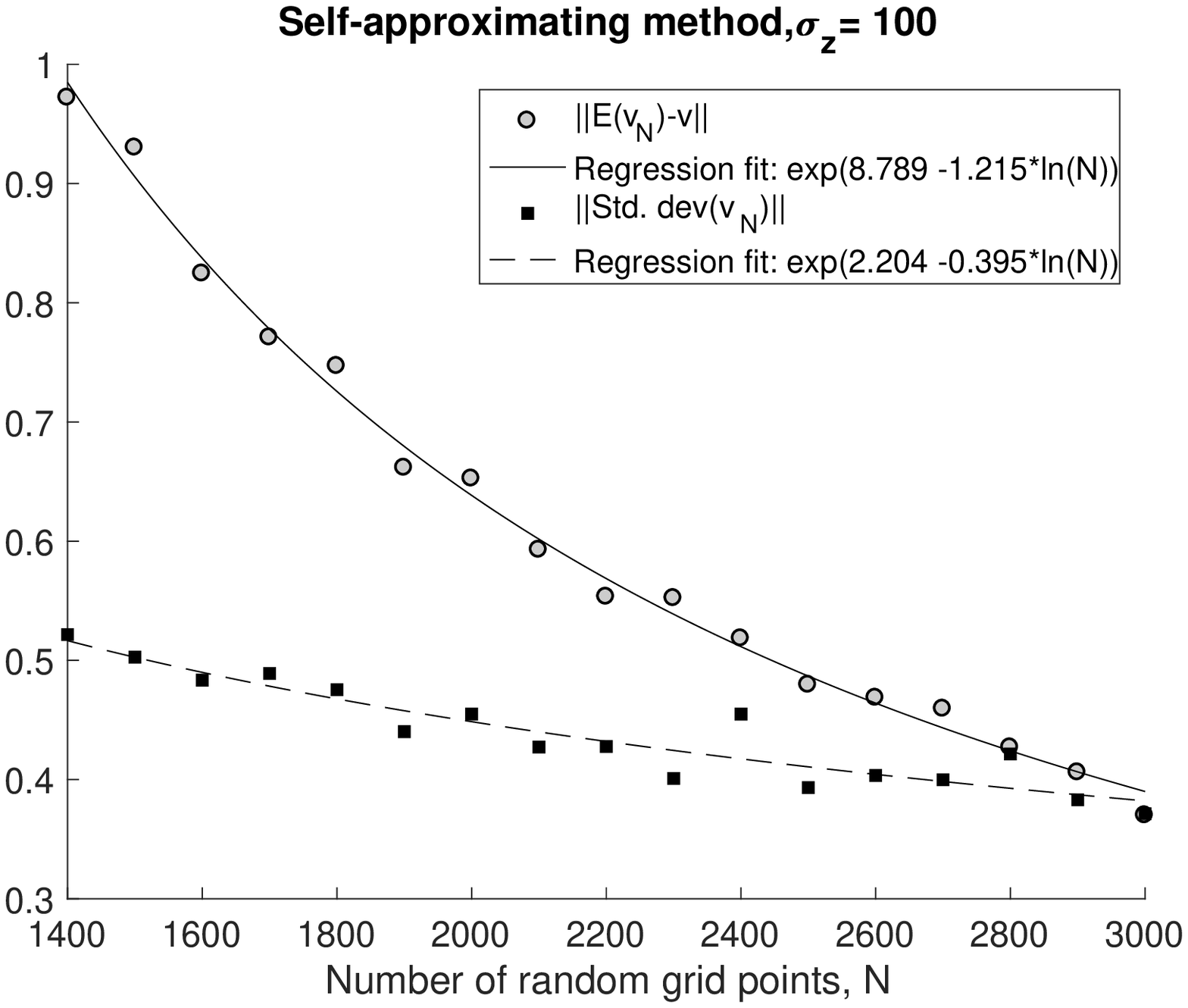}
\par\end{centering}
{\footnotesize{}Notes: Discount factor is $\text{\ensuremath{\beta}}=0.95$,
utility function parameters are $\theta_{c}=2$, $RC=10$, $\lambda=1$
and and parameters for transition density $f(z'|z,d)$ are $\sigma_{Z}=100$,
$a=2$, $b=5$ and $\pi=0$.000000001. Point-wise bias and variance
was estimated in $500$ evaluation points based on $S=200$ replications.
We report the sup norm of the bias and the standard deviation for
each N for both methods and NLS regression fits of $||\sqrt{{Var\}}}||_{\infty}=\exp(\alpha_{SD}+\rho_{SD}\ln(N)$
and $||Bias||_{\infty}=\exp(\alpha_{Bias}+\rho_{Bias}\ln(N)$.}{\footnotesize\par}
\end{figure}

\subsubsection*{Sieve approximation error}

In the implementation of the sieve-based method we use as sieve basis
the tensor product of univariate Chebyshev polynomials or B-splines.
That is, given, say, $J$ univariate basis functions, say, $p_{1},...,p_{J}$,
we construct our bivariate basis functions as $B_{i,j}\left(z_{1},z_{2}\right)=p_{i}\left(z_{1}\right)p_{j}\left(z_{2}\right)$
for $i,j=1,...,J$ yielding a total of $K=J^{2}$ bivariate basis
functions. In particularly, we do not exploit the additive structure
of the problem since we are interested in the practical contents of
Theorems \ref{thm: proj rate} where no particular sparsity/special
structure of the model is assumed to be known. 

However, in practive, Chebyshev polynomials very easily pick up the
additive structure and effectively sets the coefficients of the cross-product
terms to zero. This is illustrated in Table \ref{tab:Coefficients-on-tensor}
in Appendix \ref{sec:Sieve-spaces}, where we report the coefficients
for one particular projection-based bivariate value function estimate
using a tensor product of $J=5$ Chebyshev polynomials. However, this
is due to the particular properties of the Chebyshev polynomials and
is not enforced by us in the implementation. For example, if we instead
use B-splines, the ``estimated'' coefficients of the cross-product
terms were significantly different from zero, c.f. Table \ref{tab:Coefficients-on-tensor-bspline}
in Appendix \ref{sec:Sieve-spaces}.

In the left-hand side panel (a) of Figure \ref{fig:norm_high_d},
we report the uniform bias of the projection-based method with $N$
chosen very large for the additive bivariate model. We find that the
bias vanishes as $K$ increases as in the one-dimensional model. However,
convergence is now slower in $K$ relative to the univariate case
and we require $K=50$ to obtain a sieve approximation bias of $10^{-2}$
while $K=7$ sufficed in the univariate case. This is consistent with
theoretical error rates for polynomial interpolation where the rate
slows down as the dimension of the problem increases, c.f. Section
\ref{subsec:Numerical-implementation}. 

\begin{figure}
\caption{Bias of value function in bivariate DDP for varying $K$ .\label{fig:norm_high_d}}

\noindent \centering{}\begin{tikzpicture}[]
\begin{axis}[height = {60.0mm}, ylabel = {$||E(V)-V_0||$}, title = {(a) Basic Model}, xmin = {1.0}, xmax = {225.0}, ymax = {25.96605423454403}, ymode = {log}, xlabel = {$K = J^{d_z}$}, {unbounded coords=jump, scaled x ticks = false, xticklabel style={rotate = 0}, xmajorgrids = false,
    xtick = {50.0,100.0,150.0,200.0}, xticklabels = {50,100,150,200},
    xtick align = inside, scaled y ticks = false, yticklabel style={rotate = 0}, log basis y=10, ymajorgrids = false, ytick = {0.0001,0.001,0.01,0.1,1.0,10.0}, yticklabels = {$10^{-4}$,$10^{-3}$,$10^{-2}$,$10^{-1}$,$10^{0}$,$10^{1}$}, ytick align = inside,     xshift = 0.0mm,
    yshift = 0.0mm,
    axis background/.style={fill={rgb,1:red,1.00000000;green,1.00000000;blue,1.00000000}}
}, ymin = {6.220869254036643e-5}, width = {70.0mm}]\addplot+ [color = black,
draw opacity=1.0,
line width=1,
solid,mark = none,
mark size = 2.0,
mark options = {
    color = black, draw opacity = 1.0,
    line width = 1,
    rotate = 0,
    solid
}]coordinates {
(1.0, 25.96605423454403)
(4.0, 15.473414532393248)
(9.0, 0.8262958253584181)
(16.0, 0.33221784041361246)
(25.0, 0.1198437624284665)
(36.0, 0.02422665488304787)
(49.0, 0.019612055432070008)
(64.0, 0.004234806634592303)
(81.0, 0.002999501508448077)
(100.0, 0.0007378123762293853)
(121.0, 0.0003459005954056238)
(144.0, 0.00023370275029321874)
(169.0, 0.00013909289346258902)
(196.0, 0.00010918456974451374)
(225.0, 6.220869254036643e-5)
};
\end{axis}

\end{tikzpicture}\begin{tikzpicture}[]
\begin{axis}[height = {60.0mm}, ylabel = {$||E(V)-V_0||$}, title = {(b) Model with interactions}, xmin = {1.0},
    xmax = {225.0},
    ymax = {60.96605423454403}, ymode = {log}, xlabel = {$K = J^{d_z}$}, {unbounded coords=jump, scaled x ticks = false, xticklabel style={rotate = 0}, xmajorgrids = false,
    xtick = {50.0,100.0,150.0,200.0}, xticklabels = {50,100,150,200},
    xtick align = inside, scaled y ticks = false, yticklabel style={rotate = 0}, log basis y=10, ymajorgrids = false,
    ytick = {0.0001,0.001,0.01,0.1,1.0,10.0}, yticklabels = {$10^{-4}$,$10^{-3}$,$10^{-2}$,$10^{-1}$,$10^{0}$,$10^{1}$},
    ytick align = inside,     xshift = 0.0mm,
    yshift = 0.0mm,
    axis background/.style={fill={rgb,1:red,1.00000000;green,1.00000000;blue,1.00000000}}
}, ymin = {6.220869254036643e-5}, width = {70.0mm}]\addplot+ [color = black,
draw opacity=1.0,
line width=1,
solid,mark = none,
mark size = 2.0,
mark options = {
    color = black, draw opacity = 1.0,
    line width = 1,
    rotate = 0,
    solid
}]coordinates {
(1,59.637078792423615)
(4,42.61202717231602)
(9,24.968764395765728)
(16,7.84800786862214)
(25,1.8504965142465153)
(36,0.47665259524937653)
(49,1.0264764436439364)
(64,0.2733687279546828)
(81,0.10208350909686459)
(100,0.12628316244317972)
(144,0.024055191906768414)
(169,0.013355665029038732)
(196,0.02022608108572399)
(225,0.009718294217506696)
};
\end{axis}

\end{tikzpicture}
\end{figure}
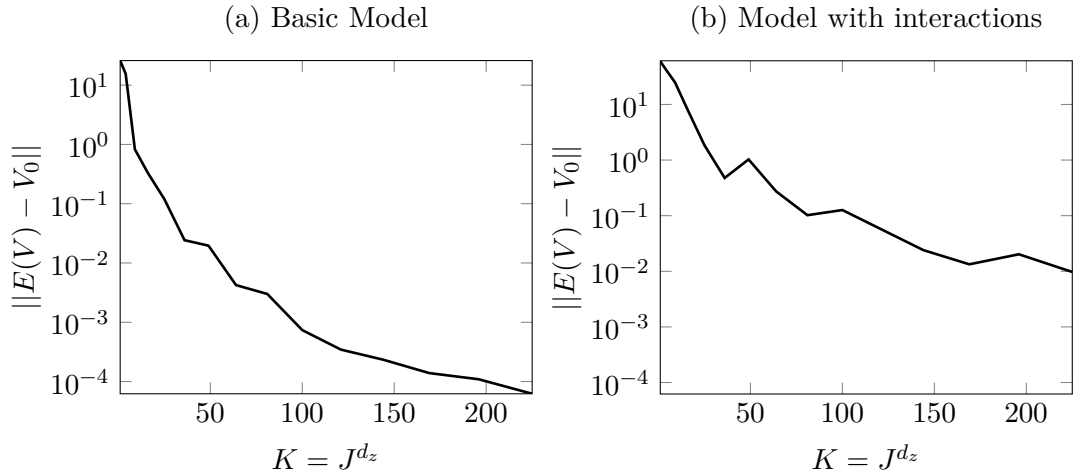

\subsubsection*{A non-additive DDP}

One concern with the numerical results reported for the bivariate
additive model in the previous section is that they may understate
the curse of dimensionality of the sieve method: The true value function
is by construction additive in the two state variables and so interaction
terms do not appear. This in turn implies that the computational complexity
of solving this particular model is relatively low; in particular,
the solution should be well-approximated by lower-dimensional sieves
($K$ small).

To investigate how the sieve method performs when applied to a more
complex, non-additive model, we here consider a slightly more complicated
bivariate model where we include a multiplicative interaction term
so that maintenance and replacement costs of the two busses interact,
$\bar{u}(z,d)=\sum_{i=1}^{2}u(z_{i},d_{i})-u(z_{1},d_{1})u(z_{2},d_{2})/20$.
Such a structure could, for example, reflect that capacity constraints
make it more costly to simultaneously replace the engines of both
busses. The resulting value funtion will have a more complicated multidimensional
structure and so we expect that the computational cost of our sieve
method should be higher in this scenario.

The sieve approximation bias of our solution method for this model
is reported in the right-hand side panel (b) in Figure \ref{fig:norm_high_d}.
Compared to panel (a) \textendash{} the additive case \textendash{}
we see that more sieve terms are required in order to reach a specific
absolute error level in the model with interactions. In Table \ref{interactiontable}
in Appendix \ref{sec:Sieve-spaces} the coefficients on the first
ten basis functions in each dimension and their interactions are reported.
Compared to the Chebyshev-based solution earlier we see quite significant
coefficients on the coefficients for the cross-terms. However, the
coefficients on the basis functions tend to zero quite quickly as
K increase. The sup-norm of the difference in the value function at
$40.000$ evaluation grids is on the order of $10^{-5}$ when comparing
the solutions with $K=50^{2}=2500$ and $K=30^{2}=900$ basis functions,
and individual coefficients fall below $10^{-6}$ for univariate basis
functions and cross products beyond the 22nd univariate basis functions,
and below $10^{-8}$ around the 30th basis functions.However, it is
important to stress that a large $K$ here only comes with a computational
cost while a large $K$ has little effect on the variance of the sieve
method. All together, we find that the sieve-based solution method
works well also in higher dimensions, in particular when the model
has a particular structure that can be utilized in the solution method.

\section{Conclusion}

We have proposed two novel methods for numerical computation of either
the so-called integrated or expected value functions in a general
class of dynamic discrete choice models.. Both methods rely on a smoothed
simulated version of the Bellman operators definining the integrated
and expected values functions. The smoothing facilitates both the
practical implementation and the theoretical analysis of the approximate
value functions. Under regularity conditions, we develop an asymptotic
theory for the two methods as the number of simulations used to compute
the simulated Bellman operators diverge. A set of numerical experiments
show that our first method, the so-called self-approximating method
can be somewhat unstable while the second one, which relies on sieve
methods, apppears much more numerically robust. The next step is to
develop methods for choosing the number of simulations, sieve basis
functions and smoothing parameter $\lambda$ in a given setting so
that the resulting approximate solution is of a good quality. Another
area of research is to investigate how the proposed solution methods
can be used for the estimation of dynamic discrete choice models.

\newpage{}

\bibliographystyle{chicago}
\bibliography{DDC-sieve}

\newpage{}

\appendix
\counterwithin{thm}{section}

\appendix

\section{\label{app:Auxiliary-Results}Auxiliary Results}

We derive two general results for approximate solutions to functional
fixed-points. Let $\left(\mathcal{X},\left\Vert \cdot\right\Vert \right)$
be a normed vector space and $\Psi:\mathcal{X}\rightarrow\mathcal{\mathcal{X}}$
be some contraction mapping w.r.t. $\left\Vert \cdot\right\Vert $
so that there exists a unique solution $x_{0}\in\mathcal{\mathcal{X}}$
to $x=\Psi\left(x\right)$. Let $\Psi_{N}$ be an approximation to
$\Psi$ and let $\Pi_{K}$ be a projection operator, $K,N\geq1$.
\begin{thm}
\label{thm: General rate}Suppose (i) $\left\Vert \Psi_{N}\left(x_{0}\right)-\Psi\left(x_{0}\right)\right\Vert =O_{p}(\rho_{\Psi,N})$
for some $\rho_{\Psi,N}\rightarrow0$ and (ii) for some $\beta<1$,
$\left\Vert \Psi_{N}\left(x\right)-\Psi_{N}\left(y\right)\right\Vert \leq\beta\Vert x-y\Vert$
for all $N$ large enough and all $x,y$. Then there exists a unique
solution $x_{N}\in\mathcal{\mathcal{X}}$ to $x=\Psi_{N}\left(x\right)$
with probability approaching one (w.p.a.1) satisfying $\Vert x_{N}-x_{0}\Vert=O_{P}\left(\rho_{\Psi,N}\right).$

Suppose furthermore (iii) $\Pi_{K}:\mathcal{X}\rightarrow\mathcal{\mathcal{X}}$
satisfies $\Vert\Pi_{K}\left(x_{N}\right)-x_{N}\Vert=O_{p}(\rho_{\Pi,K})$
for some $\rho_{\Pi,K}\rightarrow0$. Then there exists a unique solution
$\hat{x}_{N}\in\mathcal{\mathcal{X}}$ to $x=\left(\Pi_{K}\Psi_{N}\right)\left(x\right)$
w.p.a.1 satisfying
\[
\left\Vert \hat{x}_{N}-x_{0}\right\Vert \leq\left\Vert \hat{x}_{N}-x_{N}\right\Vert +\left\Vert x_{N}-x_{0}\right\Vert =O_{p}(\rho_{\Pi,K})+O_{p}(\rho{}_{\Psi,N}).
\]
\end{thm}
\begin{proof}
We first observe that due to (ii), there exists a unique solution
$x_{N}=\Psi_{N}\left(x_{N}\right)$ for all $N$ large enough which
satisfies
\begin{align*}
\Vert x_{N}-x_{0}\Vert & =\Vert\Psi_{N}(x_{N})-\Psi(x_{0})\Vert\leq\Vert\Psi_{N}(x_{N})-\Psi_{N}(x_{0})\Vert+\Vert\Psi_{N}(x_{0})-\Psi(x_{0})\Vert\\
 & \leq\beta\Vert x_{N}-x_{0}\Vert+\Vert\Psi_{N}(x_{0})-\Psi(x_{0})\Vert,
\end{align*}
and so $\Vert x_{N}-x_{0}\Vert\leq\Vert\Psi_{N}(x_{0})-\Psi(x_{0})\Vert/\left(1-\beta\right)=O_{P}\left(\rho_{\Psi,N}\right).$
Next, combining (ii) and (iii), we see that $\Pi_{K}\Psi_{N}$ is
a contraction mapping w.p.a.1. with Lipschitz coeffient $\beta$,
and so $\hat{x}_{N}$ defined in the theorem exists and is unique
w.p.a.1. Moreover, by the same arguments employed in the analysis
of $x_{N}$, 
\[
\Vert\hat{x}_{N}-x_{N}\Vert\leq\frac{\Vert\Pi_{K}\Psi_{N}(x_{N})-\Psi_{N}(x_{N})\Vert}{1-\beta}=\frac{\Vert\Pi_{K}\left(x_{N}\right)-x_{N}\Vert}{1-\beta}=O_{p}(\rho_{\Pi,K}).
\]
\end{proof}
\begin{thm}
\label{thm: general dist}Suppose the following conditions are satisfied:
(i) $\Vert x_{N}-x_{0}\Vert=O_{P}\left(\rho_{\Psi,N}\right)$; (ii)
$\rho_{\Psi,N}^{-1}\left\{ \Psi_{N}(x_{0})-\Psi(x_{0})\right\} \rightsquigarrow\mathbb{G}$
in $\left(\mathcal{X},\left\Vert \cdot\right\Vert \right)$; (iii)
$\Psi_{N}(x_{0})$ is Frechet differentiable at $x_{0}$ w.p.a.1 with
Frechet differential $\nabla\Psi_{N}(x_{0})\left[\cdot\right]:\partial\mathcal{X}\mapsto\mathcal{X}$
for some function set $\partial\mathcal{X}$, where $x_{N}-x_{0}\in\partial\mathcal{X}$
w.p.a.1, such that $\left\Vert \Psi_{N}(x_{N})-\Psi_{N}(x_{0})-\nabla\Psi_{N}(x_{0})\left[x_{N}-x_{0}\right]\right\Vert =o_{P}\left(\left\Vert x_{N}-x_{0}\right\Vert \right)$
and (iv) $\sup_{dx\in\partial\mathcal{X}:\left\Vert dx\right\Vert =1}\left\Vert \left\{ \nabla\Psi_{N}(x_{0})-\nabla\Psi(x_{0})\right\} \left[dx\right]\right\Vert =o_{p}\left(1\right)$.
Then $\left\{ I-\nabla\Psi(x_{0})\right\} \left[\rho_{\Psi,N}\left\{ x_{N}-x_{0}\right\} \right]\rightsquigarrow\mathbb{G}.$
If furthermore (v) $I-\nabla\Psi(x_{0})\left[\cdot\right]:\mathcal{\partial\mathcal{X}}\mapsto\mathcal{X}$
has a continuous inverse, then $\rho_{\Psi,N}\left\{ x_{N}-x_{0}\right\} \rightsquigarrow\left\{ I-\nabla\Psi(x_{0})\right\} ^{-1}\left[\mathbb{G}\right]$.
\end{thm}
\begin{proof}
To show the first claim, combine a functional Taylor expansion with
conditions (i) and (iii),
\begin{eqnarray*}
0 & = & \left(I-\Psi_{N}\right)(x_{N})=\left(I-\Psi_{N}\right)(x_{0})+\left\{ I-\nabla\Psi_{N}(x_{0})\right\} \left[x_{N}-x_{0}\right]+o_{P}\left(\left\Vert x_{N}-x_{0}\right\Vert \right)\\
 & = & \Psi(x_{0})-\Psi_{N}(x_{0})+\left\{ I-\nabla\Psi_{N}(x_{0})\right\} \left[x_{N}-x_{0}\right]+o_{P}\left(\rho_{\Psi,N}\right).
\end{eqnarray*}
Next, by (iv),
\begin{eqnarray*}
 &  & \left\Vert \left\{ I-\nabla\Psi_{N}(x_{0})\right\} \left[x_{N}-x_{0}\right]-\left\{ I-\nabla\Psi(x_{0})\right\} \left[x_{N}-x_{0}\right]\right\Vert \\
 & = & \left\Vert \left\{ \nabla\Psi_{N}(x_{0})-\nabla\Psi(x_{0})\right\} \left[x_{N}-x_{0}\right]\right\Vert \\
 & \leq & \sup_{dx\in\partial\mathcal{X}:\left\Vert dx\right\Vert =1}\left\Vert \left\{ \nabla\Psi_{N}(x_{0})-\nabla\Psi(x_{0})\right\} \left[dx\right]\right\Vert \left\Vert x_{N}-x_{0}\right\Vert \\
 & = & o_{P}\left(\rho_{\Psi,N}\right).
\end{eqnarray*}
Combining this with (ii),
\begin{eqnarray*}
\left\{ I-\nabla\Psi(x_{0})\right\} \left[\rho_{\Psi,N}^{-1}\left\{ x_{N}-x_{0}\right\} \right] & =\rho_{\Psi,N}^{-1} & \left\{ \Psi_{N}(x_{0})-\Psi(x_{0})\right\} +o_{P}\left(1\right)\rightsquigarrow\mathbb{G},
\end{eqnarray*}
The second claim follows by (v) and the continuous mapping theorem.
\end{proof}
It is important here to note the tension between the requirement that
that $x_{N}-x_{0}\in\partial\mathcal{X}$ in (iii) and$\sup_{dx\in\partial\mathcal{X}:\left\Vert dx\right\Vert =1}\left\Vert \left\{ \nabla\Psi_{N}(x_{0})-\nabla\Psi(x_{0})\right\} \left[dx\right]\right\Vert =o_{p}\left(1\right)$.
The first condition will hold if we choose $\partial\mathcal{X}$
large enough. But at the same time, we need to show uniform convergence
over the same space which will generally only hold if $\partial\mathcal{X}$
is Glivenko-Cantelli. In the application to value function approximation,
this is achieved by choosing $\partial\mathcal{X}=\boldsymbol{\mathbb{C}}_{r}^{1}\left(\mathcal{Z}\right)$
defined in Section \ref{subsec:Function-approximation} for some $r<\infty$.

\section{\label{App:Proofs}Proofs}
\begin{proof}[Proof of Theorem \ref{Thm: contraction}]
First note that for any $V\left(z;\lambda\right)\in\boldsymbol{\mathbb{B}}\left(\mathcal{Z}\times\left(0,\bar{\lambda}\right)\right)^{D}$
and with $C$ denoting a generic constant,
\begin{eqnarray}
\left|G_{\lambda}\left(u_{\psi}\left(U,z\right)+\beta V_{\psi}\left(U;z,\lambda\right)\right)\right| & \leq & C\left(1+\left\Vert u_{\psi}\left(U,z\right)\right\Vert +\beta\left\Vert V_{\psi}\left(U;z,\lambda\right)\right\Vert \right)\label{eq: ineq}\\
 & \leq & C\left(1+\bar{u}_{\psi}\left(U\right)+\beta\left\Vert V\right\Vert _{\infty}\right),\nonumber 
\end{eqnarray}
and so, using Assumption \ref{assu: mom bound},
\begin{eqnarray*}
\left\Vert \Gamma(V)\right\Vert _{\infty} & \leq & \sup_{\left(z,\lambda\right)\in\mathcal{Z}\times\left(0,\bar{\lambda}\right)}E\left[\left|G_{\lambda}\left(u_{\psi}\left(U;z\right)+\beta V_{\psi}(U;z,\lambda)\right)\right|w_{\psi}\left(U;z\right)\right]\\
 & \leq & CE\left[\left|\left(1+\bar{u}_{\psi}\left(U\right)+\beta\left\Vert V\right\Vert _{\infty}\right)\right|\bar{w}_{\psi}\left(U\right)\right]<\infty,
\end{eqnarray*}
which shows that $\Gamma:\boldsymbol{\mathbb{B}}\left(\mathcal{Z}\times\left(0,\bar{\lambda}\right)\right)^{D}\mapsto\boldsymbol{\mathbb{B}}\left(\mathcal{Z}\times\left(0,\bar{\lambda}\right)\right)^{D}$.
Recycling eq. (\ref{eq: ineq}),
\begin{eqnarray*}
\left\Vert \Gamma_{N}(V)\right\Vert _{\infty} & \leq & \frac{\sum_{i=1}^{N}\left\Vert G_{\lambda}\left(u_{\psi}\left(U_{i};z\right)+\beta V_{\psi}\left(U_{i};z,\lambda\right)\right)\right\Vert w_{\psi}\left(U_{i};z\right)}{\sum_{i=1}^{N}w_{\psi}\left(U_{i};z\right)}\\
 & \leq & C\left(\frac{\sum_{i=1}^{N}\left\Vert \bar{u}_{\psi}\left(U\right)\right\Vert w_{\psi}\left(U_{i};z\right)}{\sum_{i=1}^{N}w_{\psi}\left(U_{i};z\right)}+1+\beta\left\Vert V\right\Vert _{\infty}\right)\\
 & \leq & C\left(\frac{\sum_{i=1}^{N}\left\Vert \bar{u}_{\psi}\left(U_{i}\right)\right\Vert \bar{w}_{\psi}\left(U_{i}\right)}{\inf_{z\in\mathcal{Z}}\sum_{i=1}^{N}w_{\psi}\left(U_{i};z\right)}+1+\beta\left\Vert V\right\Vert _{\infty}\right)<\infty.
\end{eqnarray*}
Thus, for any given $N\geq1$, $\Gamma_{N}:\boldsymbol{\mathbb{B}}\left(\mathcal{Z}\times\left(0,\bar{\lambda}\right)\right)^{D}\mapsto\boldsymbol{\mathbb{B}}\left(\mathcal{Z}\times\left(0,\bar{\lambda}\right)\right)^{D}.$
To show that $\Gamma_{N}:\boldsymbol{\mathbb{B}}\left(\mathcal{Z}\times\left(0,\bar{\lambda}\right)\right)^{D}\mapsto\boldsymbol{\mathbb{B}}\left(\mathcal{Z}\times\left(0,\bar{\lambda}\right)\right)^{D}$
is a contraction, use that, by quasi-linearity of $G_{\lambda}\left(r\right)$,
for any $V_{1},V_{2}\in\boldsymbol{\mathbb{B}}\left(\mathcal{Z}\times\left(0,\bar{\lambda}\right)\right)^{D}$,
\begin{eqnarray*}
\Gamma_{N}(V_{1})(z,\lambda,d) & = & \sum_{i=1}^{N}G_{\lambda}\left(u_{\psi}\left(U_{i};z\right)+\beta V_{\psi,2}\left(U_{i};z,\lambda\right)+\beta\left[V_{\psi,1}\left(U_{i};z,\lambda\right)-V_{\psi,2}\left(U_{i};z\right)\right]\right)w_{N,i}\left(z,d\right)\\
 & \leq & \sum_{i=1}^{N}G_{\lambda}\left(u_{\psi}\left(U_{i};z\right)+\beta V_{\psi,2}\left(U_{i};z,\lambda\right)+\beta\left\Vert V_{1}-V_{2}\right\Vert {}_{\infty}\mathbf{1}_{D}\right)w_{N,i}\left(z,d\right)\\
 & = & \sum_{i=1}^{N}G_{\lambda}\left(u_{\psi}\left(U_{i};z\right)+\beta V_{\psi,2}\left(U_{i};z,\lambda\right)\right)w_{N,i}\left(z,d\right)+\beta\left\Vert V_{1}-V_{2}\right\Vert {}_{\infty}\sum_{i=1}^{N}w_{N,i}\left(z,d\right)\\
 & = & \Gamma_{N}(V_{2})(z,\lambda,d)+\beta\left\Vert V_{1}-V_{2}\right\Vert {}_{\infty},
\end{eqnarray*}
where $\mathbf{1}_{d}=\left(1,....,1\right)\in\mathbb{R}^{D}$ and
we have used that $\sum_{i=1}^{N}w_{N,i}\left(z,d\right)=1$ by construction.
The proof of $\Gamma$ being a contraction is analogous.

Next, we prove that $V_{N}\left(z,\lambda\right)$ is $s\geq1$ times
continuously differentiable under Assumption \ref{Ass: smoothness}:
We know that $\Gamma_{N}$ is a contraction mapping on $\mathbb{B}\left(\mathcal{Z}\times\left(0,\bar{\lambda}\right)\right)^{D}$.
But the set of $s\geq0$ continuously differentiable functions $\mathbb{C}^{s}\left(\mathcal{Z}\times\left(0,\bar{\lambda}\right)\right)^{D}$
is a closed subset of $\mathbb{B}\left(\mathcal{Z}\times\left(0,\bar{\lambda}\right)\right)^{D}$
and so the result will follow if $\Gamma_{N,\lambda}\left(\mathbb{C}_{s}\left(\mathcal{Z}\times\left(0,\bar{\lambda}\right)\right)^{D}\right)\subseteq\mathbb{C}^{s}\left(\mathcal{Z}\times\left(0,\bar{\lambda}\right)\right)^{D}$.
But for any $V\in\mathbb{C}^{s}\left(\mathcal{Z}\times\left(0,\bar{\lambda}\right)\right)^{D}$,
it follows straightforwardly by the chain rule in conjuction with
the stated assumptions that $\Gamma_{N}(V)(z,\lambda)=\sum_{i=1}^{N}G_{\lambda}\left(u_{\psi}\left(U_{i};z\right)+\beta V_{\psi}\left(U_{i};z,\lambda\right)\right)w_{N,i}\left(z\right)$
is $s\geq0$ continuously differentiable w.r.t. $\left(z,\lambda\right)$.
The proof of the Lipschitz property under Assumption \ref{Ass: smoothness}(i)
is similar and so left out.
\end{proof}
\begin{proof}[Proof of Theorem \ref{thm: smoothing}]
We only show the result for $V_{0}$; the proof for the $V_{N}$
is analogous. Applying (\ref{eq: G approx error}), the following
holds for any $V$,
\begin{align*}
\left|\Gamma(V)(z,0,d)-\Gamma(V)(z,\lambda,d)\right|\leq & \int\left|\max_{d\in\mathcal{D}}\left\{ u(s^{\prime},d)+\beta V(z^{\prime},d^{\prime})\right\} -G_{\lambda}\left(u(s^{\prime})+\beta V(z^{\prime})\right)\right|dF_{S}(s^{\prime}|z,d)\\
\leq & \sup_{r\in\mathbb{R}^{D}}\left|G_{\lambda}\left(r\right)-\max_{d\in\mathcal{D}}r\left(d\right)\right|\int_{\mathcal{Z}\times\mathcal{E}}dF_{s}(ds^{\prime}|z,d)\\
\leq & \lambda\log D.
\end{align*}
The result now follows from the first part of Theorem \ref{thm: General rate}
with $\Psi_{N}\left(\cdot\right)=\Gamma\left(\cdot\right)\left(\cdot,\lambda_{N}\right)$. 
\end{proof}
Our asymptotic analysis of $V_{N}$ proceeds in two steps: First,
we develop a master theorem that delivers the desired result under
a set of high-level conditions on the model and chosen importance
sampler. The conditions are formulated to cover a wide range of different
specifications, including both the case of $Z_{t}$ being continuously
distributed or having countable support. Also, the master theorem
allows for a wide range of the per-period utility functions and importance
samplers. To state the high-level conditions, we recall the following
definitions (see \citealp{VW}): A class $\mathcal{F}$ of measurable
functions mapping $U$ into $\mathbb{R}$ is called $P$$_{U}$-Glivenko-Cantelli
if $\sup_{f\in\mathcal{F}}\left|\frac{1}{N}\sum_{i=1}^{N}f\left(U_{i}\right)-E\left[f\left(U\right)\right]\right|\rightarrow^{P}0$
and it is called $P$$_{U}$-Donsker if $\sup_{f\in\mathcal{F}}\frac{1}{\sqrt{N}}\sum_{i=1}^{N}\left\{ f\left(U_{i}\right)-E\left[f\left(U\right)\right]\right\} \rightsquigarrow\mathbb{G}$
in the space of all bounded functions from $\mathcal{F}$ to $\mathbb{R}$,
where $\mathbb{G}$ is a tight Gaussian process.
\begin{thm}
\label{thm: Master} (i) Suppose that Assumption \ref{assu: mom bound}
is satisfied and the function classes $\mathcal{W}:=\left\{ \left.U\mapsto w_{\psi}\left(U;z\right)\right|z\in\mathcal{Z}\right\} $
and
\begin{equation}
\mathcal{\mathcal{G}}=\left\{ \left.U\mapsto G_{\lambda}\left(u_{\psi}\left(U;z\right)+\beta V_{0}\left(\psi_{Z}\left(U;z\right),\lambda\right)\right)w_{\psi}\left(U;z\right)\right|\left(z,\lambda\right)\in\mathcal{Z}\times\left(0,\bar{\lambda}\right)\right\} \label{eq: G set def}
\end{equation}
are $P$$_{U}$-Donsker. Then the first part of Theorem \ref{Thm: Rust approx rate}
holds.

(ii) Suppose furthermore that $V_{N}-V_{0}\in\partial\mathcal{V}$
where $\partial\mathcal{V}$ is $P_{U}$-Glivenko-Cantelli with an
integrable envelope function and
\[
\mathcal{G}^{\prime}=\left\{ \left.U\mapsto\sum_{d\in\mathcal{D}}\dot{G}_{d,\lambda}\left(u_{\psi}\left(U;z\right)+\beta V_{\psi,0}\left(U;z,\lambda\right)\right)w_{\psi}\left(U;z\right)\right|\left(z,\lambda\right)\in\mathcal{Z}\times\left(0,\bar{\lambda}\right)\right\} 
\]
is $P_{U}$-Glivenko-Cantelli. Then the conclusions of Theorem \ref{Thm: Rust normal}
also hold.
\end{thm}
\begin{proof}
To show (i), we apply the first part of Theorem \ref{thm: General rate}
with $\Psi_{N}=\Gamma_{N}$, which is a contraction w.p.a.1, c.f.
Theorem \ref{Thm: contraction}. First write
\begin{equation}
\Gamma_{N}(V)(z,\lambda)=\frac{\tilde{\Gamma}_{N}(V)(z,\lambda)}{W_{N}\left(z\right)}.\label{eq: Gamma_N ratio}
\end{equation}
where
\begin{eqnarray}
\tilde{\Gamma}_{N}(V)(z,\lambda) & = & \frac{1}{N}\sum_{i=1}^{N}G_{\lambda}\left(u_{\psi}\left(U_{i};z\right)+\beta V_{\psi}(U_{i};z)\right)w_{\psi}\left(U_{i};z\right),\label{eq: Gamma-tilde def}\\
W_{N}\left(z\right) & = & \frac{1}{N}\sum_{j=1}^{N}w_{\psi}\left(U_{i};z\right),\label{eq: W def}
\end{eqnarray}
The Donsker condition on $\mathcal{G}$ and $\mathcal{W}$ now implies
that 
\begin{equation}
\sqrt{N}\left(\tilde{\Gamma}_{N}(V_{0})-\Gamma(V_{0}),W_{N}-1\right)\rightsquigarrow\left(\mathbb{G}_{1},\mathbb{G}_{2}\right)\label{eq: Gamma_N conv}
\end{equation}
on $\mathcal{B}\left(\mathcal{Z}\times\left(0,\bar{\lambda}\right)\right)$,
where $\left(\mathbb{G}_{1},\mathbb{G}_{2}\right)$ is a Gaussian
process, and so
\begin{eqnarray}
\sqrt{N}\left\{ \Gamma_{N}(V_{0})-\Gamma(V_{0})\right\}  & = & \sqrt{N}\left\{ \tilde{\Gamma}_{N}(V_{0})-\Gamma(V_{0})\right\} -\Gamma(V_{0})\sqrt{N}\left\{ W_{N}-1\right\} +o_{P}\left(1\right)\nonumber \\
 & \rightsquigarrow & \mathbb{G}:=\mathbb{G}_{1}-\Gamma(V_{0})\mathbb{G}_{2}.\label{eq: Gamma_N conv 2}
\end{eqnarray}
In particular, $\left\Vert \Gamma_{N}(V_{0})-\Gamma(V_{0})\right\Vert _{\infty}=O_{P}\left(1/\sqrt{N}\right)$.
We conclude from Theorem \ref{thm: General rate} that $\Vert V_{N}-V_{0}\Vert_{\infty}=O_{p}(1/\sqrt{N})$.

To show the second part, we apply Theorem \ref{thm: general dist}.
Weak convergence was derived above and it is easily seen that the
influence function of $\Gamma_{N}(V_{0})$ takes the form given in
eq. (\ref{eq: g def}) and so the Gaussian process $\mathbb{G}\left(z,\lambda\right)$
in eq. (\ref{eq: Gamma_N conv 2}) has covariance kernel given in
(\ref{eq: Omega def}). The Frechet differential $dV\mapsto\nabla\Gamma_{N}(V_{N})\left[dV\right]$
was derived in (\ref{eq: dGamma_N def}). It is a linear operator
with $\left\Vert \nabla\Gamma_{N}(V_{N})\left[dV\right]\right\Vert \leq\beta\left\Vert dV\right\Vert $
and so $dV\mapsto\left\{ I-\nabla\Gamma_{N}(V_{N})\right\} \left[dV\right]$
has a well-defined continuous inverse. Thus, what remains is to verify
(iv) of Theorem \ref{thm: general dist}. This is done by showing
uniform convergence of $dV\mapsto\nabla\tilde{\Gamma}_{N}(V_{0})\left[dV\right]$
over $\mathcal{B}\left(\mathcal{Z}\times\left(0,\bar{\lambda}\right)\times\partial\mathcal{V}\right)$.
But
\begin{eqnarray}
\nabla\tilde{\Gamma}_{N}(V_{0})\left[dV\right](z) & = & \frac{\beta}{N}\sum_{i=1}^{N}\sum_{d\in\mathcal{D}}\dot{G}_{d,\lambda}\left(u_{\psi}\left(U_{i};z\right)+\beta V_{\psi,0}\left(U_{i};z,\lambda\right)\right)dV_{\psi}\left(U_{i};z,\lambda,d\right)w_{\psi}\left(U_{i};z\right)\label{eq: dGamma_N def-1}
\end{eqnarray}
where $V_{\psi}\left(U;z,\lambda,d\right)\in\partial\mathcal{V}_{\psi}$
with
\[
\partial\mathcal{V}_{\psi}=\left\{ \left.U\mapsto dV\left(\psi_{Z}\left(U;z\right),\lambda\right)\right|\left(z,\lambda,dV\right)\in\mathcal{Z}\times\left(0,\bar{\lambda}\right)\times\partial\mathcal{V}\right\} 
\]
which is Glivenko-Cantelli since $\partial\mathcal{V}$ and $\left\{ \left.U\mapsto\psi_{Z}\left(U;z\right)\right|z\in\mathcal{Z}\right\} $
both have this property. Since $\mathcal{G}^{\prime}$ is also Glivenko-Cantelli,
it now follows from Theorem 3 in \citet{Vaart&Wellner2000} that $\mathcal{G}^{\prime}\cdot\partial\mathcal{V}_{\psi}$
is Glivenko-Cantelli as well which yields the desired result.
\end{proof}
\begin{proof}[Proof of Theorem \ref{Thm: Rust approx rate}]
 To show $\Vert V_{N}-V_{0}\Vert_{\infty}=O_{P}\left(1/\sqrt{N}\right)$,
we verify the conditions of part (i) in Theorem \ref{thm: Master}.
First observe that $V_{0}(z,\lambda)$ is Lipschitz in $\left(z,\lambda\right)$,
c.f. Theorem \ref{Thm: contraction}, and that $r\mapsto G_{\lambda}\left(r\right)$
is also Lipschitz uniformly in $\lambda\in(0,\bar{\lambda})$. Next,
we show that $G_{\lambda}\left(r\right)$ is also Lispchitz w.r.t.
$\lambda$ uniformly in $r$ by verifying that $\partial G_{\lambda}\left(r\right)/\left(\partial\lambda\right)$
is bounded uniformly in $\lambda\in(0,\bar{\lambda})$: Write
\[
G_{\lambda}\left(r\right)=\lambda\log\left[\sum_{d\in\mathcal{D}}\exp\left(\frac{r\left(d\right)}{\lambda}\right)\right]=\max_{d\in D}r\left(d\right)+\lambda\log\left[\sum_{d\in\mathcal{D}}\exp\left(\frac{\bar{r}\left(d\right)}{\lambda}\right)\right],
\]
where $\bar{r}\left(d\right)=r\left(d\right)-\max_{d\in\mathcal{D}}r\text{\ensuremath{\left(d\right)}}\leq0$,
$d\in\mathcal{D}$, to obtain
\begin{eqnarray}
\dot{G}_{\lambda}^{\left(\lambda\right)}(r)=\frac{\partial G_{\lambda}\left(r\right)}{\partial\lambda} & = & \log\left[\sum_{d\in\mathcal{D}}\exp\left(\frac{\bar{r}\left(d\right)}{\lambda}\right)\right]-\frac{\sum_{d\in\mathcal{D}}\exp\left(\frac{\bar{r}\left(d\right)}{\lambda}\right)\frac{\bar{r}\left(d\right)}{\lambda}}{\sum_{d\in\mathcal{D}}\exp\left(\frac{\bar{r}\left(d\right)}{\lambda}\right)}.\label{eq: G lambda-deriv}
\end{eqnarray}
Since $1\leq\sum_{d\in\mathcal{D}}\exp\left(\frac{\bar{r}\left(d\right)}{\lambda}\right)\leq D$
and $-De^{-1}\leq\sum_{d\in\mathcal{D}}\exp\left(\frac{\bar{r}\left(d\right)}{\lambda}\right)\frac{\bar{r}\left(d\right)}{\lambda}\leq0$
for all $\lambda>0$ and all $r\in\mathbb{R}^{D}$, we conclude that
$\left|\dot{G}_{\lambda}^{\left(\lambda\right)}(r)\right|\leq\log\left(D\right)+De^{-1}$
and so . Next,
\begin{eqnarray*}
 &  & \left|G_{\lambda}\left(u_{\psi}(U;z)+\beta V_{0}\left(\psi_{Z}\left(U,z\right),\lambda\right)\right)-G_{\lambda^{\prime}}\left(u_{\psi}(U;z^{\prime})+\beta V_{0}\left(\psi_{Z}\left(U,z^{\prime}\right),\lambda^{\prime}\right)\right)\right|\\
\leq &  & \left|G_{\lambda}\left(u_{\psi}(U;z)+\beta V_{0}\left(\psi_{Z}\left(U,z\right),\lambda\right)\right)-G_{\lambda^{\prime}}\left(u_{\psi}(U;z)+\beta V_{0}\left(\psi_{Z}\left(U,z\right),\lambda\right)\right)\right|\\
 & + & \left|G_{\lambda^{\prime}}\left(u_{\psi}(U;z)+\beta V_{0}\left(\psi_{Z}\left(U,z\right),\lambda\right)\right)-G_{\lambda^{\prime}}\left(u_{\psi}(U;z^{\prime})+\beta V_{0}\left(\psi_{Z}\left(U,z^{\prime}\right),\lambda^{\prime}\right)\right)\right|\\
\leq &  & C\left\{ \left|\lambda-\lambda^{\prime}\right|+\left\Vert u_{\psi}(U;z)-u_{\psi}(U;z^{\prime})\right\Vert +\left\Vert V_{0}\left(\psi_{Z}\left(U,z\right),\lambda\right)-V_{0}\left(\psi_{Z}\left(U,z^{\prime}\right),\lambda^{\prime}\right)\right\Vert \right\} ,
\end{eqnarray*}
and it now follows that under Assumption \ref{Ass: smoothness} together
with the Lipschitz property of $V_{0}$ that $\mathcal{G}$ as defined
in eq. (\ref{eq: G set def}) is Type IV class under $P_{U}$ with
index 2 according to the definition on p. 2278 in Andrews (1994) which
yields the first part of the theorem.

Next, we analyze $\partial V_{N}/\left(\partial z_{j}\right)$, $j=1,...,d_{Z}$.
Since $\sup_{z\in\mathcal{Z}}\left|\sum_{i=1}^{N}w\left(S_{i}\left(z,d\right)|z,d\right)/N-1\right|=O_{P}\left(1/\sqrt{N}\right)$,
we replace $w_{N,i}\left(z,d\right)$ by $w_{\psi}\left(U_{i};z,d\right)/N$
in the following. Now, taking derivatives w.r.t. $z_{j}$, $j=1,...,d_{Z}$,
on both sides of eq. (\ref{eq: Sim Fixed point}),
\begin{eqnarray}
\frac{\partial V_{N}\left(z,\lambda\right)}{\partial z_{j}} & = & \nabla\Gamma_{N}(V_{N})\left[\frac{\partial V_{N}}{\partial z_{j}}\right]\left(z,\lambda\right)+\Gamma_{N,j}^{\left(z\right)}(V_{N})\left(z,\lambda\right),\label{eq: V_N z deriv}
\end{eqnarray}
where $\nabla\Gamma_{N}$ was defined in ($\ref{eq: dGamma_N def}$)
and
\begin{eqnarray*}
\dot{\Gamma}_{N,j}^{\left(z\right)}(V_{N})\left(z,\lambda,d\right) & =\frac{1}{N} & \sum_{i=1}^{N}\sum_{d\in\mathcal{D}}\dot{G}_{\lambda,d}^{\left(r\right)}\left(u_{\psi}\left(U_{i};z,d\right)+\beta V_{\psi,N}\left(U_{i};z,\lambda\right)\right)\frac{\partial u_{\psi}\left(U_{i};z,d\right)}{\partial z_{j}}w_{\psi}\left(U_{i};z\right)\\
 &  & +\frac{1}{N}\sum_{i=1}^{N}G_{\lambda}\left(u_{\psi}\left(U_{i},d\right)+\beta V_{\psi,N}\left(U_{i};z,\lambda\right)\right)\frac{\partial w_{\psi}\left(U_{i};z\right)}{\partial z_{j}}.
\end{eqnarray*}
where $\dot{G}_{\lambda,d}^{\left(r\right)}(r)$ was defined in (\ref{eq: G r-deriv}).
Similarly,
\begin{eqnarray}
\frac{\partial V_{N}\left(z,\lambda\right)}{\partial\lambda} & = & \nabla\Gamma_{N}(V_{N})\left[\frac{\partial V_{N}}{\partial\lambda}\right]\left(z,\lambda\right)+\dot{\Gamma}_{N,j}^{\left(\lambda\right)}(V_{N})\left(z,\lambda\right),\label{eq: V_N z deriv-1}
\end{eqnarray}
where
\begin{eqnarray*}
\dot{\Gamma}_{N,j}^{\left(\lambda\right)}(V_{N})\left(z,\lambda,d\right) & =\frac{1}{N} & \sum_{i=1}^{N}\dot{G}_{\lambda}^{\left(\lambda\right)}\left(u_{\psi}\left(U_{i};z,d\right)+\beta V_{\psi,N}\left(U_{i};z,\lambda\right)\right)w_{\psi}\left(U_{i};z\right),
\end{eqnarray*}
and
\begin{equation}
\dot{G}_{\lambda}^{\left(\lambda\right)}\left(r\right)=\log\left[\sum_{d\in\mathcal{D}}\exp\left(\frac{r\left(d\right)}{\lambda}\right)\right]-\frac{\sum_{d\in\mathcal{D}}\exp\left(\frac{r\left(d\right)}{\lambda}\right)\frac{r\left(d\right)}{\lambda}}{\sum_{d\in\mathcal{D}}\exp\left(\frac{r\left(d\right)}{\lambda}\right)}\label{eq: G lambda-deriv-1}
\end{equation}
The mapping $dV\mapsto\nabla\Gamma_{N}(V_{N})\left[dV\right]$ is
a bounded linear operator with $\left\Vert \nabla\Gamma_{N}(V_{N})\left[dV\right]\right\Vert \leq\beta\left\Vert dV\right\Vert $
and so
\[
\frac{\partial V_{N}\left(z,\lambda\right)}{\partial z_{j}}=\left\{ I-\nabla\Gamma_{N}(V_{N})\right\} ^{-1}\left[\dot{\Gamma}_{N,j}^{\left(z\right)}(V_{N})\right]\left(z,\lambda\right).
\]
Thus,
\[
\left\Vert \frac{\partial V_{N}}{\partial z_{j}}\right\Vert _{\infty}=\left\Vert \left\{ I-\nabla\Gamma_{N}(V_{N})\right\} ^{-1}\left[\Gamma_{N,j}^{\left(z\right)}(V_{N})\right]\right\Vert _{\infty}\leq\frac{\left\Vert \Gamma_{N,j}^{\left(z\right)}(V_{N})\right\Vert _{\infty}}{1-\beta},
\]
where,
\begin{eqnarray*}
\left\Vert \Gamma_{N,j}^{\left(z\right)}(V_{N})\right\Vert _{\infty} & \leq\frac{1}{N} & \sum_{i=1}^{N}\left\Vert \frac{\partial u_{\psi}\left(U_{i};\cdot\right)}{\partial z_{j}}\right\Vert _{\infty}\left\Vert w_{\psi}\left(U_{i};\cdot\right)\right\Vert _{\infty}\\
 &  & +\frac{1}{N}\sum_{i=1}^{N}\left\{ \left\Vert u_{\psi}\left(U_{i};\cdot\right)\right\Vert _{\infty}+\beta\left\Vert V_{N}\right\Vert _{\infty}\right\} \left\Vert \frac{\partial w_{\psi}\left(U_{i};\cdot\right)}{\partial z_{j}}\right\Vert _{\infty}.
\end{eqnarray*}
We know $\left\Vert V_{N}\right\Vert _{\infty}\rightarrow^{P}\left\Vert V\right\Vert _{\infty}$
and, under Assumption \ref{Ass: smoothness}, we can appeal to the
ULLN to obtain 
\begin{eqnarray*}
\frac{1}{N}\sum_{i=1}^{N}\left\Vert \frac{\partial u_{\psi}\left(U_{i};\cdot\right)}{\partial z_{j}}\right\Vert _{\infty}\left\Vert w_{\psi}\left(U_{i};\cdot\right)\right\Vert _{\infty} & \rightarrow^{P} & E\left[\left\Vert \frac{\partial u_{\psi}\left(U;\cdot\right)}{\partial z_{j}}\right\Vert _{\infty}\left\Vert w_{\psi}\left(U;\cdot\right)\right\Vert _{\infty}\right],\\
\frac{1}{N}\sum_{i=1}^{N}\left\Vert u_{\psi}\left(U_{i};\cdot\right)\right\Vert _{\infty}\left\Vert \frac{\partial w_{\psi}\left(U_{i};\cdot\right)}{\partial z_{j}}\right\Vert _{\infty} & \rightarrow^{P} & E\left[\left\Vert u_{\psi}\left(U;\cdot\right)\right\Vert _{\infty}\left\Vert \frac{\partial w_{\psi}\left(U;\cdot\right)}{\partial z_{j}}\right\Vert _{\infty}\right].
\end{eqnarray*}
We conclude that $\left\Vert \Gamma_{N,j}^{\left(z\right)}(V_{N})\right\Vert _{\infty}$
and therefore also $\left\Vert \partial V_{N}/\left(\partial z_{j}\right)\right\Vert _{\infty}$
are bounded w.p.a.1. Similarly, it follows that $\left\Vert \partial V_{N}/\left(\partial\lambda\right)\right\Vert _{\infty}$
is bounded w.p.a.1. and so $V_{N}\in\boldsymbol{\mathbb{C}}_{r}^{1}\left(\mathcal{Z}\times\left(0,\bar{\lambda}\right)\right)$
w.p.a.1 for some fixed $r<\infty$. 

Finally, observe that $dV\mapsto\varPsi_{N,j}\left(dV\right)\left(z,\lambda\right)=\nabla\Gamma_{N}(V_{N})\left[dV\right]\left(z,\lambda\right)+\Gamma_{N,j}^{\left(z\right)}(V_{N})\left(z,\lambda\right)$
is a contraction mapping on $\boldsymbol{\mathbb{B}}\left(\mathcal{Z}\times\left[0,\bar{\lambda}\right]\right)^{D}$
and so we can apply Theorem \ref{thm: General rate}. First, we expand
each of the two terms w.r.t. $V_{N}$,
\begin{eqnarray*}
 &  & \nabla\Gamma_{N}(V_{N})\left[dV\right](z;\lambda)-\nabla\Gamma_{N}(V_{0})\left[dV\right](z;\lambda)\\
 & = & \beta\sum_{i=1}^{N}\sum_{d_{1},d_{2}\in\mathcal{D}}\ddot{G}_{\lambda,d_{1},d_{2}}^{\left(r\right)}\left(u_{\psi}\left(U_{i};z\right)+\beta\bar{V}_{\psi,N}\left(U_{i};z,\lambda\right)\right)\left\{ V_{\psi,N}\left(U_{i};z,\lambda,d_{2}\right)-V_{\psi,0}\left(U_{i};z,\lambda,d\right)\right\} \\
 & \times & dV\left(U_{i};z,\lambda,d_{1}\right)w_{\psi}\left(U_{i};z\right),
\end{eqnarray*}
where $\ddot{G}_{\lambda,d_{1},d_{2}}^{\left(r\right)}(r)=\frac{\partial^{2}G_{\lambda}(r)}{\partial r\left(d_{1}\right)\partial r\left(d_{1}\right)}$.
It is easily checked that $\left|\ddot{G}_{\lambda,d_{1},d_{2}}^{\left(r\right)}(r)\right|\leq C/\lambda$
for some $C<\infty$ and so the right hand side in the above equation
is bounded by $C/\lambda\left\Vert V_{N}-V_{0}\right\Vert _{\infty}\left\Vert dV\right\Vert _{\infty}=O_{P}\left(\sqrt{N}/\lambda\right)$
for any given $dV\in\boldsymbol{\mathbb{B}}\left(\mathcal{Z}\times\left[0,\bar{\lambda}\right]\right)^{D}$.
By similar arguments, we can show that $\left\Vert \dot{\Gamma}_{N,j}^{\left(z\right)}(V_{N})-\dot{\Gamma}_{N,j}^{\left(z\right)}(V_{0})\right\Vert _{\infty}=O_{P}\left(\sqrt{N}/\lambda\right)$
and $\left\Vert \dot{\Gamma}_{N,j}^{\left(\lambda\right)}(V_{N})-\dot{\Gamma}_{N,j}^{\left(\lambda\right)}(V_{0})\right\Vert _{\infty}=O_{P}\left(\sqrt{N}/\lambda\right)$.
Theorem \ref{thm: General rate} now yields the second part of the
theorem.
\end{proof}
\begin{proof}[Proof of Theorem \ref{Thm: Rust normal}]
 We verify the conditions in part (ii) of Theorem \ref{thm: Master}
with $\partial\mathcal{V}=\boldsymbol{\mathbb{C}}_{r}^{1}\left(\mathcal{Z}\right)^{D}$
and $r<\infty$ given in Theorem \ref{Thm: Rust approx rate}. First,
by arguments similar to the ones in the analysis of $\mathcal{G}$
in the proof of Theorem \ref{Thm: Rust approx rate}, $\mathcal{G}^{\prime}$
is Glivenko-Cantelli due to the Lipschitz property of $V_{0}\left(\psi_{Z}\left(U;z\right),\lambda\right)$
and the other components entering the function set under Assumption
\ref{Ass: smoothness}. Second, $\boldsymbol{\mathbb{C}}_{r}^{1}\left(\mathcal{Z}\right)^{D}$
has finite Bracketing number according to Theorem 2.7.1 in \citet{VW}
and so is also Glivenko-Cantelli.

 The rate result is an immediate consequence of Theorem \ref{thm: General rate}
together with Assumption \ref{ass: projection 1}. For the weak convergence
result, we use the decomposition (\ref{eq: V-hat decomp}) where $\left\Vert \hat{V}_{N}-V_{N}\right\Vert _{\infty}=O_{p}(\rho_{\Pi,K})=o_{P}\left(1/\sqrt{N}\right)$
while the second term converges weakly according to Theorem \ref{Thm: Rust normal}.
\end{proof}

\section{\label{sec:Sieve-spaces}Additional numerical details for sieve method}

\subsection*{Chebyshev basis functions}

Chebyshev polynomials of the first kind have well-known good properties
for approximating functions on bounded intervals. Recall that Chebyshev
polynomials are defined on $\left[-1,1\right]$. We then choose $-\infty<z^{\min}<z^{\max}<\infty$
and define the $k$th basis function as follows for any $z\in\mathbb{R}$:

\[
B_{c,k}(z)=\begin{cases}
\cos\left((k-1)\arccos(T(z))\right), & |T(z)|\leq1\\
\text{({sign}}(T(z)))^{k}, & |T(z)|>1
\end{cases},
\]
where $T(z)=2\frac{z-z^{\min}}{z^{\max}-z^{\min}}-1$ maps $z$ into
the interval $[-1,1]$. In particular, the basis functions are ``truncated''
and are set to one outside the interval $\left[z^{\min},z^{\max}\right]$.
This is done to avoid any erratic extrapolation. We then choose the
grid points $z_{1},...,z_{M}$ in (\ref{eq: Pi least-squares}) as
the Chebyshev nodes in order to minimize the presence of Runge's phenomenon.
Thus, $M=K$ in this case.

\subsection*{B-Splines}

We use cardinal B(asis)-splines to form our B-spline spaces, so they
are represented by a knot vector with equidistant entries $(0,\frac{1}{M+1},\frac{2}{M+1},\ldots,\frac{M}{M+1},1)$,
and the Cox-de Boor recursion

\begin{align*}
\bar{B}_{i,0}(z)= & \begin{cases}
1 & \text{{\,if\,}}t_{i}\leq z<t_{i+1}\\
0 & \text{{otherwise}}
\end{cases}\\
\bar{B}_{i,k}(z)= & \frac{z-t_{i}}{t_{i+k}-t_{i}}\bar{B}_{i,k-1}(z)+\frac{t_{i+k+1}-z}{t_{i+k+1}-t_{i+1}}\bar{B}_{i+1,k-1}(z).
\end{align*}
For interpolation purposes we use the so-called Universal (Parameters)
Method by \citet{tjahjowidodo2017direct}. This amounts to choosing
the $M$ grid points as the unique maximizers of all B-splines of
degree $k\geq1$, or any point if $k=0$ in which case we set it to
the first $K$ elements of the knot vector. The above are defined
on the unit interval $[0,1]$ and so the final basis functions are
chosen as
\[
B_{c,k}(z)=\begin{cases}
\bar{B}_{k}(T\left(z\right)) & 0\leq T(z)\leq1\\
\text{({sign}}(T(z))) & otherwise
\end{cases}.
\]
where now $T(z)=\frac{z-z^{\min}}{z^{\max}-z^{\min}}$.

\begin{table}
\caption{Coefficients on tensor product Chebyshev basis functions in the 2D
model of engine replacement for $K=J^{2}=25$, $N=200$.\label{tab:Coefficients-on-tensor}}

\centering
\begin{tabular}{lrrrrr}
\toprule
$J_1\textbackslash J_2$& 1&2&3&4&5\\
\midrule
1&-38.4713&-4.4754&1.6176&-0.256420&-0.064960\\
2&-4.4754&1.9662e-14&-6.2341e-15&-1.1318e-15&2.5392e-15\\
3&1.6176&7.2256e-15&5.6179e-14&-2.0548e-14&-3.3049e-15\\
4&-0.2564&-1.0110e-14&-8.8673e-15&7.5672e-15&-2.9838e-15\\
5&-0.0649&4.0869e-15&-1.5251e-14&3.4571e-15&1.4218e-15\\
\bottomrule
\end{tabular}
\end{table}

\begin{table}
\caption{Coefficients on tensor product 2nd order B-Spline basis functions
in the 2D model of engine replacement for $K=J^{2}=25$, $N=200$.\label{tab:Coefficients-on-tensor-bspline}}

\centering
\begin{tabular}{lrrrrr}
\toprule
$J_1\textbackslash J_2$& 1&2&3&4&5\\
\midrule
1&-21.9800 & -27.1194 & -30.0583 & -31.1025 & -31.506\\
2&-27.1194 & -32.2589 & -35.1977 & -36.2419 & -36.6455\\
3&-30.0583 & -35.1977 & -38.1366 & -39.1808 & -39.5844\\
4&-31.1025 & -36.2419 & -39.1808 & -40.2250 & -40.6285\\
5&-31.5060 & -36.6455 & -39.5844 & -40.6285 & -41.0321\\
\bottomrule
\end{tabular}
\end{table}

\begin{table}
\caption{Coefficients on the ten basis functions, and their products, upon
convergence with $K=50^{2}.$}

\begin{tabular}{lrrrrrrrrrr}
\toprule
$J_1\textbackslash J_2$&1&2&3&4&5&6&7&8&9&10\\
\midrule 1 &-52.200 & -5.070 & 2.700 & -1.170 & 0.478 & -0.145 & -0.006 & 0.038 & -0.024 & 0.007\\
2 &-5.070 & -1.560 & 0.135 & 0.225 & -0.057 & -0.023 & 0.013 & -0.006 & 0.004 & 0.001\\
3 &2.700 & 0.135 & -0.176 & 0.032 & 0.046 & -0.021 & -0.004 & 0.005 & -0.002 & 0.001\\
4 &-1.170 & 0.225 & 0.032 & -0.087 & 0.019 & 0.019 & -0.012 & 0.002 & 0.001 & -0.001\\
5 &0.478 & -0.057 & 0.046 & 0.019 & -0.038 & 0.010 & 0.009 & -0.008 & 0.002 & 0.000\\
6 &-0.145 & -0.023 & -0.021 & 0.019 & 0.010 & -0.018 & 0.005 & 0.005 & -0.004 & 0.002\\
7 &-0.006 & 0.013 & -0.004 & -0.012 & 0.009 & 0.005 & -0.009 & 0.003 & 0.002 & -0.003\\
8 &0.038 & -0.006 & 0.005 & 0.002 & -0.008 & 0.005 & 0.003 & -0.005 & 0.002 & 0.001\\
9 &-0.024 & 0.004 & -0.002 & 0.001 & 0.002 & -0.004 & 0.002 & 0.002 & -0.003 & 0.001\\
10 &0.007 & 0.001 & 0.001 & -0.001 & 0.0002 & 0.002 & -0.003 & 0.001 & 0.001 & -0.002\\ \bottomrule \end{tabular}

\label{interactiontable}
\end{table}

\end{document}